\documentclass[%
 reprint,
 superscriptaddress,
 aps,
 prx
]{revtex4-2}

\usepackage{graphicx}
\usepackage{dcolumn}
\usepackage{bm}

\usepackage{braket}
\usepackage[shortlabels]{enumitem}
\usepackage{amsmath}
\usepackage{amssymb}
\usepackage{amsthm}

\newtheorem{theorem}{Theorem}
\newtheorem{corollary-theorem}[theorem]{Corollary}

\newtheorem{corollary-proposition}[theorem]{Corollary}
\newtheorem{lemma}[theorem]{Lemma}
\newtheorem{corollary-lemma}[theorem]{Corollary}

\begin{document}


\title{Qudit-based quantum error-correcting codes from\\irreducible representations of $\mathrm{SU}(d)$}

\author{Robert Frederik Uy}
 \email{rfu20@cam.ac.uk}
\affiliation{%
 Peterhouse, University of Cambridge, Cambridge CB2 1RD, United Kingdom
}%
\affiliation{
 Cavendish Laboratory, University of Cambridge, Cambridge CB3 0HE, United Kingdom
}%

\author{Dorian A. Gangloff}
 \email{dag50@cam.ac.uk}
\affiliation{
 Cavendish Laboratory, University of Cambridge, Cambridge CB3 0HE, United Kingdom
}%

\date{\today}

\begin{abstract}
Qudits naturally correspond to multi-level quantum systems, which offer an efficient route towards quantum information processing, but their reliability is contingent upon quantum error correction capabilities. In this paper, we present a general procedure for constructing error-correcting qudit codes through the irreducible representations of $\mathrm{SU}(d)$ for any odd integer $d \geq 3$. Using the Weyl character formula and inner product of characters, we deduce the relevant branching rules, through which we identify the physical Hilbert spaces that contain valid code spaces. We then discuss how two forms of permutation invariance and the Heisenberg-Weyl symmetry of $\mathfrak{su}(d)$ can be exploited to simplify the construction of error-correcting codes. Finally, we use our procedure to construct an infinite class of error-correcting codes encoding a logical qudit into $(d-1)^2$ physical qudits.
\end{abstract}

\maketitle


\section{Introduction}
Despite our best efforts to suppress quantum noise, it continues to stymie reliable large-scale quantum computation \cite{Preskill2018qcnisq, RevModPhys.95.045005}. Quantum states are inextricably coupled to their environment, and this inevitably leads to the corruption of fragile quantum information \cite{PRXQuantum.4.020307, Hu2018}. Good quantum error-correcting codes are thus vital to the realization of the full potential of quantum computers \cite{PhysRevA.108.062403,Bravyi2024}.

Quantum error correction is typically performed by encoding quantum information across a sufficiently large number of physical states \cite{CAI202150}. Although most quantum error-correcting codes encode logical qubits into multiple physical qubits, qudit-based computation offers a myriad of advantages. One theoretical benefit is that the number of computational units needed to obtain a logical Hilbert space of the same dimension is reduced by a factor of $\log_2 d$ \cite{PRXQuantum.4.030327}, where $d$ is the dimension of one qudit. Moreover, from an experimental perspective, most physical systems are not two-dimensional \cite{Nikolaeva2024}. Therefore, utilizing them as qudits is easier and eliminates a source of error since it obviates the need to isolate the relevant two-dimensional system from the other states in the physical Hilbert space \cite{PhysRevResearch.6.013050}.

While error-correcting codes for qudits have been constructed before \cite{PhysRevX.6.031006,OUYANG201743,PhysRevA.101.042305,PhysRevA.105.042427,dutta2024}, representation theory may offer new insight into this. In \cite{Gross2021}, Gross explored how to construct error-correcting codes through the irreducible representations of $\mathrm{SU}(2)$. These are suitable for encoding qubits into single-spin systems since they possess $\mathfrak{su}(2)$ symmetry. However, the concept behind his procedure could be generalized to other physical systems with a Lie algebra of Hamiltonians \cite{Gross2021}, including $\mathfrak{su}(d)$-symmetric qudits. The idea is to choose our code space within a vector space on which some irrep of a Lie group $G$, restricted to a finite subgroup $K \leq G$, is defined. By construction, we can implement any operation $\mathsf{O} \in K$ on the logical qudit, and the code space is endowed with a Lie algebra $\mathfrak{g}$ of physically relevant error operators.

Recently, Herbert et al. worked on the generalization to $\mathrm{SU}(3)$, yielding a method for constructing error-correcting codes for qutrits \cite{Herbert2023}. The natural next step would be to consider the general case of $\mathrm{SU}(d)$ and, through this, formulate a recipe for constructing error-correcting codes that encode a logical qudit of higher dimension into multiple physical qudits (Figure \ref{fig:logical qudit}). This is by no means a simple task, not least because it requires us to deal with arbitrarily large code spaces and arbitrarily large sets of error correction conditions. However, as we will soon show, symmetry can considerably simplify the problem.

\begin{figure}[b]
	\centering
	\includegraphics[width=\columnwidth]{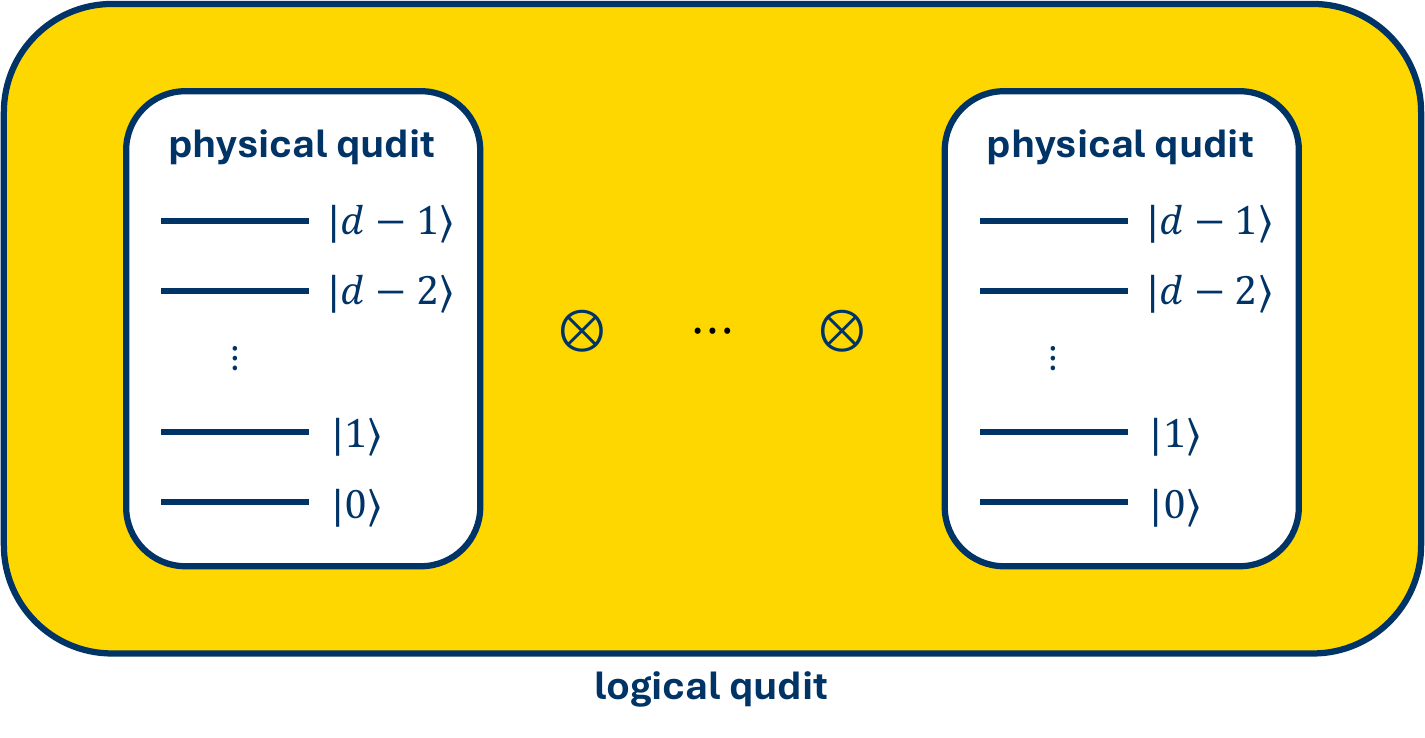}
	\caption{\label{fig:logical qudit} A logical qudit encoded into multiple physical qudits, each of dimension $d$.}
\end{figure}

In this paper, we present a general procedure for constructing quantum error-correcting codes for qudits of odd dimension $d \geq 3$ through the irreducible representations of $\mathrm{SU}(d)$. To allow error correction, we restrict the $\mathrm{SU}(d)$ irreps to $\mathrm{HW}(d)$, the Heisenberg-Weyl group, which was chosen since its image under certain irreps contains the qudit generalizations of the qubit Pauli operators \cite{Gottesman1999,PhysRevA.86.022308,Sarkar2024}. We then determine the branching rules for $\mathrm{SU}(d) \downarrow \mathrm{HW}(d)$, from which we can deduce which irreps are defined on symmetric vector spaces that contain valid code spaces. This is followed by the use of the Heisenberg-Weyl symmetry of $\mathfrak{su}(d)$ to show that some error-correction conditions are equivalent to each other and thereby obtain a reduced set of sufficient conditions for a code to be error-correcting. We then show that it is convenient to consider so-called \textit{sparse} and \textit{doubly permutation-invariant} codes and use these simplifications to construct a class of error-correcting codes encoding a logical qudit, for any odd integer $d \geq 5$, into $(d-1)^2$ physical qudits.

\section{Irreps and code spaces}
A \textit{code} $\mathcal{C}$ is defined as an embedding of a logical Hilbert space $\mathcal{H}_L$ into a physical Hilbert space $\mathcal{H}_P$ \cite{PhysRevX.10.041018,gottesman2024, Liu2023,10.21468/SciPostPhys.12.5.157,PhysRevA.107.032411}. The image of $\mathcal{C}$ is called the \textit{code space}, whose elements are referred to as \textit{code words} \cite{PhysRevX.10.041018, gottesman2024}. Since we are dealing with qudits, the logical Hilbert space is $\mathcal{H}_L = \mathbf{C}^d$, for which we use the orthonormal basis $\{\ket{k}_L \mid k \in \mathbf{Z}_d\}$. We then denote by $\ket{\overline{k}}$ the image of $\ket{k}_L$ under $\mathcal{C}$, i.e. the code word corresponding to the logical qudit state $\ket{k}_L$.


\subsection{Irreps of $\mathrm{SU}(d)$ on symmetric spaces}
We choose $\mathcal{H}_P$ to be a vector space on which an irrep of $\mathrm{SU}(d)$ is defined, thereby ensuring that the code space is preserved by any $\mathrm{SU}(d)$ operation (i.e. any logical single-qudit gate, up to a global phase factor). Out of all these vector spaces, it would be far easier mathematically to consider symmetric vector spaces; constructing a basis for such spaces is simple and the action of the error operators in $\mathfrak{su}(d)$ obeys a straightforward and easily generalizable pattern. Physically, this is also a natural choice. By the symmetrization postulate for indistinguishable particles, the states of a system of bosons are necessarily permutation-invariant (or totally symmetric) \cite{zwiebach2022mastering,Ouyang2024newtakepermutation,PhysRevLett.85.194,PhysRevResearch.5.033018}. Moreover, the ground subspace of any Heisenberg ferromagnet's Hilbert space contains permutation-invariant states \cite{PhysRevA.90.062317,Aydin2024familyof}, and this model of quantum magnetism is naturally realized in numerous physical systems \cite{PhysRevB.103.144417}. We therefore aim to construct codes that map to a symmetric physical Hilbert space.

In line with this, we give a brief overview of some relevant results about the irreps of $\mathrm{SU}(d)$, particularly those defined on symmetric vector spaces. For further information, readers may consult Refs. \cite{hall2015lie,Georgi:1999wka,fulton2013representation}.

We begin by noting that there is a one-to-one correspondence between the irreps of $\mathrm{SU}(d)$, $\mathfrak{su}(d)$, and $\mathfrak{sl}(d,\mathbf{C}) \cong \mathfrak{su}(d)_\mathbf{C}$ \cite{hall2015lie}. Thus, in a not-so-egregious abuse of notation, we will denote these maps by the same symbol. More importantly, this means that we can study the irreps of $\mathrm{SU}(d)$ by applying the general theory of semisimple Lie algebras to $\mathfrak{sl}(d,\mathbf{C})$.

By the theorem of the highest weight \cite{hall2015lie}, each irrep of $\mathfrak{sl}(d,\mathbf{C})$ is uniquely determined by its highest weight $\boldsymbol{\omega}$, which in turn can be identified with a multiplet $(\omega_0,\cdots,\omega_{d-2}) \in \mathbf{N}_0^{d-1}$. Moreover, it can be shown that the irrep with highest weight $\boldsymbol{\omega} = (N,0,\cdots,0)$, for any $N \in \mathbf{N}$, is on a symmetric vector space $\mathrm{Sym}^N(\mathbf{C}^d)$. We shall denote this irrep by $\boldsymbol{\pi}_{d;N}$. For the rest of this paper, we will only be considering irreps of this type.

Having chosen a set of candidates for $\mathcal{H}_P$, let's now construct a basis for these vector spaces. Fix $N \in \mathbf{N}$. For each $\mathbf{u} \in \mathbf{N}_0^d$ with $\sum_{\xi \in \mathbf{Z}_d} u_\xi = N$, define the multi-set
\begin{equation}
	\mathcal{M}_\mathbf{u} = \big\{ \underbrace{0,\cdots,0}_{u_0}, \cdots, \underbrace{d-1,\cdots,d-1}_{u_{d-1}} \big\}.
\end{equation}
and thereafter the \textit{permutation-invariant vector}
\begin{equation}
	\ket{S_\mathbf{u}} = \sum_{i_1,\cdots,i_N \  \mathrm{s.t.} \  \{i_1,\cdots,i_N\} = \mathcal{M}_\mathbf{u}} \ket{i_1} \otimes \cdots \otimes \ket{i_N}.
\end{equation}
From this, we can construct the set
\begin{equation}
    \left\{ \ket{S_\mathbf{u}} \  \bigg| \  \mathbf{u}\in\mathbf{N}_0^d, \sum_{\xi \in \mathbf{Z}_d} u_\xi = N \right\},
\end{equation}
which is a basis for $\mathrm{Sym}^N(\mathbf{C}^d)$.

This is, however, not the end of the story. Thanks to the way tensor product representations are defined, codes whose range is $\mathrm{Sym}^N(\mathbf{C}^d)$ can transversally implement any logical operation $\mathsf{O} \in \mathrm{SU}(d)$. Unfortunately, the Eastin-Knill theorem \cite{PhysRevLett.102.110502,PhysRevLett.131.240601,PhysRevLett.133.030602} precludes such codes from being error-correcting. In order to obtain a proper subgroup of $\mathrm{Sym}^N(\mathbf{C}^d)$ that remains invariant under some set of logical operations, we must restrict the irreps to a finite subgroup of logical operations that we want to implement --- the Heisenberg-Weyl group.

\subsection{Heisenberg-Weyl group}
Let $d \geq 3$ be an odd integer, $\mathbf{Z}_d$ be the ring of integers modulo $d$, and $\mathbf{Z}_d^\times$ be the multiplicative group of integers modulo $d$. For any vector space $V$, let $\mathrm{GL}(V)$ be the group of invertible linear maps $V \longrightarrow V$.

We define the \textit{Heisenberg-Weyl group}, otherwise known as the qudit Pauli group, to be the set
\begin{equation}
    \mathrm{HW}(d) = \left\{
    \begin{pmatrix}
        1 & a & c\\
        0 & 1 & b\\
        0 & 0 & 1
    \end{pmatrix}
    \bigg| \  a,b,c \in \mathbf{Z}_d
    \right\}
\end{equation}
with the operation matrix multiplication and identity element $\mathsf{diag}(1,1,1)$. Since we want to be able to implement any operation $\mathsf{O} \in \mathrm{HW}(d)$ on the logical qudit, we are interested in the faithful $d$-dimensional irreps of $\mathrm{HW}(d)$, which we now construct. For each $z \in \mathbf{Z}_d^\times$, define $\boldsymbol{\rho}_z \colon \mathrm{HW}(d) \longrightarrow \mathrm{GL}(\mathbf{C}^d)$ as the group homomorphism that maps
\begin{equation}
    \label{HW(d) faithful irreps}
    \begin{pmatrix}
        1 & a & c\\
        0 & 1 & b\\
        0 & 0 & 1
    \end{pmatrix}
    \mapsto
    \zeta_d^{cz} \mathsf{X}_d^a \mathsf{Z}_d^{bz},
\end{equation}
where $\zeta_d = \exp(2\pi i/d)$ and
\begin{align}
    \mathsf{X}_d &= \sum_{j \in \mathbf{Z}_d} \ket{j \oplus 1}\bra{j},\\
    \mathsf{Z}_d &= \sum_{k \in \mathbf{Z}_d} \zeta_d^k \ket{k}\bra{k}.
\end{align}
These are all of the faithful $d$-dimensional irreps of $\mathrm{HW}(d)$ \cite{Grassberger2001,Schulte2004}. Above, the symbol $\oplus$ is used to denote addition modulo $d$. It will be useful to know that the character of $\boldsymbol{\rho}_z$ is the map $\chi_{\boldsymbol{\rho}_z^{}}^{} \colon \mathrm{HW}(d) \longrightarrow \mathbf{C}$ given by
\begin{equation}
    \label{HW(d) characters}
    \begin{pmatrix}
        1 & a & c\\
        0 & 1 & b\\
        0 & 0 & 1
    \end{pmatrix}
    \mapsto
    d\zeta_d^{cz} \delta_{a0}\delta_{b0}.
\end{equation}

Before proceeding, we note that it is necessary to restrict ourselves to odd $d$ because, for any $z \in \mathbf{Z}_d^\times$, the group $\boldsymbol{\rho}_z(\mathrm{HW}(d))$ is not a subgroup of $\mathrm{SU}(d)$. A quick check shows that, whenever $d$ is even, the determinant of $\mathsf{Z}_d$, which is an element of $\boldsymbol{\rho}_z(\mathrm{HW}(d))$ for any $z \in \mathbf{Z}_d^\times$, is $-1$.

\subsection{Branching rules and valid code spaces}
Now that we have identified a set of candidate physical Hilbert spaces from $\mathrm{SU}(d)$ irreps as well as a finite subgroup of logical operations, the next step is to determine which of these candidate physical Hilbert spaces contain valid code spaces.

In our analysis, we consider the embedding $\mathrm{HW}(d) \hookrightarrow \mathrm{SU}(d)$ given by the map $\boldsymbol{\rho}_1$. A \textit{valid code space} is thus defined as a direct sum of subspaces $W \subset \mathrm{Sym}^N(\mathbf{C}^d)$ on which $\boldsymbol{\rho}_1$ may be defined. But why is this the appropriate definition? First things first, we define our code $\mathcal{C}$ in such a way that for all $k \in \mathbf{Z}_d$,
\begin{align}
	\mathsf{\overline{Z}}_d \ket{\overline{0}} &= \ket{\overline{0}},\\
	\mathsf{\overline{X}}_d \ket{\overline{k}} &= \ket{\overline{k \oplus 1}} \label{X action on code words}
\end{align}
where the bars over operators indicate that we are taking the image under the relevant $\mathrm{SU}(d)$ irrep. Now, since $\boldsymbol{\rho}_1$ is the fundamental irrep of $\mathrm{HW}(d)$, it follows that the action of $\mathsf{\overline{X}}_d$ and $\mathsf{\overline{Z}}_d$ on valid code spaces, under our definition, is given by
\begin{align}
	\mathsf{\overline{X}}_d &\colon \ket{\overline{k}} \mapsto \ket{\overline{k\oplus1}},\\
	\qquad \mathsf{\overline{Z}}_d &\colon \ket{\overline{k}} \mapsto \zeta_d^k \ket{\overline{k}},
\end{align}
for all $k \in \mathbf{Z}_d$. Hence, our definition ensures that the corresponding operations $\mathsf{X}_L$ and $\mathsf{Z}_L$ on a logical qudit can be implemented in the way they are typically defined.

A remark is in order. It is worth noting that, although our work primarily adheres to this convention, other embeddings and definitions are technically just as valid. After all, there is nothing fundamental about the phase factor that $\mathsf{Z}_L$ applies to each $\ket{k}_L \in \mathcal{H}_L$. We could, for instance, have defined its action on $\mathcal{H}_L$ to be $\ket{k}_L \mapsto \zeta_d^{-k}\ket{k}_L$ or even $\ket{k}_L \mapsto \zeta_d^{2k}\ket{k}_L$. In such cases, we then define a valid code space as a direct sum of subspaces on which $\boldsymbol{\rho}_{d-1}$ and $\boldsymbol{\rho}_2$ may be defined. We discuss this in greater detail towards the end; for now, we will anchor our analysis in standard definitions.

We can determine the values of $N$ for which the vector space $\mathrm{Sym}^N(\mathbf{C}^d)$ contains a valid code space by studying how $\boldsymbol{\pi}_{d;N}^{}\big|_{\mathrm{HW}(d)}^{}$ (which is in general a reducible representation of $\mathrm{HW}(d)$) decomposes into a direct sum of $\mathrm{HW}(d)$ irreps. The rules for such a decomposition are called \textit{branching rules}. One way to deduce a branching rule is to use the inner product between characters. In particular, the multiplicity of the irrep $\boldsymbol{\rho}_\eta$ (for any $\eta \in \mathbf{Z}_d^\times$) in the decomposition of $\boldsymbol{\pi}_{d;N}^{}\big|_{\mathrm{HW}(d)}^{}$ into a direct sum of irreps is given by \cite{fulton2013representation, FALLBACHER2015229}
\begin{align}
    \label{multiplicity in branching rule}
    \bigg\langle \chi_{\boldsymbol{\rho}_\eta^{}}^{}, \chi&_{\boldsymbol{\pi}_{d;N}^{}\big|_{\mathrm{HW}(d)}^{}}^{} \bigg\rangle \nonumber\\
    &= \sum_{g \in \mathrm{HW}(d)} \chi_{\boldsymbol{\rho}_\eta}^*(g) \chi_{\boldsymbol{\pi}_{d;N}^{}\big|_{\mathrm{HW}(d)}^{}}^{}(g).
\end{align}

We've already calculated $\chi_{\boldsymbol{\rho}_\eta}^{}$ previously in Eq. (\ref{HW(d) characters}), so it remains to evaluate $\chi_{\boldsymbol{\pi}_{d;N}^{}\big|_{\mathrm{HW}(d)}^{}}^{}$. In fact, it suffices to determine the values that $\chi_{\boldsymbol{\pi}_{d;N}^{}\big|_{\mathrm{HW}(d)}^{}}^{}$ take in the centre $Z(\mathrm{HW}(d))$ because $\chi_{\boldsymbol{\rho}_\eta}^{}$ vanishes for all elements $g \in \mathrm{HW}(d) \backslash Z(\mathrm{HW}(d))$. We do this in the following lemma:

\begin{lemma}
    \label{character Heis centre}
    For any $N \in \mathbf{N}$, the character of the irrep $\boldsymbol{\pi}_{d;N}\big|_{\mathrm{HW}(d)}$ takes on the following values in the centre $Z(\mathrm{HW}(d))$:
    \begin{equation}
        \chi_{\boldsymbol{\pi}_{d;N}^{}}^{} \left(\zeta_d^\ell \, \mathsf{I}\right) = \binom{N+d-1}{N} \, \zeta_d^{\ell N}
    \end{equation}
    for all $\ell \in \mathbf{Z}_d$.
\end{lemma}

\begin{proof}
    We use the Weyl character formula expressed in terms of complete symmetric polynomials \cite{fulton2013representation}. In particular, for the special case of $\boldsymbol{\pi}_{d;N}^{}$, we have for all $\mathsf{M} \in \mathrm{SU}(d)$,
    \begin{equation}
	\chi_{\boldsymbol{\pi}_{d;N}^{}}^{}(\mathsf{M}) = \sum_{1 \leq i_1 \leq \cdots \leq i_N \leq d} \lambda_{i_1} \lambda_{i_2} \cdots \lambda_{i_N}
    \end{equation}
    where $\lambda_1,\cdots,\lambda_d$ are the eigenvalues of $\mathsf{M}$.
	
   Recall that, for any $\ell \in \mathbf{Z}_d$, the eigenvalues of $\zeta_d^\ell \, \mathsf{I}$ are all $\zeta_d^\ell$. By a simple counting argument, it is easy to see that the number of increasing functions $\{1,\cdots,N\} \longrightarrow \{1,\cdots,d\}$ is $\displaystyle\binom{N+d-1}{N}$. The result follows immediately from these observations.
\end{proof}

With the above lemma established, we are now ready to state and prove the branching rules:
\begin{theorem}
    \label{branching rules}
    For any $N \in \mathbf{N}$ with $(N,d) = 1$, the branching rules for $\mathrm{SU}(d) \downarrow \mathrm{HW}(d)$ are
    \begin{equation}
	\boldsymbol{\pi}_{d;N}^{}\big|_{\mathrm{HW}(d)} = \boldsymbol{\rho}_\eta^{\oplus \frac{1}{d}\dim(\mathrm{Sym}^N(\mathbf{C}^d))},
    \end{equation}
    where $\eta \in \mathbf{Z}_d^\times$ is such that $\eta \equiv N \mod d$. (In this particular case, $\boldsymbol{\oplus}$ is used to denote a direct sum.)
\end{theorem}

\begin{proof}
    First note that
    \begin{equation}
	\dim(\mathrm{Sym}^N(\mathbf{C}^d)) = \binom{N+d-1}{N}.
    \end{equation}
    We then use Eq. (\ref{multiplicity in branching rule}) to obtain the multiplicity of the $\mathrm{HW}(d)$ irreps in $\boldsymbol{\pi}_{d;N}\big|_{\mathrm{HW}(d)}$. Fix $\eta \in \mathbf{Z}_d^\times$. Observe that, for any $m \in \mathbf{N}$ and any $N \equiv \eta \mod d$,
    \begin{align}
	\bigg\langle \chi_{\boldsymbol{\rho}_m^{}}^{},& \ \chi_{\boldsymbol{\pi}_{d;N}^{}\big|_{\mathrm{HW}(d)}^{}}^{}\bigg\rangle\nonumber\\
        =& \frac{1}{|\mathrm{HW}(d)|}\sum_{\ell \in \mathbf{Z}_d} \chi_{\boldsymbol{\rho}_m^{}}^*(\zeta_d^\ell \, \mathsf{I}) \cdot \chi_{\boldsymbol{\pi}_{d;N}^{}\big|_{\mathrm{HW}(d)}^{}}^{}(\zeta_d^\ell \, \mathsf{I})\nonumber\\
	=& \frac{1}{d^3}\sum_{\ell \in \mathbf{Z}_d} d\zeta_d^{-\ell m} \cdot \dim(\mathrm{Sym}^N(\mathbf{C}^d)) \zeta_d^{\ell \eta}\nonumber\\
	=& \frac{1}{d}\dim(\mathrm{Sym}^N(\mathbf{C}^d)) \, \delta_{m\eta},
    \end{align}
    where we have used the fact that $\chi_{\boldsymbol{\rho}_m^{}}^{}(g) = 0$ for all $g \in \mathrm{HW}(d) \backslash Z(\mathrm{HW}(d))$. The result then follows directly.
\end{proof}

As a direct consequence of Theorem \ref{branching rules}, we have:
\begin{corollary-theorem}
    \label{valid code spaces}
    If $N \equiv 1 \mod d$, then the vector space $\mathrm{Sym}^N(\mathbf{C}^d)$ contains a valid code space.
\end{corollary-theorem}

\begin{proof}
    By Theorem \ref{branching rules}, $\boldsymbol{\rho}_1$ is in the decomposition of $\boldsymbol{\pi}_{d;N}^{}\big|_{\mathrm{HW}(d)}^{}$ if $N \equiv 1 \mod d$.
\end{proof}

In fact, more can be said beyond Corollary \ref{valid code spaces}. If $N \equiv 1 \mod d$, $\mathrm{Sym}^N(\mathbf{C}^d)$ does not only contain a valid code space --- it \textit{is} a valid code space. This follows from the fact that, if $N \equiv 1 \mod d$, the only $\mathrm{HW}(d)$ irrep in the decomposition of the representation $\boldsymbol{\pi}_{d;N}^{}\big|_{\mathrm{HW}(d)}$ is $\boldsymbol{\rho}_1$. Of course, as noted earlier, we still cannot have $\mathcal{C}(\mathbf{C}^d) = \mathrm{Sym}^N(\mathbf{C}^d)$ if we want $\mathcal{C}$ to be error-correcting. But as Theorem \ref{branching rules} shows, $\mathrm{Sym}^N(\mathbf{C}^d)$ has proper subsets that work as valid code spaces, which is what we aimed to achieve. In fact, we have considerable choice over these subsets since $\mathrm{Sym}^N(\mathbf{C}^d)$ as a whole is a valid code space.

\subsection{Code structure}
To reiterate, the logical Hilbert space in our case is $\mathcal{H}_L = \mathbf{C}^d$. The physical Hilbert space, on the other hand, is taken to be $\mathcal{H}_P = \mathrm{Sym}^N(\mathbf{C}^d)$ for some $N \equiv 1 \mod d$. For the remainder of this article, it is assumed that $N \equiv 1 \mod d$, unless otherwise stated.

In addition, we note that a code $\mathcal{C}$ is uniquely determined by its image on $\{\ket{0}_L\}$ since, by Eq. (\ref{X action on code words}), specifying $\ket{\overline{0}}$ fixes all the other code words. Thus, we will often be defining our codes by specifying the code word $\ket{\overline{0}}$.

\section{Reduced Knill-Laflamme conditions}
A code $\mathcal{C} \colon \mathcal{H}_L \longrightarrow \mathcal{H}_P$ is called a \textit{quantum error-correcting code} (QECC) that corrects a set $\mathcal{E}$ of errors if and only if for all $\ket{\overline{i}},\ket{\overline{j}} \in \mathcal{C}(\mathcal{H}_L)$ and all $\mathsf{E}_a,\mathsf{E}_b \in \mathcal{E}$,
\begin{equation}
    \label{Knill-Laflamme conditions}
    \bra{\overline{i}} \mathsf{E}_a^\dagger \mathsf{E}_b^{}  
    \ket{\overline{j}} = c_{a,b} \braket{\overline{i}|\overline{j}}
\end{equation}
for some set $\left\{c_{a,b}\right\}$ of constants that only depend on $\mathsf{E}_a,\mathsf{E}_b$ \cite{gottesman2024,MikeandIke2010}. This set of equations is also known as the \textit{Knill-Laflamme conditions}. Notably, it suffices to show that these conditions are satisfied on a basis of error operators \cite{gottesman2024,MikeandIke2010}.

With $\mathcal{H}_P$ being a vector space on which a representation of $\mathrm{SU}(d)$ is defined, it would be mathematically elegant to consider the Lie algebra $\mathfrak{su}(d)$ as our error operators --- more precisely, the image of $\mathfrak{su}(d)$ under the relevant representation. Remarkably, this mathematical assumption coincides with what is physically relevant! The elements of $\mathfrak{su}(d)$ correspond to dit flips, phase flips, or combinations of both. Of course, there's also the case of no error occurring, i.e. the error operator is $\mathsf{I}$.

In constructing a basis for $\mathcal{E}$, it would be useful to define the following matrices for all $(j,k) \in \mathbf{Z}_d^2, \ell \in \mathbf{Z}_d$:
\begin{align}
	\mathsf{S}^{(j,k)}_d &= \ket{j}\bra{k} + \ket{k}\bra{j},\\
	\mathsf{A}^{(j,k)}_d &= -i\ket{j}\bra{k} + i\ket{k}\bra{j},\\
	\mathsf{D}^{(\ell)}_d &= \ket{\ell}\bra{\ell} - \ket{\ell\oplus1}\bra{\ell\oplus1},
\end{align}
which are essentially the qudit generalizations of the Pauli error operators on qubits. Physically, $\mathsf{S}^{(j,k)}_d$ are dit flips, $\mathsf{D}^{(\ell)}_d$ are phase flips, and $\mathsf{A}^{(j,k)}_d$ are combinations of dit and phase flips. Having defined these matrices, let us now explicitly construct a basis for $\mathcal{E}$ \cite{pfeifer2003lie}:
\begin{align}
    \mathcal{B}_{\mathcal{E}} = \{\mathsf{I}, \mathsf{\tilde{S}}^{(j,k)}_d,& \mathsf{\tilde{A}}^{(j,k)}_d, \mathsf{\tilde{D}}^{(\ell)}_d \mid\nonumber\\
    &0 \leq j < k \leq d-1, 0 \leq \ell \leq d-2\},
\end{align}
where the tilde indicates that it is the image under the relevant irrep, i.e. $\tilde{\mathsf{G}} = \boldsymbol{\pi}_{d;N}(\mathsf{G})$ for every $\mathsf{G} \in \mathfrak{su}(d)$.

As was the case in Refs. \cite{Gross2021,PhysRevA.108.022424,Herbert2023}, we can leverage the symmetry of the Lie algebra with respect to the finite subgroup to reduce the number of error-correction conditions we need to check. The following lemma establishes the Heisenberg-Weyl symmetry of $\mathfrak{su}(d)$ elements:
\begin{lemma}
    \label{Heisenberg-Weyl symmetry}
    Observe that
    \begin{align}
	\mathsf{\overline{X}}_d^\dagger \tilde{\mathsf{S}}^{(j,k)}_d \mathsf{\overline{X}}_d &= \tilde{\mathsf{S}}^{(j\ominus1,k\ominus1)}_d,\\
	\mathsf{\overline{X}}_d^\dagger \tilde{\mathsf{A}}^{(j,k)}_d \mathsf{\overline{X}}_d &= \tilde{\mathsf{A}}^{(j\ominus1,k\ominus1)}_d,\\
	\mathsf{\overline{X}}_d^\dagger \tilde{\mathsf{D}}_d^{(\ell)} \mathsf{\overline{X}}_d &= \tilde{\mathsf{D}}_d^{(\ell\ominus1)},\\
	\mathsf{\overline{Z}}_d^\dagger \tilde{\mathsf{S}}^{(j,k)}_d \mathsf{\overline{Z}}_d &= \Re\left(\zeta_d^{k-j}\right)\tilde{\mathsf{S}}^{(j,k)}_d - \Im\left(\zeta_d^{k-j}\right)\tilde{\mathsf{A}}^{(j,k)}_d,
    \end{align}
    for all $0 \leq j < k \leq d-1, 0 \leq \ell \leq d-2$. Above, $\ominus$ denotes subtraction modulo $d$.
\end{lemma}

\begin{proof}
    These follow directly from definitions; explicit calculations can be found in the Supplemental Material \cite{supp}.
\end{proof}

Using the Heisenberg-Weyl symmetry of $\mathfrak{su}(d)$, we obtain a reduced set of conditions for a code to be error-correcting:
\begin{theorem}
    A code $\mathcal{C} \colon \mathbf{C}^d \longrightarrow \mathrm{Sym}^N(\mathbf{C}^d)$ is a QECC if it satisfies the following conditions:
    \begin{enumerate}[itemsep=2pt]
        \item[\rm (C1)] $\bra{\overline{i}} \mathsf{\tilde{S}}^{(0,n)}_d \ket{\overline{j}} = 0$ for all $i,j \in \mathbf{Z}_d$ with $i \neq j$ and $n \in \left\{1,\cdots,\frac{d-1}{2}\right\}$
        \item[\rm (C2)] $\bra{\overline{i}} \tilde{\Gamma}_d^{(p,q)} \tilde{\Upsilon}_d^{(r,s)} \ket{\overline{j}} = c_{(\tilde{\Gamma}_d^{(p,q)},\tilde{\Upsilon}_d^{(r,s)})}\delta_{ij}$ for all $i,j \in \mathbf{Z}_d$, $\Gamma,\Upsilon \in \{\mathsf{S},\mathsf{A}\}$, $0 \leq p < q \leq d-1$, and $0 \leq r < s \leq d-1$
        \item[\rm (C3)] $\bra{\overline{k}} \mathsf{\tilde{D}}_d^{(d-2)} \ket{\overline{k}} = c_{(\mathsf{I}, \mathsf{\tilde{D}}^{(d-2)}_d)}$ for all $k \in \mathbf{Z}_d$
        \item[\rm (C4)] $\bra{\overline{k}} \mathsf{\tilde{D}}^{(\ell)}_d \mathsf{\tilde{D}}^{(d-2)}_d \ket{\overline{k}} = c_{(\mathsf{\tilde{D}}^{(\ell)}_d, \mathsf{\tilde{D}}^{(d-2)}_d)}$ for all $k \in \mathbf{Z}_d$ and $\ell \in \{0, \cdots, d-2\}$
    \end{enumerate}
    for some set of constants $\left\{ c_{(\mathsf{E},\mathsf{F})} \mid \mathsf{E},\mathsf{F} \in \mathcal{E} \right\}$.
\end{theorem}

\begin{proof}
    See Supplemental Material \cite{supp}.
\end{proof}

\section{Sparse and doubly permutation-invariant codes}
Owing to the Heisenberg-Weyl symmetry of $\mathfrak{su}(d)$, the Knill-Laflamme conditions have been somewhat reduced --- from $d^6$ quadratic forms to $d^6-2d^5+d^4+d^3/2+d/2$. However, we evidently still have a lot of conditions to satisfy, so devising a procedure for constructing error-correcting codes for a qudit with large $d$, let alone a general odd-dimensional qudit, remains a herculean task. It therefore behooves us to make further simplifications. Abstracting the ideas in Ref. \cite{Herbert2023} to a general form that works for any odd-dimensional qudit, we consider two types of codes and show how these further reduce the Knill-Laflamme conditions to just three quadratic forms.

\subsection{Sparse codes}
Firstly, we consider sparse codes, which naturally protect against most errors involving dit flips. Informally, these codes map logical qudit states to orthogonal subspaces of $\mathcal{H}_P$ that are sufficiently ``far'' apart such that any sequence of at most 2 distinct dit flips will not send a code word to a subspace containing a code word. To get an intuition for what would be a good metric and minimum distance to use in our definition, see Lemma S3 in the Supplemental Material \cite{supp}.

For each $d$, define the following sets for all $N \equiv 1 \mod d$ and all $\nu \in \mathbf{Z}_d$:
\begin{equation}
    \mathcal{U}_{d,N,\nu} = \bigg\{\mathbf{u} \in \mathbf{N}_0^d \  \bigg| \sum_{\xi \in \mathbf{Z}_d} u_\xi = N \ , \  \bigoplus_{\xi \in \mathbf{Z}_d} \xi u_\xi = \nu \bigg\}.
\end{equation}
Suppose that $\mathcal{C} \colon \mathbf{C}^d \longrightarrow \mathrm{Sym}^N(\mathbf{C}^d)$ is a code with
\begin{equation}
    \ket{\overline{0}} = \sum_{\mathbf{u} \in \mathcal{Z}_{d,N,0}} \alpha_{\mathbf{u}} \ket{S_{\mathbf{u}}},
\end{equation}
for some set of constants $\left\{ \alpha_{\mathbf{u}} \mid \mathbf{u} \in \mathcal{Z}_{d,N,0}\right\}$. We say that $\mathcal{C}$ is \textit{sparse} if
\begin{equation}
    \sum_{\xi \in \mathbf{Z}_d} |u_\xi - u'_{\xi+\delta}| > 4
\end{equation}
for all $\mathbf{u},\mathbf{u}' \in \mathcal{Z}_{d,N,0}$ with $\alpha_{\mathbf{u}}, \alpha_{\mathbf{u}'} \neq 0$ and all $\delta \in \mathbf{Z}_d$.

Having formally defined sparse codes, we can now be more precise about what we mean by the natural protection they offer against most dit flip errors.
\begin{theorem}
    If a code $\mathcal{C} \colon \mathbf{C}^d \longrightarrow \mathrm{Sym}^N(\mathbf{C}^d)$ is sparse, then $\mathcal{C}$ satisfies Condition C1 and Condition C2 with $(p,q) \neq (r,s)$.
\end{theorem}

\begin{proof}
    See Supplemental Material \cite{supp}.
\end{proof}

We note that a sparse code does not necessarily satisfy Condition 2 with $(p,q)=(r,s)$. To see why, observe that the image of $\ket{S_\mathbf{u}} \in \mathrm{Sym}^N(\mathbf{C}^d)$ under two identical dit flips could have a non-zero component along $\ket{S_\mathbf{u}}$. Thus, given $0 \leq p < q \leq d-1$, the inner products $\bra{\overline{k}} \tilde{\Gamma}_d^{(p,q)} \tilde{\Upsilon}_d^{(p,q)} \ket{\overline{k}}$ might not vanish for all $k \in \mathbf{Z}_d$, in which case we really have to make sure that they are all equal to each other. However, the sparsity of a code does not guarantee this.

\subsection{Doubly permutation-invariant codes}
Next, we consider doubly permutation-invariant codes, which are meant to deal with the phase flip errors. To get a sense of such codes, consider the action of a phase flip on $\ket{S_\mathbf{u}} \in \mathrm{Sym}^N(\mathbf{C}^d)$:
\begin{equation}
	\mathsf{\tilde{D}}_d^{(\ell)} \ket{S_\mathbf{u}} = (u_\ell - u_{\ell+1})\ket{S_\mathbf{u}},
\end{equation}
which shows that $\ket{S_\mathbf{u}}$ is an eigenvector of $\mathsf{\tilde{D}}_d^{(\ell)}$ with an eigenvalue that is antisymmetric with respect to $u_\ell$ and $u_{\ell+1}$. Roughly speaking, it is therefore reasonable to conjecture that the invariance of a code word under permutations of $u_1,\cdots,u_{d-1}$ would lead to the vanishing of a lot of the inner products in Conditions C3 and C4.

Let us now be formal about our description of doubly permutation-invariant codes. To do this, we first define the following sets:
\begin{align}
    \mathcal{V}_{d,N} = \bigg\{\mathbf{v} \in \mathbf{N}_0^d \  \bigg| \ v_1 \equiv \cdots& \equiv v_{d-1} \  \mathrm{mod} \  d,\nonumber\\
    & \sum_{\xi \in \mathbf{Z}_d} v_\xi = N\bigg\}\\
    \mathcal{W}_{d,N} = \mathcal{V}_{d,N} \cap \bigg\{\mathbf{w} \in \mathbf{N}_0^d \  \bigg| \ &\bigoplus_{\xi \in \mathbf{Z}_d} \xi w_\xi = 0 \bigg\}
\end{align}

For each $d$, define the relation $\mathcal{R}$ on the set $\mathcal{V}_{d,N}$ by $(v_0,v_1,\cdots,v_{d-1}) \  \mathcal{R} \  (v'_0,v'_1,\cdots,v'_{d-1})$ if and only if $(v'_0,v'_1,\cdots,v'_{d-1}) = (v_0,v_{\sigma(1)},\cdots,v_{\sigma(d-1)})$ for some $\sigma \in S_{d-1}$. A quick inspection shows that $\mathcal{R}$ is an equivalence relation. Note that restricting to $\mathcal{W}_{d,N} \subset \mathcal{V}_{d,N}$ gives us an equivalence relation $\hat{\mathcal{R}}$ on $\mathcal{W}_{d,N}$.

For each $\mathbf{v} \in \mathcal{V}_{d,N}$, we can define a \textit{doubly permutation-invariant vector} by
\begin{equation}
    \ket{\ket{S_\mathbf{v}}} = \sum_{\mathbf{v}' \in [\mathbf{v}]_{\mathcal{R}}^{}} \ket{S_{\mathbf{v}'}}.
\end{equation}
We then define a \textit{doubly symmetric vector space} as
\begin{equation}
    \mathrm{DSym}^N(\mathbf{C}^d) = \mathrm{span}_\mathbf{C}\left\{ \ket{\ket{S_\mathbf{v}}} \mid \mathbf{v} \in \mathcal{V}_{d,N} \right\}.
\end{equation}

Now, we say that a code is a \textit{doubly permutation-invariant code} if it maps to a doubly symmetric physical Hilbert space, i.e.
\begin{equation}
    \mathcal{H}_P = \mathrm{DSym}^N(\mathbf{C}^d).
\end{equation}
Equivalently, a code is doubly permutation-invariant if the code word for $\ket{0}_L$ is of the form
\begin{equation}
    \label{DPI code word}
    \ket{\overline{0}} = \sum_{\mathbf{w} \in f(\mathcal{W}_{d,N}/\hat{\mathcal{R}})} \alpha_{\mathbf{w}} \ket{\ket{S_\mathbf{w}}}
\end{equation}
for some set of constants $\left\{ \alpha_{\mathbf{w}} \mid \mathbf{w} \in f(\mathcal{W}_{d,N}/\hat{\mathcal{R}}) \right\}$ and some injection $f \colon \mathcal{W}_{d,N}/\hat{\mathcal{R}} \longrightarrow \mathcal{W}_{d,N}$ that maps each equivalence class to a chosen representative.

We can use the double permutation invariance of a code to simplify Conditions C3 and C4:
\begin{theorem}
	\label{reduced KL conditions for DPI codes}
    Let $\mathcal{C} \colon \mathbf{C}^d \longrightarrow \mathrm{DSym}^N(\mathbf{C}^d)$ be a doubly permutation-invariant code. If $\mathcal{C}$ satisfies
    \begin{enumerate}[\rm (QF1),leftmargin=3em]
        \item $\bra{\overline{d-1}} \mathsf{\tilde{D}}_d^{(d-2)} \ket{\overline{d-1}} = 0$
        \item $\bra{\overline{0}} \mathsf{\tilde{D}}_d^{(d-2)}\mathsf{\tilde{D}}_d^{(d-2)} \ket{\overline{0}} = \bra{\overline{d-1}} \mathsf{\tilde{D}}_d^{(d-2)}\mathsf{\tilde{D}}_d^{(d-2)} \ket{\overline{d-1}}$
    \end{enumerate}
    then it satisfies Condition C3 and Condition C4.
\end{theorem}

\begin{proof}
    See Supplemental Material \cite{supp}.
\end{proof}

\subsection{SDPI codes}
Now, we combine both types and introduce the notion of a \textit{sparse and doubly permutation-invariant (SDPI) code}. By itself, the sparsity of a code does not guarantee that the code would satisfy Condition C2 with $(p,q)=(r,s)$. However, when this is combined with double permutation invariance, the condition simplifies into just one quadratic form:
\begin{theorem}
	Let $\mathcal{C} \colon \mathbf{C}^d \longrightarrow \mathrm{DSym}^N(\mathbf{C}^d)$ be an SDPI code. If $\mathcal{C}$ satisfies
    \begin{enumerate}[\rm (QF3),leftmargin=3em]
        \item $\bra{\overline{0}} \mathsf{\tilde{S}}_d^{(0,1)}\mathsf{\tilde{S}}_d^{(0,1)} \ket{\overline{0}} = \bra{\overline{d-1}} \mathsf{\tilde{S}}_d^{(0,1)}\mathsf{\tilde{S}}_d^{(0,1)} \ket{\overline{d-1}}$
    \end{enumerate}
    then it satisfies Condition C2 with $(p,q) = (r,s)$.
\end{theorem}

\begin{proof}
	 See Supplemental Material \cite{supp}.
\end{proof}

We now arrive at one of the main results of this paper. Corollary \ref{fully reduced KL conditions} shows that, for SDPI codes, the $d^6$ quadratic forms that make up the Knill-Laflamme conditions reduce to just three!

\begin{corollary-theorem}
	\label{fully reduced KL conditions}
	If an SDPI code $\mathcal{C} \colon \mathbf{C}^d \longrightarrow \mathrm{DSym}^N(\mathbf{C}^d)$ satisfies the quadratic forms QF1 to QF3, then $\mathcal{C}$ is a QECC.
\end{corollary-theorem}

\begin{proof}
	This follows trivially from Theorems 7 and 8.
\end{proof}

\subsection{Procedure}
Now that we have established Corollary \ref{fully reduced KL conditions}, we can apply this knowledge to devise a procedure for constructing error-correcting qudit codes that work for any odd integer $d \geq 3$. The procedure is as follows:

\begin{enumerate}
	\item Choose a dimension $d$ for the logical qudit and the number of physical qudits $N$, noting that it must satisfy $N \equiv 1 \mod d$.
	\item Pick a subset $\mathcal{S} \subset f(\mathcal{W}_{d,N}/\hat{\mathcal{R}})$ such that $\mathcal{C}$ is sparse if
		\begin{equation}
  			\alpha_\mathbf{s} \neq 0, \qquad \alpha_\mathbf{t} = 0
		\end{equation}
	for all $\mathbf{s} \in \mathcal{S}$ and $\mathbf{t} \in f(\mathcal{W}_{d,N}/\hat{\mathcal{R}}) \backslash \mathcal{S}$.
	\item Define the code $\mathcal{C} \colon \mathbf{C}^d \longrightarrow \mathrm{DSym}^N(\mathbf{C}^d)$ by
	\begin{equation}
		\ket{\overline{0}} = \sum_{\mathbf{s} \in \mathcal{S}} \alpha_\mathbf{s} \ket{\ket{S_\mathbf{s}}},
	\end{equation}
	with $\alpha_\mathbf{s}$ arbitrary. Recall from Sec. II D that this suffices in defining $\mathcal{C}$ since the other code words are given by $\ket{\overline{k}} = \mathsf{\overline{X}}_d^k \ket{\overline{0}}$, for any $k \in \mathbf{Z}_d$.
	\item Evaluate the inner products in QF1 to QF3 for the code defined above, and solve the system of quadratic forms for non-trivial solution(s) $\{\alpha_\mathbf{s} \mid \mathbf{s} \in \mathcal{S}\}$.
	\item If a non-trivial solution does not exist, return to Step 2 and choose another subset $\mathcal{S}$, or return to Step 1 and choose a larger $N$.
\end{enumerate}

When picking a subset $\mathcal{S}$ in Step 2, note that for a code to be error-correcting, there necessarily exist $\mathbf{s}', \mathbf{s}'' \in \mathcal{S}$ such that $s'_i > s'_0$ for some $i \in \{1,\cdots,d-1\}$ and $s''_j < s''_0$ for some $j \in \{1,\cdots,d-1\}$. To see why, observe that QF1 gives us
\begin{equation}
    \label{concrete QF1}
    \sum_{\mathbf{s} \in \mathcal{S}} \beta_\mathbf{s} \braket{S_\mathbf{s}|S_\mathbf{s}} \left|\left| \alpha_\mathbf{s} \right|\right|^2 = 0
\end{equation}
for some set of coefficients $\{\beta_\mathbf{s} \mid \mathbf{s} \in \mathcal{S}\}$. If the aforementioned necessary condition is not satisfied, then $\beta_\mathbf{s}>0$ or $\beta_\mathbf{s}<0$ for all $\mathbf{s} \in \mathcal{S}$. This then implies that there is no non-trivial solution to Eq. (\ref{concrete QF1}).

\section{Some error-correcting SDPI codes}
With the procedure established, we can now employ it to construct concrete examples of error-correcting SDPI codes.

\subsection{Codes for qutrits}
An example of an error-correcting SDPI code for qutrits is the map $\mathcal{C}_1 \colon \mathbf{C}^3 \longrightarrow \mathrm{DSym}^{13}(\mathbf{C}^3)$ defined by
\begin{align}
    \ket{\overline{0}} = \frac{1}{9}\sqrt{\frac{41}{5}}& \ket{\ket{S_{(13,0,0)}}} + \frac{1}{9\sqrt{55}} \ket{\ket{S_{(4,0,9)}}}\nonumber\\
    &+ \frac{1}{18\sqrt{385}} \ket{\ket{S_{(3,5,5)}}}.
    \label{qutrit code example}
\end{align}
We now show how we arrive at this through our procedure. For $N=13$, observe that we can pick the subset $\mathcal{S} = \{(13,0,0), (4,0,9), (3,5,5)\}$. We then obtain the following system of linear equations in $\xi_1,\xi_2,\xi_3$ (or, equivalently, quadratic equations in $\alpha_{(13,0,0)}, \alpha_{(4,0,9)}, \alpha_{(3,5,5)}$):
\begin{align}
	13\xi_1 - \xi_2- 2\xi_3 &= 0,\\
	169\xi_1 - 121\xi_2 + 4\xi_3 &= 0,\\
	13\xi_1 + 71\xi_2 - 22\xi_3 &= 0,
\end{align}
where $\xi_1$,$\xi_2$,$\xi_3$ are defined as
\begin{align}
	\xi_1 &= \braket{S_{(13,0,0)}|S_{(13,0,0)}}||\alpha_{(13,0,0)}||^2\\
	\xi_2 &= \braket{S_{(4,0,9)}|S_{(4,0,9)}}||\alpha_{(4,0,9)}||^2\\
	\xi_3 &= \braket{S_{(3,5,5)}|S_{(3,5,5)}}||\alpha_{(3,5,5)}||^2
\end{align}
Solving the linear system using \textsc{Mathematica} or elementary linear-algebraic methods yields the code in Eq. (\ref{qutrit code example}).

\subsection{Codes for 5-level qudits and beyond}
For $d=5$, an example is the code $\mathcal{C}_2 \colon \mathbf{C}^5 \longrightarrow \mathrm{DSym}^{16}(\mathbf{C}^5)$ defined by
\begin{align}
    \ket{\overline{0}} = \frac{1}{5\sqrt{5}}& \ket{\ket{S_{(16,0,0,0,0)}}} + \frac{1}{5}\sqrt{\frac{2}{5005}} \ket{\ket{S_{(6,10,0,0,0)}}}\nonumber\\
    &+ \frac{1}{175\sqrt{4290}} \ket{\ket{S_{(0,4,4,4,4)}}}.
\end{align}
Similarly, for $d=7$, we have the code $\mathcal{C}_3 \colon \mathbf{C}^7 \longrightarrow \mathrm{DSym}^{36}(\mathbf{C}^7)$ defined by
\begin{align}
    \ket{\overline{0}} = \frac{1}{7}&\sqrt{\frac{13}{7}} \ket{\ket{S_{(36,0,0,0,0,0,0)}}} \nonumber\\
    &+ \frac{1}{14}\sqrt{\frac{3}{5883955}} \ket{\ket{S_{(8,28,0,0,0,0,0)}}}\nonumber\\
    &+ \frac{1}{300179880\sqrt{60500902}} \ket{\ket{S_{(0,6,6,6,6,6,6)}}}
\end{align}

In fact, the form of the code words given for $d=5$ and $d=7$ can be generalized. Theorem \ref{general QECC} below shows that we can encode a logical qudit of any odd dimension $d \geq 5$ into $(d-1)^2$ physical qudits.
\begin{theorem}
	\label{general QECC}
    Let $\mathbf{a} = ((d-1)^2,\underbrace{0,\cdots,0}_{d-1})$, $\mathbf{b} = (d+1,d(d-3),\underbrace{0,\cdots,0}_{d-2})$, $\mathbf{c} = (0,\underbrace{d-1,\cdots,d-1}_{d-1})$. For each odd integer $d \geq 5$, there exists an error-correcting SDPI code with the logical code word
    \begin{equation}
        \label{general code word}
        \ket{\overline{0}} = \alpha_\mathbf{a}\ket{\ket{S_\mathbf{a}}} + \alpha_\mathbf{b}\ket{\ket{S_\mathbf{b}}} + \alpha_\mathbf{c}\ket{\ket{S_\mathbf{c}}}
    \end{equation}
    for some coefficients $\alpha_\mathbf{a}$, $\alpha_\mathbf{b}$, $\alpha_\mathbf{c}$. In particular, these coefficients are given by
    \begin{align}
        \alpha_\mathbf{a} &= \sqrt{\frac{d^3 - 5d^2 + d - 1}{2d^4 - 6d^3}},\\
        \alpha_\mathbf{b} &= \sqrt{\binom{(d-1)^2}{\mathbf{b}}^{-1} \frac{(d-1)(1-d\alpha_\mathbf{a}^2)}{d^2+d}},\\
        \alpha_\mathbf{c} &= \sqrt{\binom{(d-1)^2}{\mathbf{c}}^{-1} \left(1 - \alpha_\mathbf{a}^2 + \frac{1-d}{d^2+d}\alpha_\mathbf{b}^2\right)}.
    \end{align}
\end{theorem}

\begin{proof}
    See Supplemental Material \cite{supp}.
\end{proof}

\subsection{Alternative definitions}
As mentioned in Section II C, we could also consider alternative definitions of the action of $\mathsf{Z}_L$ on $\mathcal{H}_L$. In general, we could define it, for any $\eta \in \mathbf{Z}_d^\times$, as
\begin{equation}
    \mathsf{Z}_L \ket{k}_L = \zeta_d^{k\eta} \ket{k}_L
\end{equation}
for all $k \in \mathbf{Z}_d$. Ultimately, we are still able to implement all qudit Pauli logical operations, and that is what really matters anyway.

Once we have chosen $\eta \in \mathbf{Z}_d^\times$, we must then restrict ourselves to SDPI codes that map to $\mathcal{H}_P = \mathrm{DSym}^N(\mathbf{C}^d)$ with $N \equiv \eta \mod d$.

Now, let us consider a concrete example of such codes. Suppose we fix $\eta = d-1$. In this case, we can construct an error-correcting SDPI code for $d=7$ mapping to a physical Hilbert space of smaller dimension than the previous example. In particular, this is the map $\mathcal{C}_4 \colon \mathbf{C}^7 \longrightarrow \mathrm{DSym}^{20}(\mathbf{C}^7)$ defined by
\begin{align}
    \ket{\overline{0}} = \frac{1}{7}&\sqrt{\frac{3}{7}} \ket{\ket{S_{(20,0,0,0,0,0,0)}}}\nonumber\\
    &+ \frac{1}{14\sqrt{6783}} \ket{\ket{S_{(6,14,0,0,0,0,0)}}}\nonumber\\
    &+ \frac{1}{11760\sqrt{230945}} \ket{\ket{S_{(2,3,3,3,3,3,3)}}}
\end{align}

\section{Conclusion}
In conclusion, we formulated a general procedure for constructing quantum error-correcting codes for qudits of any odd dimension $d \geq 3$. To do so, we determined the branching rules for $\mathrm{SU}(d) \downarrow \mathrm{HW}(d)$, which to our knowledge has not yet been reported in the literature. Moreover, by exploiting various forms of symmetry, we managed to reduce the Knill-Laflamme conditions from a Brobdingnagian set of $d^6$ quadratic forms to just 3. This allowed us to devise a general procedure, which we then used to explicitly construct an infinite class of error-correcting codes that encode a logical qudit into $(d-1)^2$ physical qudits.

Our work opens a broad spectrum of promising directions for further research. On the theoretical side, it would be interesting to extend our framework to a universal set of logical qudit gates. Currently, our procedure yields error-correcting codes that only allow us to perform qudit Pauli operations, which does not form a universal gate set. It has been shown that the set $\mathbb{C}(d) \cup \left\{ \mathsf{CSUM},\mathsf{T} \right\}$ is a universal set of qudit gates \cite{PhysRevA.86.022316,PhysRevX.2.041021,nebe2006self,PhysRevA.99.052307,Ringbauer2022}. Above, $\mathbb{C}(d)$ is the single-qudit Clifford group, $\mathsf{CSUM}$ is the generalization of the $\mathsf{CNOT}$ gate for qubits, and $\mathsf{T}$ is a single-qudit non-Clifford gate. Similar to the qubit case \cite{Gross2021}, one possible avenue for future work would then be to consider the single-qudit Clifford group as the finite subgroup of logical operations and subsequently determine how to implement $\mathsf{CSUM}$ and $\mathsf{T}$.

Furthermore, it would be worthwhile to investigate, for each $d$, the minimum value of $N$ for which there exists an error-correcting SDPI code that maps to $\mathrm{DSym}^N(\mathbf{C}^d)$. This would give insight into how resource-efficient our error-correcting SDPI codes are, i.e. how many physical qudits are required to encode a single logical qudit. In principle, this work can be done computationally, but we expect that this would become very computationally intensive as $d$ grows. To obtain general inequalities that work for the set of all odd integers $d \geq 3$ or at least some subset thereof, a proper mathematical analysis of the problem would probably be required.

Finally, going beyond our abstract representation-theoretic foray into quantum error correction on qudits, another natural next step would be to experimentally realize our error-correcting codes. Since the physical Hilbert space of our codes was chosen to be a vector space on which $\mathrm{SU}(d)$ irreps are defined, a natural physical implementation would be with ultracold atoms that realize $\mathrm{SU}(d)$ spin models \cite{PhysRevA.93.051601,Zhang2014,Taie2022}, arrays of high-spin nuclei in silicon \cite{fernandezdefuentes2024}, or nuclear spin ensembles in GaAs quantum dots \cite{appel2024manybodyquantumregisterspin}.

\begin{acknowledgments}
R.F.U. thanks Charlotte Franke for introducing him to Refs. \cite{Gross2021,PhysRevA.108.022424}, which inspired the work in this paper. R.F.U. also thanks Khoi Le Nguyen Nguyen for a helpful discussion on branching rules. This work was supported in part by the UKRI EPSRC Quantum Computing and Simulation Hub (EP/T001062/1) through the Partnership Resource Fund PRF-09-I-06 (D.G.). D.G. acknowledges a Royal Society University Research Fellowship.
\end{acknowledgments}


\bibliography{ms}

\providecommand{\noopsort}[1]{}\providecommand{\singleletter}[1]{#1}%
\begin{thebibliography}{55}%
\makeatletter
\providecommand \@ifxundefined [1]{%
 \@ifx{#1\undefined}
}%
\providecommand \@ifnum [1]{%
 \ifnum #1\expandafter \@firstoftwo
 \else \expandafter \@secondoftwo
 \fi
}%
\providecommand \@ifx [1]{%
 \ifx #1\expandafter \@firstoftwo
 \else \expandafter \@secondoftwo
 \fi
}%
\providecommand \natexlab [1]{#1}%
\providecommand \enquote  [1]{``#1''}%
\providecommand \bibnamefont  [1]{#1}%
\providecommand \bibfnamefont [1]{#1}%
\providecommand \citenamefont [1]{#1}%
\providecommand \href@noop [0]{\@secondoftwo}%
\providecommand \href [0]{\begingroup \@sanitize@url \@href}%
\providecommand \@href[1]{\@@startlink{#1}\@@href}%
\providecommand \@@href[1]{\endgroup#1\@@endlink}%
\providecommand \@sanitize@url [0]{\catcode `\\12\catcode `\$12\catcode `\&12\catcode `\#12\catcode `\^12\catcode `\_12\catcode `\%12\relax}%
\providecommand \@@startlink[1]{}%
\providecommand \@@endlink[0]{}%
\providecommand \url  [0]{\begingroup\@sanitize@url \@url }%
\providecommand \@url [1]{\endgroup\@href {#1}{\urlprefix }}%
\providecommand \urlprefix  [0]{URL }%
\providecommand \Eprint [0]{\href }%
\providecommand \doibase [0]{https://doi.org/}%
\providecommand \selectlanguage [0]{\@gobble}%
\providecommand \bibinfo  [0]{\@secondoftwo}%
\providecommand \bibfield  [0]{\@secondoftwo}%
\providecommand \translation [1]{[#1]}%
\providecommand \BibitemOpen [0]{}%
\providecommand \bibitemStop [0]{}%
\providecommand \bibitemNoStop [0]{.\EOS\space}%
\providecommand \EOS [0]{\spacefactor3000\relax}%
\providecommand \BibitemShut  [1]{\csname bibitem#1\endcsname}%
\let\auto@bib@innerbib\@empty
\bibitem [{\citenamefont {Preskill}(2018)}]{Preskill2018qcnisq}%
  \BibitemOpen
  \bibfield  {author} {\bibinfo {author} {\bibfnamefont {J.}~\bibnamefont {Preskill}},\ }\bibfield  {title} {\bibinfo {title} {Quantum {C}omputing in the {NISQ} era and beyond},\ }\href {https://doi.org/10.22331/q-2018-08-06-79} {\bibfield  {journal} {\bibinfo  {journal} {{Quantum}}\ }\textbf {\bibinfo {volume} {2}},\ \bibinfo {pages} {79} (\bibinfo {year} {2018})}\BibitemShut {NoStop}%
\bibitem [{\citenamefont {Cai}\ \emph {et~al.}(2023)\citenamefont {Cai}, \citenamefont {Babbush}, \citenamefont {Benjamin}, \citenamefont {Endo}, \citenamefont {Huggins}, \citenamefont {Li}, \citenamefont {McClean},\ and\ \citenamefont {O'Brien}}]{RevModPhys.95.045005}%
  \BibitemOpen
  \bibfield  {author} {\bibinfo {author} {\bibfnamefont {Z.}~\bibnamefont {Cai}}, \bibinfo {author} {\bibfnamefont {R.}~\bibnamefont {Babbush}}, \bibinfo {author} {\bibfnamefont {S.~C.}\ \bibnamefont {Benjamin}}, \bibinfo {author} {\bibfnamefont {S.}~\bibnamefont {Endo}}, \bibinfo {author} {\bibfnamefont {W.~J.}\ \bibnamefont {Huggins}}, \bibinfo {author} {\bibfnamefont {Y.}~\bibnamefont {Li}}, \bibinfo {author} {\bibfnamefont {J.~R.}\ \bibnamefont {McClean}},\ and\ \bibinfo {author} {\bibfnamefont {T.~E.}\ \bibnamefont {O'Brien}},\ }\bibfield  {title} {\bibinfo {title} {Quantum error mitigation},\ }\href {https://doi.org/10.1103/RevModPhys.95.045005} {\bibfield  {journal} {\bibinfo  {journal} {Rev. Mod. Phys.}\ }\textbf {\bibinfo {volume} {95}},\ \bibinfo {pages} {045005} (\bibinfo {year} {2023})}\BibitemShut {NoStop}%
\bibitem [{\citenamefont {Anto-Sztrikacs}\ \emph {et~al.}(2023)\citenamefont {Anto-Sztrikacs}, \citenamefont {Nazir},\ and\ \citenamefont {Segal}}]{PRXQuantum.4.020307}%
  \BibitemOpen
  \bibfield  {author} {\bibinfo {author} {\bibfnamefont {N.}~\bibnamefont {Anto-Sztrikacs}}, \bibinfo {author} {\bibfnamefont {A.}~\bibnamefont {Nazir}},\ and\ \bibinfo {author} {\bibfnamefont {D.}~\bibnamefont {Segal}},\ }\bibfield  {title} {\bibinfo {title} {Effective-{H}amiltonian theory of open quantum systems at strong coupling},\ }\href {https://doi.org/10.1103/PRXQuantum.4.020307} {\bibfield  {journal} {\bibinfo  {journal} {PRX Quantum}\ }\textbf {\bibinfo {volume} {4}},\ \bibinfo {pages} {020307} (\bibinfo {year} {2023})}\BibitemShut {NoStop}%
\bibitem [{\citenamefont {Hu}\ \emph {et~al.}(2018)\citenamefont {Hu}, \citenamefont {Mu}, \citenamefont {Cai}, \citenamefont {Ma}, \citenamefont {Xu}, \citenamefont {Wang}, \citenamefont {Song}, \citenamefont {Zou},\ and\ \citenamefont {Sun}}]{Hu2018}%
  \BibitemOpen
  \bibfield  {author} {\bibinfo {author} {\bibfnamefont {L.}~\bibnamefont {Hu}}, \bibinfo {author} {\bibfnamefont {X.}~\bibnamefont {Mu}}, \bibinfo {author} {\bibfnamefont {W.}~\bibnamefont {Cai}}, \bibinfo {author} {\bibfnamefont {Y.}~\bibnamefont {Ma}}, \bibinfo {author} {\bibfnamefont {Y.}~\bibnamefont {Xu}}, \bibinfo {author} {\bibfnamefont {H.}~\bibnamefont {Wang}}, \bibinfo {author} {\bibfnamefont {Y.}~\bibnamefont {Song}}, \bibinfo {author} {\bibfnamefont {C.-L.}\ \bibnamefont {Zou}},\ and\ \bibinfo {author} {\bibfnamefont {L.}~\bibnamefont {Sun}},\ }\bibfield  {title} {\bibinfo {title} {Experimental repetitive quantum channel simulation},\ }\href {https://doi.org/https://doi.org/10.1016/j.scib.2018.11.010} {\bibfield  {journal} {\bibinfo  {journal} {Science Bulletin}\ }\textbf {\bibinfo {volume} {63}},\ \bibinfo {pages} {1551} (\bibinfo {year} {2018})}\BibitemShut {NoStop}%
\bibitem [{\citenamefont {Lim}\ \emph {et~al.}(2023)\citenamefont {Lim}, \citenamefont {Liu},\ and\ \citenamefont {Ardavan}}]{PhysRevA.108.062403}%
  \BibitemOpen
  \bibfield  {author} {\bibinfo {author} {\bibfnamefont {S.}~\bibnamefont {Lim}}, \bibinfo {author} {\bibfnamefont {J.}~\bibnamefont {Liu}},\ and\ \bibinfo {author} {\bibfnamefont {A.}~\bibnamefont {Ardavan}},\ }\bibfield  {title} {\bibinfo {title} {Fault-tolerant qubit encoding using a spin-7/2 qudit},\ }\href {https://doi.org/10.1103/PhysRevA.108.062403} {\bibfield  {journal} {\bibinfo  {journal} {Phys. Rev. A}\ }\textbf {\bibinfo {volume} {108}},\ \bibinfo {pages} {062403} (\bibinfo {year} {2023})}\BibitemShut {NoStop}%
\bibitem [{\citenamefont {Bravyi}\ \emph {et~al.}(2024)\citenamefont {Bravyi}, \citenamefont {Cross}, \citenamefont {Gambetta}, \citenamefont {Maslov}, \citenamefont {Rall},\ and\ \citenamefont {Yoder}}]{Bravyi2024}%
  \BibitemOpen
  \bibfield  {author} {\bibinfo {author} {\bibfnamefont {S.}~\bibnamefont {Bravyi}}, \bibinfo {author} {\bibfnamefont {A.~W.}\ \bibnamefont {Cross}}, \bibinfo {author} {\bibfnamefont {J.~M.}\ \bibnamefont {Gambetta}}, \bibinfo {author} {\bibfnamefont {D.}~\bibnamefont {Maslov}}, \bibinfo {author} {\bibfnamefont {P.}~\bibnamefont {Rall}},\ and\ \bibinfo {author} {\bibfnamefont {T.~J.}\ \bibnamefont {Yoder}},\ }\bibfield  {title} {\bibinfo {title} {High-threshold and low-overhead fault-tolerant quantum memory},\ }\href {https://doi.org/10.1038/s41586-024-07107-7} {\bibfield  {journal} {\bibinfo  {journal} {Nature}\ }\textbf {\bibinfo {volume} {627}},\ \bibinfo {pages} {778} (\bibinfo {year} {2024})}\BibitemShut {NoStop}%
\bibitem [{\citenamefont {Cai}\ \emph {et~al.}(2021)\citenamefont {Cai}, \citenamefont {Ma}, \citenamefont {Wang}, \citenamefont {Zou},\ and\ \citenamefont {Sun}}]{CAI202150}%
  \BibitemOpen
  \bibfield  {author} {\bibinfo {author} {\bibfnamefont {W.}~\bibnamefont {Cai}}, \bibinfo {author} {\bibfnamefont {Y.}~\bibnamefont {Ma}}, \bibinfo {author} {\bibfnamefont {W.}~\bibnamefont {Wang}}, \bibinfo {author} {\bibfnamefont {C.-L.}\ \bibnamefont {Zou}},\ and\ \bibinfo {author} {\bibfnamefont {L.}~\bibnamefont {Sun}},\ }\bibfield  {title} {\bibinfo {title} {Bosonic quantum error correction codes in superconducting quantum circuits},\ }\href {https://doi.org/https://doi.org/10.1016/j.fmre.2020.12.006} {\bibfield  {journal} {\bibinfo  {journal} {Fundamental Research}\ }\textbf {\bibinfo {volume} {1}},\ \bibinfo {pages} {50} (\bibinfo {year} {2021})}\BibitemShut {NoStop}%
\bibitem [{\citenamefont {Fischer}\ \emph {et~al.}(2023)\citenamefont {Fischer}, \citenamefont {Chiesa}, \citenamefont {Tacchino}, \citenamefont {Egger}, \citenamefont {Carretta},\ and\ \citenamefont {Tavernelli}}]{PRXQuantum.4.030327}%
  \BibitemOpen
  \bibfield  {author} {\bibinfo {author} {\bibfnamefont {L.~E.}\ \bibnamefont {Fischer}}, \bibinfo {author} {\bibfnamefont {A.}~\bibnamefont {Chiesa}}, \bibinfo {author} {\bibfnamefont {F.}~\bibnamefont {Tacchino}}, \bibinfo {author} {\bibfnamefont {D.~J.}\ \bibnamefont {Egger}}, \bibinfo {author} {\bibfnamefont {S.}~\bibnamefont {Carretta}},\ and\ \bibinfo {author} {\bibfnamefont {I.}~\bibnamefont {Tavernelli}},\ }\bibfield  {title} {\bibinfo {title} {Universal qudit gate synthesis for transmons},\ }\href {https://doi.org/10.1103/PRXQuantum.4.030327} {\bibfield  {journal} {\bibinfo  {journal} {PRX Quantum}\ }\textbf {\bibinfo {volume} {4}},\ \bibinfo {pages} {030327} (\bibinfo {year} {2023})}\BibitemShut {NoStop}%
\bibitem [{\citenamefont {Nikolaeva}\ \emph {et~al.}(2024)\citenamefont {Nikolaeva}, \citenamefont {Kiktenko},\ and\ \citenamefont {Fedorov}}]{Nikolaeva2024}%
  \BibitemOpen
  \bibfield  {author} {\bibinfo {author} {\bibfnamefont {A.~S.}\ \bibnamefont {Nikolaeva}}, \bibinfo {author} {\bibfnamefont {E.~O.}\ \bibnamefont {Kiktenko}},\ and\ \bibinfo {author} {\bibfnamefont {A.~K.}\ \bibnamefont {Fedorov}},\ }\bibfield  {title} {\bibinfo {title} {Efficient realization of quantum algorithms with qudits},\ }\href {https://doi.org/10.1140/epjqt/s40507-024-00250-0} {\bibfield  {journal} {\bibinfo  {journal} {EPJ Quantum Technology}\ }\textbf {\bibinfo {volume} {11}},\ \bibinfo {pages} {43} (\bibinfo {year} {2024})}\BibitemShut {NoStop}%
\bibitem [{\citenamefont {Kehrer}\ \emph {et~al.}(2024)\citenamefont {Kehrer}, \citenamefont {Nadolny},\ and\ \citenamefont {Bruder}}]{PhysRevResearch.6.013050}%
  \BibitemOpen
  \bibfield  {author} {\bibinfo {author} {\bibfnamefont {T.}~\bibnamefont {Kehrer}}, \bibinfo {author} {\bibfnamefont {T.}~\bibnamefont {Nadolny}},\ and\ \bibinfo {author} {\bibfnamefont {C.}~\bibnamefont {Bruder}},\ }\bibfield  {title} {\bibinfo {title} {{Improving transmon qudit measurement on IBM Quantum hardware}},\ }\href {https://doi.org/10.1103/PhysRevResearch.6.013050} {\bibfield  {journal} {\bibinfo  {journal} {Phys. Rev. Res.}\ }\textbf {\bibinfo {volume} {6}},\ \bibinfo {pages} {013050} (\bibinfo {year} {2024})}\BibitemShut {NoStop}%
\bibitem [{\citenamefont {Michael}\ \emph {et~al.}(2016)\citenamefont {Michael}, \citenamefont {Silveri}, \citenamefont {Brierley}, \citenamefont {Albert}, \citenamefont {Salmilehto}, \citenamefont {Jiang},\ and\ \citenamefont {Girvin}}]{PhysRevX.6.031006}%
  \BibitemOpen
  \bibfield  {author} {\bibinfo {author} {\bibfnamefont {M.~H.}\ \bibnamefont {Michael}}, \bibinfo {author} {\bibfnamefont {M.}~\bibnamefont {Silveri}}, \bibinfo {author} {\bibfnamefont {R.~T.}\ \bibnamefont {Brierley}}, \bibinfo {author} {\bibfnamefont {V.~V.}\ \bibnamefont {Albert}}, \bibinfo {author} {\bibfnamefont {J.}~\bibnamefont {Salmilehto}}, \bibinfo {author} {\bibfnamefont {L.}~\bibnamefont {Jiang}},\ and\ \bibinfo {author} {\bibfnamefont {S.~M.}\ \bibnamefont {Girvin}},\ }\bibfield  {title} {\bibinfo {title} {New class of quantum error-correcting codes for a bosonic mode},\ }\href {https://doi.org/10.1103/PhysRevX.6.031006} {\bibfield  {journal} {\bibinfo  {journal} {Phys. Rev. X}\ }\textbf {\bibinfo {volume} {6}},\ \bibinfo {pages} {031006} (\bibinfo {year} {2016})}\BibitemShut {NoStop}%
\bibitem [{\citenamefont {Ouyang}(2017)}]{OUYANG201743}%
  \BibitemOpen
  \bibfield  {author} {\bibinfo {author} {\bibfnamefont {Y.}~\bibnamefont {Ouyang}},\ }\bibfield  {title} {\bibinfo {title} {Permutation-invariant qudit codes from polynomials},\ }\href {https://doi.org/https://doi.org/10.1016/j.laa.2017.06.031} {\bibfield  {journal} {\bibinfo  {journal} {Linear Algebra and its Applications}\ }\textbf {\bibinfo {volume} {532}},\ \bibinfo {pages} {43} (\bibinfo {year} {2017})}\BibitemShut {NoStop}%
\bibitem [{\citenamefont {Mazurek}\ \emph {et~al.}(2020)\citenamefont {Mazurek}, \citenamefont {Farkas}, \citenamefont {Grudka}, \citenamefont {Horodecki},\ and\ \citenamefont {Studzi\ifmmode~\acute{n}\else \'{n}\fi{}ski}}]{PhysRevA.101.042305}%
  \BibitemOpen
  \bibfield  {author} {\bibinfo {author} {\bibfnamefont {P.}~\bibnamefont {Mazurek}}, \bibinfo {author} {\bibfnamefont {M.}~\bibnamefont {Farkas}}, \bibinfo {author} {\bibfnamefont {A.}~\bibnamefont {Grudka}}, \bibinfo {author} {\bibfnamefont {M.}~\bibnamefont {Horodecki}},\ and\ \bibinfo {author} {\bibfnamefont {M.}~\bibnamefont {Studzi\ifmmode~\acute{n}\else \'{n}\fi{}ski}},\ }\bibfield  {title} {\bibinfo {title} {Quantum error-correction codes and absolutely maximally entangled states},\ }\href {https://doi.org/10.1103/PhysRevA.101.042305} {\bibfield  {journal} {\bibinfo  {journal} {Phys. Rev. A}\ }\textbf {\bibinfo {volume} {101}},\ \bibinfo {pages} {042305} (\bibinfo {year} {2020})}\BibitemShut {NoStop}%
\bibitem [{\citenamefont {Schmidt}\ and\ \citenamefont {van Loock}(2022)}]{PhysRevA.105.042427}%
  \BibitemOpen
  \bibfield  {author} {\bibinfo {author} {\bibfnamefont {F.}~\bibnamefont {Schmidt}}\ and\ \bibinfo {author} {\bibfnamefont {P.}~\bibnamefont {van Loock}},\ }\bibfield  {title} {\bibinfo {title} {{Quantum error correction with higher Gottesman-Kitaev-Preskill codes: Minimal measurements and linear optics}},\ }\href {https://doi.org/10.1103/PhysRevA.105.042427} {\bibfield  {journal} {\bibinfo  {journal} {Phys. Rev. A}\ }\textbf {\bibinfo {volume} {105}},\ \bibinfo {pages} {042427} (\bibinfo {year} {2022})}\BibitemShut {NoStop}%
\bibitem [{\citenamefont {Dutta}\ \emph {et~al.}(2024)\citenamefont {Dutta}, \citenamefont {Biswas},\ and\ \citenamefont {Mandayam}}]{dutta2024}%
  \BibitemOpen
  \bibfield  {author} {\bibinfo {author} {\bibfnamefont {S.}~\bibnamefont {Dutta}}, \bibinfo {author} {\bibfnamefont {D.}~\bibnamefont {Biswas}},\ and\ \bibinfo {author} {\bibfnamefont {P.}~\bibnamefont {Mandayam}},\ }\href {https://arxiv.org/abs/2406.02444} {\bibinfo {title} {Noise-adapted qudit codes for amplitude-damping noise}} (\bibinfo {year} {2024}),\ \Eprint {https://arxiv.org/abs/2406.02444} {arXiv:2406.02444 [quant-ph]} \BibitemShut {NoStop}%
\bibitem [{\citenamefont {Gross}(2021)}]{Gross2021}%
  \BibitemOpen
  \bibfield  {author} {\bibinfo {author} {\bibfnamefont {J.~A.}\ \bibnamefont {Gross}},\ }\bibfield  {title} {\bibinfo {title} {Designing codes around interactions: The case of a spin},\ }\href {https://doi.org/10.1103/PhysRevLett.127.010504} {\bibfield  {journal} {\bibinfo  {journal} {Phys. Rev. Lett.}\ }\textbf {\bibinfo {volume} {127}},\ \bibinfo {pages} {010504} (\bibinfo {year} {2021})}\BibitemShut {NoStop}%
\bibitem [{\citenamefont {Herbert}\ \emph {et~al.}(2023)\citenamefont {Herbert}, \citenamefont {Gross},\ and\ \citenamefont {Newman}}]{Herbert2023}%
  \BibitemOpen
  \bibfield  {author} {\bibinfo {author} {\bibfnamefont {X.}~\bibnamefont {Herbert}}, \bibinfo {author} {\bibfnamefont {J.}~\bibnamefont {Gross}},\ and\ \bibinfo {author} {\bibfnamefont {M.}~\bibnamefont {Newman}},\ }\href {https://arxiv.org/abs/2312.00162} {\bibinfo {title} {{Qutrit codes within representations of SU(3)}}} (\bibinfo {year} {2023}),\ \Eprint {https://arxiv.org/abs/2312.00162} {arXiv:2312.00162 [quant-ph]} \BibitemShut {NoStop}%
\bibitem [{\citenamefont {Gottesman}(1999)}]{Gottesman1999}%
  \BibitemOpen
  \bibfield  {author} {\bibinfo {author} {\bibfnamefont {D.}~\bibnamefont {Gottesman}},\ }\bibfield  {title} {\bibinfo {title} {Fault-tolerant quantum computation with higher-dimensional systems},\ }in\ \href@noop {} {\emph {\bibinfo {booktitle} {Quantum Computing and Quantum Communications}}},\ \bibinfo {editor} {edited by\ \bibinfo {editor} {\bibfnamefont {C.~P.}\ \bibnamefont {Williams}}}\ (\bibinfo  {publisher} {Springer Berlin Heidelberg},\ \bibinfo {address} {Berlin, Heidelberg},\ \bibinfo {year} {1999})\ pp.\ \bibinfo {pages} {302--313}\BibitemShut {NoStop}%
\bibitem [{\citenamefont {Cafaro}\ \emph {et~al.}(2012)\citenamefont {Cafaro}, \citenamefont {Maiolini},\ and\ \citenamefont {Mancini}}]{PhysRevA.86.022308}%
  \BibitemOpen
  \bibfield  {author} {\bibinfo {author} {\bibfnamefont {C.}~\bibnamefont {Cafaro}}, \bibinfo {author} {\bibfnamefont {F.}~\bibnamefont {Maiolini}},\ and\ \bibinfo {author} {\bibfnamefont {S.}~\bibnamefont {Mancini}},\ }\bibfield  {title} {\bibinfo {title} {Quantum stabilizer codes embedding qubits into qudits},\ }\href {https://doi.org/10.1103/PhysRevA.86.022308} {\bibfield  {journal} {\bibinfo  {journal} {Phys. Rev. A}\ }\textbf {\bibinfo {volume} {86}},\ \bibinfo {pages} {022308} (\bibinfo {year} {2012})}\BibitemShut {NoStop}%
\bibitem [{\citenamefont {Sarkar}\ and\ \citenamefont {Yoder}(2024)}]{Sarkar2024}%
  \BibitemOpen
  \bibfield  {author} {\bibinfo {author} {\bibfnamefont {R.}~\bibnamefont {Sarkar}}\ and\ \bibinfo {author} {\bibfnamefont {T.~J.}\ \bibnamefont {Yoder}},\ }\bibfield  {title} {\bibinfo {title} {The qudit {P}auli group: non-commuting pairs, non-commuting sets, and structure theorems},\ }\href {https://doi.org/10.22331/q-2024-04-04-1307} {\bibfield  {journal} {\bibinfo  {journal} {{Quantum}}\ }\textbf {\bibinfo {volume} {8}},\ \bibinfo {pages} {1307} (\bibinfo {year} {2024})}\BibitemShut {NoStop}%
\bibitem [{\citenamefont {Faist}\ \emph {et~al.}(2020)\citenamefont {Faist}, \citenamefont {Nezami}, \citenamefont {Albert}, \citenamefont {Salton}, \citenamefont {Pastawski}, \citenamefont {Hayden},\ and\ \citenamefont {Preskill}}]{PhysRevX.10.041018}%
  \BibitemOpen
  \bibfield  {author} {\bibinfo {author} {\bibfnamefont {P.}~\bibnamefont {Faist}}, \bibinfo {author} {\bibfnamefont {S.}~\bibnamefont {Nezami}}, \bibinfo {author} {\bibfnamefont {V.~V.}\ \bibnamefont {Albert}}, \bibinfo {author} {\bibfnamefont {G.}~\bibnamefont {Salton}}, \bibinfo {author} {\bibfnamefont {F.}~\bibnamefont {Pastawski}}, \bibinfo {author} {\bibfnamefont {P.}~\bibnamefont {Hayden}},\ and\ \bibinfo {author} {\bibfnamefont {J.}~\bibnamefont {Preskill}},\ }\bibfield  {title} {\bibinfo {title} {Continuous symmetries and approximate quantum error correction},\ }\href {https://doi.org/10.1103/PhysRevX.10.041018} {\bibfield  {journal} {\bibinfo  {journal} {Phys. Rev. X}\ }\textbf {\bibinfo {volume} {10}},\ \bibinfo {pages} {041018} (\bibinfo {year} {2020})}\BibitemShut {NoStop}%
\bibitem [{\citenamefont {Gottesman}()}]{gottesman2024}%
  \BibitemOpen
  \bibfield  {author} {\bibinfo {author} {\bibfnamefont {D.}~\bibnamefont {Gottesman}},\ }\bibfield  {title} {\bibinfo {title} {Surviving as a quantum computer in a classical world {(May 2024 draft)}},\ }\bibinfo {note} {https://www.cs.umd.edu/class/spring2024/cmsc858G}\BibitemShut {NoStop}%
\bibitem [{\citenamefont {Liu}\ and\ \citenamefont {Zhou}(2023)}]{Liu2023}%
  \BibitemOpen
  \bibfield  {author} {\bibinfo {author} {\bibfnamefont {Z.-W.}\ \bibnamefont {Liu}}\ and\ \bibinfo {author} {\bibfnamefont {S.}~\bibnamefont {Zhou}},\ }\bibfield  {title} {\bibinfo {title} {Approximate symmetries and quantum error correction},\ }\href {https://doi.org/10.1038/s41534-023-00788-4} {\bibfield  {journal} {\bibinfo  {journal} {npj Quantum Information}\ }\textbf {\bibinfo {volume} {9}},\ \bibinfo {pages} {119} (\bibinfo {year} {2023})}\BibitemShut {NoStop}%
\bibitem [{\citenamefont {Akers}\ and\ \citenamefont {Penington}(2022)}]{10.21468/SciPostPhys.12.5.157}%
  \BibitemOpen
  \bibfield  {author} {\bibinfo {author} {\bibfnamefont {C.}~\bibnamefont {Akers}}\ and\ \bibinfo {author} {\bibfnamefont {G.}~\bibnamefont {Penington}},\ }\bibfield  {title} {\bibinfo {title} {{Quantum minimal surfaces from quantum error correction}},\ }\href {https://doi.org/10.21468/SciPostPhys.12.5.157} {\bibfield  {journal} {\bibinfo  {journal} {SciPost Phys.}\ }\textbf {\bibinfo {volume} {12}},\ \bibinfo {pages} {157} (\bibinfo {year} {2022})}\BibitemShut {NoStop}%
\bibitem [{\citenamefont {Fan}\ \emph {et~al.}(2023)\citenamefont {Fan}, \citenamefont {Fischler},\ and\ \citenamefont {Kubischta}}]{PhysRevA.107.032411}%
  \BibitemOpen
  \bibfield  {author} {\bibinfo {author} {\bibfnamefont {Y.}~\bibnamefont {Fan}}, \bibinfo {author} {\bibfnamefont {W.}~\bibnamefont {Fischler}},\ and\ \bibinfo {author} {\bibfnamefont {E.}~\bibnamefont {Kubischta}},\ }\bibfield  {title} {\bibinfo {title} {{Quantum error correction in the lowest Landau level}},\ }\href {https://doi.org/10.1103/PhysRevA.107.032411} {\bibfield  {journal} {\bibinfo  {journal} {Phys. Rev. A}\ }\textbf {\bibinfo {volume} {107}},\ \bibinfo {pages} {032411} (\bibinfo {year} {2023})}\BibitemShut {NoStop}%
\bibitem [{\citenamefont {Zwiebach}(2022)}]{zwiebach2022mastering}%
  \BibitemOpen
  \bibfield  {author} {\bibinfo {author} {\bibfnamefont {B.}~\bibnamefont {Zwiebach}},\ }\href@noop {} {\emph {\bibinfo {title} {Mastering quantum mechanics: essentials, theory, and applications}}}\ (\bibinfo  {publisher} {MIT Press},\ \bibinfo {year} {2022})\BibitemShut {NoStop}%
\bibitem [{\citenamefont {Ouyang}(2024)}]{Ouyang2024newtakepermutation}%
  \BibitemOpen
  \bibfield  {author} {\bibinfo {author} {\bibfnamefont {Y.}~\bibnamefont {Ouyang}},\ }\bibfield  {title} {\bibinfo {title} {A new take on permutation-invariant quantum codes},\ }\href {https://doi.org/10.22331/qv-2024-05-13-80} {\bibfield  {journal} {\bibinfo  {journal} {{Quantum Views}}\ }\textbf {\bibinfo {volume} {8}},\ \bibinfo {pages} {80} (\bibinfo {year} {2024})}\BibitemShut {NoStop}%
\bibitem [{\citenamefont {Ruskai}(2000)}]{PhysRevLett.85.194}%
  \BibitemOpen
  \bibfield  {author} {\bibinfo {author} {\bibfnamefont {M.~B.}\ \bibnamefont {Ruskai}},\ }\bibfield  {title} {\bibinfo {title} {Pauli exchange errors in quantum computation},\ }\href {https://doi.org/10.1103/PhysRevLett.85.194} {\bibfield  {journal} {\bibinfo  {journal} {Phys. Rev. Lett.}\ }\textbf {\bibinfo {volume} {85}},\ \bibinfo {pages} {194} (\bibinfo {year} {2000})}\BibitemShut {NoStop}%
\bibitem [{\citenamefont {Yadin}\ \emph {et~al.}(2023)\citenamefont {Yadin}, \citenamefont {Morris},\ and\ \citenamefont {Brandner}}]{PhysRevResearch.5.033018}%
  \BibitemOpen
  \bibfield  {author} {\bibinfo {author} {\bibfnamefont {B.}~\bibnamefont {Yadin}}, \bibinfo {author} {\bibfnamefont {B.}~\bibnamefont {Morris}},\ and\ \bibinfo {author} {\bibfnamefont {K.}~\bibnamefont {Brandner}},\ }\bibfield  {title} {\bibinfo {title} {Thermodynamics of permutation-invariant quantum many-body systems: A group-theoretical framework},\ }\href {https://doi.org/10.1103/PhysRevResearch.5.033018} {\bibfield  {journal} {\bibinfo  {journal} {Phys. Rev. Res.}\ }\textbf {\bibinfo {volume} {5}},\ \bibinfo {pages} {033018} (\bibinfo {year} {2023})}\BibitemShut {NoStop}%
\bibitem [{\citenamefont {Ouyang}(2014)}]{PhysRevA.90.062317}%
  \BibitemOpen
  \bibfield  {author} {\bibinfo {author} {\bibfnamefont {Y.}~\bibnamefont {Ouyang}},\ }\bibfield  {title} {\bibinfo {title} {Permutation-invariant quantum codes},\ }\href {https://doi.org/10.1103/PhysRevA.90.062317} {\bibfield  {journal} {\bibinfo  {journal} {Phys. Rev. A}\ }\textbf {\bibinfo {volume} {90}},\ \bibinfo {pages} {062317} (\bibinfo {year} {2014})}\BibitemShut {NoStop}%
\bibitem [{\citenamefont {Aydin}\ \emph {et~al.}(2024)\citenamefont {Aydin}, \citenamefont {Alekseyev},\ and\ \citenamefont {Barg}}]{Aydin2024familyof}%
  \BibitemOpen
  \bibfield  {author} {\bibinfo {author} {\bibfnamefont {A.}~\bibnamefont {Aydin}}, \bibinfo {author} {\bibfnamefont {M.~A.}\ \bibnamefont {Alekseyev}},\ and\ \bibinfo {author} {\bibfnamefont {A.}~\bibnamefont {Barg}},\ }\bibfield  {title} {\bibinfo {title} {A family of permutationally invariant quantum codes},\ }\href {https://doi.org/10.22331/q-2024-04-30-1321} {\bibfield  {journal} {\bibinfo  {journal} {{Quantum}}\ }\textbf {\bibinfo {volume} {8}},\ \bibinfo {pages} {1321} (\bibinfo {year} {2024})}\BibitemShut {NoStop}%
\bibitem [{\citenamefont {Ouyang}(2021)}]{PhysRevB.103.144417}%
  \BibitemOpen
  \bibfield  {author} {\bibinfo {author} {\bibfnamefont {Y.}~\bibnamefont {Ouyang}},\ }\bibfield  {title} {\bibinfo {title} {Quantum storage in quantum ferromagnets},\ }\href {https://doi.org/10.1103/PhysRevB.103.144417} {\bibfield  {journal} {\bibinfo  {journal} {Phys. Rev. B}\ }\textbf {\bibinfo {volume} {103}},\ \bibinfo {pages} {144417} (\bibinfo {year} {2021})}\BibitemShut {NoStop}%
\bibitem [{\citenamefont {Hall}(2015)}]{hall2015lie}%
  \BibitemOpen
  \bibfield  {author} {\bibinfo {author} {\bibfnamefont {B.~C.}\ \bibnamefont {Hall}},\ }\href@noop {} {\emph {\bibinfo {title} {Lie Groups, Lie Algebras, and Representations}}},\ \bibinfo {edition} {2nd}\ ed.\ (\bibinfo  {publisher} {Springer International Publishing},\ \bibinfo {year} {2015})\BibitemShut {NoStop}%
\bibitem [{\citenamefont {Georgi}(1999)}]{Georgi:1999wka}%
  \BibitemOpen
  \bibfield  {author} {\bibinfo {author} {\bibfnamefont {H.}~\bibnamefont {Georgi}},\ }\href@noop {} {\emph {\bibinfo {title} {{Lie algebras in particle physics}}}},\ \bibinfo {edition} {2nd}\ ed.,\ Vol.~\bibinfo {volume} {54}\ (\bibinfo  {publisher} {Perseus Books},\ \bibinfo {address} {Reading, MA},\ \bibinfo {year} {1999})\BibitemShut {NoStop}%
\bibitem [{\citenamefont {Fulton}\ and\ \citenamefont {Harris}(2013)}]{fulton2013representation}%
  \BibitemOpen
  \bibfield  {author} {\bibinfo {author} {\bibfnamefont {W.}~\bibnamefont {Fulton}}\ and\ \bibinfo {author} {\bibfnamefont {J.}~\bibnamefont {Harris}},\ }\href@noop {} {\emph {\bibinfo {title} {Representation theory: a first course}}},\ Vol.\ \bibinfo {volume} {129}\ (\bibinfo  {publisher} {Springer Science \& Business Media},\ \bibinfo {year} {2013})\BibitemShut {NoStop}%
\bibitem [{\citenamefont {Eastin}\ and\ \citenamefont {Knill}(2009)}]{PhysRevLett.102.110502}%
  \BibitemOpen
  \bibfield  {author} {\bibinfo {author} {\bibfnamefont {B.}~\bibnamefont {Eastin}}\ and\ \bibinfo {author} {\bibfnamefont {E.}~\bibnamefont {Knill}},\ }\bibfield  {title} {\bibinfo {title} {Restrictions on transversal encoded quantum gate sets},\ }\href {https://doi.org/10.1103/PhysRevLett.102.110502} {\bibfield  {journal} {\bibinfo  {journal} {Phys. Rev. Lett.}\ }\textbf {\bibinfo {volume} {102}},\ \bibinfo {pages} {110502} (\bibinfo {year} {2009})}\BibitemShut {NoStop}%
\bibitem [{\citenamefont {Kubischta}\ and\ \citenamefont {Teixeira}(2023)}]{PhysRevLett.131.240601}%
  \BibitemOpen
  \bibfield  {author} {\bibinfo {author} {\bibfnamefont {E.}~\bibnamefont {Kubischta}}\ and\ \bibinfo {author} {\bibfnamefont {I.}~\bibnamefont {Teixeira}},\ }\bibfield  {title} {\bibinfo {title} {Family of quantum codes with exotic transversal gates},\ }\href {https://doi.org/10.1103/PhysRevLett.131.240601} {\bibfield  {journal} {\bibinfo  {journal} {Phys. Rev. Lett.}\ }\textbf {\bibinfo {volume} {131}},\ \bibinfo {pages} {240601} (\bibinfo {year} {2023})}\BibitemShut {NoStop}%
\bibitem [{\citenamefont {Kubischta}\ and\ \citenamefont {Teixeira}(2024)}]{PhysRevLett.133.030602}%
  \BibitemOpen
  \bibfield  {author} {\bibinfo {author} {\bibfnamefont {E.}~\bibnamefont {Kubischta}}\ and\ \bibinfo {author} {\bibfnamefont {I.}~\bibnamefont {Teixeira}},\ }\bibfield  {title} {\bibinfo {title} {Quantum codes from twisted unitary $t$-groups},\ }\href {https://doi.org/10.1103/PhysRevLett.133.030602} {\bibfield  {journal} {\bibinfo  {journal} {Phys. Rev. Lett.}\ }\textbf {\bibinfo {volume} {133}},\ \bibinfo {pages} {030602} (\bibinfo {year} {2024})}\BibitemShut {NoStop}%
\bibitem [{\citenamefont {Grassberger}\ and\ \citenamefont {Hörmann}(2001)}]{Grassberger2001}%
  \BibitemOpen
  \bibfield  {author} {\bibinfo {author} {\bibfnamefont {J.}~\bibnamefont {Grassberger}}\ and\ \bibinfo {author} {\bibfnamefont {G.}~\bibnamefont {Hörmann}},\ }\bibfield  {title} {\bibinfo {title} {{A note on representations of the finite Heisenberg group and sums of greatest common divisors}},\ }\href {https://doi.org/10.46298/dmtcs.284} {\bibfield  {journal} {\bibinfo  {journal} {{Discrete Mathematics \& Theoretical Computer Science}}\ }\textbf {\bibinfo {volume} {{4}}},\ \bibinfo {pages} {91} (\bibinfo {year} {2001})}\BibitemShut {NoStop}%
\bibitem [{\citenamefont {Schulte}(2004)}]{Schulte2004}%
  \BibitemOpen
  \bibfield  {author} {\bibinfo {author} {\bibfnamefont {J.}~\bibnamefont {Schulte}},\ }\bibfield  {title} {\bibinfo {title} {Harmonic analysis on finite {H}eisenberg groups},\ }\href {https://doi.org/https://doi.org/10.1016/j.ejc.2003.10.003} {\bibfield  {journal} {\bibinfo  {journal} {European Journal of Combinatorics}\ }\textbf {\bibinfo {volume} {25}},\ \bibinfo {pages} {327} (\bibinfo {year} {2004})}\BibitemShut {NoStop}%
\bibitem [{\citenamefont {Fallbacher}(2015)}]{FALLBACHER2015229}%
  \BibitemOpen
  \bibfield  {author} {\bibinfo {author} {\bibfnamefont {M.}~\bibnamefont {Fallbacher}},\ }\bibfield  {title} {\bibinfo {title} {Breaking classical {L}ie groups to finite subgroups -- an automated approach},\ }\href {https://doi.org/https://doi.org/10.1016/j.nuclphysb.2015.07.004} {\bibfield  {journal} {\bibinfo  {journal} {Nuclear Physics B}\ }\textbf {\bibinfo {volume} {898}},\ \bibinfo {pages} {229} (\bibinfo {year} {2015})}\BibitemShut {NoStop}%
\bibitem [{\citenamefont {Nielsen}\ and\ \citenamefont {Chuang}(2010)}]{MikeandIke2010}%
  \BibitemOpen
  \bibfield  {author} {\bibinfo {author} {\bibfnamefont {M.~A.}\ \bibnamefont {Nielsen}}\ and\ \bibinfo {author} {\bibfnamefont {I.~L.}\ \bibnamefont {Chuang}},\ }\href@noop {} {\emph {\bibinfo {title} {Quantum Computation and Quantum Information: 10th Anniversary Edition}}}\ (\bibinfo  {publisher} {Cambridge University Press},\ \bibinfo {year} {2010})\BibitemShut {NoStop}%
\bibitem [{\citenamefont {Pfeifer}(2003)}]{pfeifer2003lie}%
  \BibitemOpen
  \bibfield  {author} {\bibinfo {author} {\bibfnamefont {W.}~\bibnamefont {Pfeifer}},\ }\href@noop {} {\emph {\bibinfo {title} {The Lie algebras {$\mathrm{su}(N)$}}}}\ (\bibinfo  {publisher} {Springer},\ \bibinfo {year} {2003})\BibitemShut {NoStop}%
\bibitem [{\citenamefont {Omanakuttan}\ and\ \citenamefont {Gross}(2023)}]{PhysRevA.108.022424}%
  \BibitemOpen
  \bibfield  {author} {\bibinfo {author} {\bibfnamefont {S.}~\bibnamefont {Omanakuttan}}\ and\ \bibinfo {author} {\bibfnamefont {J.~A.}\ \bibnamefont {Gross}},\ }\bibfield  {title} {\bibinfo {title} {{Multispin Clifford codes for angular momentum errors in spin systems}},\ }\href {https://doi.org/10.1103/PhysRevA.108.022424} {\bibfield  {journal} {\bibinfo  {journal} {Phys. Rev. A}\ }\textbf {\bibinfo {volume} {108}},\ \bibinfo {pages} {022424} (\bibinfo {year} {2023})}\BibitemShut {NoStop}%
\bibitem [{sup()}]{supp}%
  \BibitemOpen
  \href@noop {} {}\bibinfo {note} {See Supplemental Material at [URL will be inserted by publisher], which includes Ref. 33, for further details.}\BibitemShut {Stop}%
\bibitem [{\citenamefont {Howard}\ and\ \citenamefont {Vala}(2012)}]{PhysRevA.86.022316}%
  \BibitemOpen
  \bibfield  {author} {\bibinfo {author} {\bibfnamefont {M.}~\bibnamefont {Howard}}\ and\ \bibinfo {author} {\bibfnamefont {J.}~\bibnamefont {Vala}},\ }\bibfield  {title} {\bibinfo {title} {Qudit versions of the qubit $\ensuremath{\pi}/8$ gate},\ }\href {https://doi.org/10.1103/PhysRevA.86.022316} {\bibfield  {journal} {\bibinfo  {journal} {Phys. Rev. A}\ }\textbf {\bibinfo {volume} {86}},\ \bibinfo {pages} {022316} (\bibinfo {year} {2012})}\BibitemShut {NoStop}%
\bibitem [{\citenamefont {Campbell}\ \emph {et~al.}(2012)\citenamefont {Campbell}, \citenamefont {Anwar},\ and\ \citenamefont {Browne}}]{PhysRevX.2.041021}%
  \BibitemOpen
  \bibfield  {author} {\bibinfo {author} {\bibfnamefont {E.~T.}\ \bibnamefont {Campbell}}, \bibinfo {author} {\bibfnamefont {H.}~\bibnamefont {Anwar}},\ and\ \bibinfo {author} {\bibfnamefont {D.~E.}\ \bibnamefont {Browne}},\ }\bibfield  {title} {\bibinfo {title} {Magic-state distillation in all prime dimensions using quantum reed-muller codes},\ }\href {https://doi.org/10.1103/PhysRevX.2.041021} {\bibfield  {journal} {\bibinfo  {journal} {Phys. Rev. X}\ }\textbf {\bibinfo {volume} {2}},\ \bibinfo {pages} {041021} (\bibinfo {year} {2012})}\BibitemShut {NoStop}%
\bibitem [{\citenamefont {Nebe}\ \emph {et~al.}(2006)\citenamefont {Nebe}, \citenamefont {Rains}, \citenamefont {Sloane} \emph {et~al.}}]{nebe2006self}%
  \BibitemOpen
  \bibfield  {author} {\bibinfo {author} {\bibfnamefont {G.}~\bibnamefont {Nebe}}, \bibinfo {author} {\bibfnamefont {E.~M.}\ \bibnamefont {Rains}}, \bibinfo {author} {\bibfnamefont {N.~J.~A.}\ \bibnamefont {Sloane}}, \emph {et~al.},\ }\href@noop {} {\emph {\bibinfo {title} {Self-dual codes and invariant theory}}},\ Vol.~\bibinfo {volume} {17}\ (\bibinfo  {publisher} {Springer},\ \bibinfo {year} {2006})\BibitemShut {NoStop}%
\bibitem [{\citenamefont {Huang}\ and\ \citenamefont {Love}(2019)}]{PhysRevA.99.052307}%
  \BibitemOpen
  \bibfield  {author} {\bibinfo {author} {\bibfnamefont {Y.}~\bibnamefont {Huang}}\ and\ \bibinfo {author} {\bibfnamefont {P.}~\bibnamefont {Love}},\ }\bibfield  {title} {\bibinfo {title} {Approximate stabilizer rank and improved weak simulation of {C}lifford-dominated circuits for qudits},\ }\href {https://doi.org/10.1103/PhysRevA.99.052307} {\bibfield  {journal} {\bibinfo  {journal} {Phys. Rev. A}\ }\textbf {\bibinfo {volume} {99}},\ \bibinfo {pages} {052307} (\bibinfo {year} {2019})}\BibitemShut {NoStop}%
\bibitem [{\citenamefont {Ringbauer}\ \emph {et~al.}(2022)\citenamefont {Ringbauer}, \citenamefont {Meth}, \citenamefont {Postler}, \citenamefont {Stricker}, \citenamefont {Blatt}, \citenamefont {Schindler},\ and\ \citenamefont {Monz}}]{Ringbauer2022}%
  \BibitemOpen
  \bibfield  {author} {\bibinfo {author} {\bibfnamefont {M.}~\bibnamefont {Ringbauer}}, \bibinfo {author} {\bibfnamefont {M.}~\bibnamefont {Meth}}, \bibinfo {author} {\bibfnamefont {L.}~\bibnamefont {Postler}}, \bibinfo {author} {\bibfnamefont {R.}~\bibnamefont {Stricker}}, \bibinfo {author} {\bibfnamefont {R.}~\bibnamefont {Blatt}}, \bibinfo {author} {\bibfnamefont {P.}~\bibnamefont {Schindler}},\ and\ \bibinfo {author} {\bibfnamefont {T.}~\bibnamefont {Monz}},\ }\bibfield  {title} {\bibinfo {title} {A universal qudit quantum processor with trapped ions},\ }\href {https://doi.org/10.1038/s41567-022-01658-0} {\bibfield  {journal} {\bibinfo  {journal} {Nature Physics}\ }\textbf {\bibinfo {volume} {18}},\ \bibinfo {pages} {1053} (\bibinfo {year} {2022})}\BibitemShut {NoStop}%
\bibitem [{\citenamefont {Beverland}\ \emph {et~al.}(2016)\citenamefont {Beverland}, \citenamefont {Alagic}, \citenamefont {Martin}, \citenamefont {Koller}, \citenamefont {Rey},\ and\ \citenamefont {Gorshkov}}]{PhysRevA.93.051601}%
  \BibitemOpen
  \bibfield  {author} {\bibinfo {author} {\bibfnamefont {M.~E.}\ \bibnamefont {Beverland}}, \bibinfo {author} {\bibfnamefont {G.}~\bibnamefont {Alagic}}, \bibinfo {author} {\bibfnamefont {M.~J.}\ \bibnamefont {Martin}}, \bibinfo {author} {\bibfnamefont {A.~P.}\ \bibnamefont {Koller}}, \bibinfo {author} {\bibfnamefont {A.~M.}\ \bibnamefont {Rey}},\ and\ \bibinfo {author} {\bibfnamefont {A.~V.}\ \bibnamefont {Gorshkov}},\ }\bibfield  {title} {\bibinfo {title} {Realizing exactly solvable $\mathrm{SU}(n)$ magnets with thermal atoms},\ }\href {https://doi.org/10.1103/PhysRevA.93.051601} {\bibfield  {journal} {\bibinfo  {journal} {Phys. Rev. A}\ }\textbf {\bibinfo {volume} {93}},\ \bibinfo {pages} {051601} (\bibinfo {year} {2016})}\BibitemShut {NoStop}%
\bibitem [{\citenamefont {Zhang}\ \emph {et~al.}(2014)\citenamefont {Zhang}, \citenamefont {Bishof}, \citenamefont {Bromley}, \citenamefont {Kraus}, \citenamefont {Safronova}, \citenamefont {Zoller}, \citenamefont {Rey},\ and\ \citenamefont {Ye}}]{Zhang2014}%
  \BibitemOpen
  \bibfield  {author} {\bibinfo {author} {\bibfnamefont {X.}~\bibnamefont {Zhang}}, \bibinfo {author} {\bibfnamefont {M.}~\bibnamefont {Bishof}}, \bibinfo {author} {\bibfnamefont {S.~L.}\ \bibnamefont {Bromley}}, \bibinfo {author} {\bibfnamefont {C.~V.}\ \bibnamefont {Kraus}}, \bibinfo {author} {\bibfnamefont {M.~S.}\ \bibnamefont {Safronova}}, \bibinfo {author} {\bibfnamefont {P.}~\bibnamefont {Zoller}}, \bibinfo {author} {\bibfnamefont {A.~M.}\ \bibnamefont {Rey}},\ and\ \bibinfo {author} {\bibfnamefont {J.}~\bibnamefont {Ye}},\ }\bibfield  {title} {\bibinfo {title} {Spectroscopic observation of {$\mathrm{SU}(N)$}-symmetric interactions in {Sr} orbital magnetism},\ }\href {https://doi.org/10.1126/science.1254978} {\bibfield  {journal} {\bibinfo  {journal} {Science}\ }\textbf {\bibinfo {volume} {345}},\ \bibinfo {pages} {1467} (\bibinfo {year} {2014})},\ \Eprint {https://arxiv.org/abs/https://www.science.org/doi/pdf/10.1126/science.1254978} {https://www.science.org/doi/pdf/10.1126/science.1254978}
  \BibitemShut {NoStop}%
\bibitem [{\citenamefont {Taie}\ \emph {et~al.}(2022)\citenamefont {Taie}, \citenamefont {Ibarra-Garc{\'\i}a-Padilla}, \citenamefont {Nishizawa}, \citenamefont {Takasu}, \citenamefont {Kuno}, \citenamefont {Wei}, \citenamefont {Scalettar}, \citenamefont {Hazzard},\ and\ \citenamefont {Takahashi}}]{Taie2022}%
  \BibitemOpen
  \bibfield  {author} {\bibinfo {author} {\bibfnamefont {S.}~\bibnamefont {Taie}}, \bibinfo {author} {\bibfnamefont {E.}~\bibnamefont {Ibarra-Garc{\'\i}a-Padilla}}, \bibinfo {author} {\bibfnamefont {N.}~\bibnamefont {Nishizawa}}, \bibinfo {author} {\bibfnamefont {Y.}~\bibnamefont {Takasu}}, \bibinfo {author} {\bibfnamefont {Y.}~\bibnamefont {Kuno}}, \bibinfo {author} {\bibfnamefont {H.-T.}\ \bibnamefont {Wei}}, \bibinfo {author} {\bibfnamefont {R.~T.}\ \bibnamefont {Scalettar}}, \bibinfo {author} {\bibfnamefont {K.~R.~A.}\ \bibnamefont {Hazzard}},\ and\ \bibinfo {author} {\bibfnamefont {Y.}~\bibnamefont {Takahashi}},\ }\bibfield  {title} {\bibinfo {title} {Observation of antiferromagnetic correlations in an ultracold {$\mathrm{SU}(N)$} {H}ubbard model},\ }\href {https://doi.org/10.1038/s41567-022-01725-6} {\bibfield  {journal} {\bibinfo  {journal} {Nature Physics}\ }\textbf {\bibinfo {volume} {18}},\ \bibinfo {pages} {1356} (\bibinfo {year} {2022})}\BibitemShut {NoStop}%
\bibitem [{\citenamefont {Fern{\'a}ndez~de Fuentes}\ \emph {et~al.}(2024)\citenamefont {Fern{\'a}ndez~de Fuentes}, \citenamefont {Botzem}, \citenamefont {Johnson}, \citenamefont {Vaartjes}, \citenamefont {Asaad}, \citenamefont {Mourik}, \citenamefont {Hudson}, \citenamefont {Itoh}, \citenamefont {Johnson}, \citenamefont {Jakob}, \citenamefont {McCallum}, \citenamefont {Jamieson}, \citenamefont {Dzurak},\ and\ \citenamefont {Morello}}]{fernandezdefuentes2024}%
  \BibitemOpen
  \bibfield  {author} {\bibinfo {author} {\bibfnamefont {I.}~\bibnamefont {Fern{\'a}ndez~de Fuentes}}, \bibinfo {author} {\bibfnamefont {T.}~\bibnamefont {Botzem}}, \bibinfo {author} {\bibfnamefont {M.~A.~I.}\ \bibnamefont {Johnson}}, \bibinfo {author} {\bibfnamefont {A.}~\bibnamefont {Vaartjes}}, \bibinfo {author} {\bibfnamefont {S.}~\bibnamefont {Asaad}}, \bibinfo {author} {\bibfnamefont {V.}~\bibnamefont {Mourik}}, \bibinfo {author} {\bibfnamefont {F.~E.}\ \bibnamefont {Hudson}}, \bibinfo {author} {\bibfnamefont {K.~M.}\ \bibnamefont {Itoh}}, \bibinfo {author} {\bibfnamefont {B.~C.}\ \bibnamefont {Johnson}}, \bibinfo {author} {\bibfnamefont {A.~M.}\ \bibnamefont {Jakob}}, \bibinfo {author} {\bibfnamefont {J.~C.}\ \bibnamefont {McCallum}}, \bibinfo {author} {\bibfnamefont {D.~N.}\ \bibnamefont {Jamieson}}, \bibinfo {author} {\bibfnamefont {A.~S.}\ \bibnamefont {Dzurak}},\ and\ \bibinfo {author} {\bibfnamefont {A.}~\bibnamefont {Morello}},\ }\bibfield  {title} {\bibinfo {title} {Navigating the 16-dimensional
  hilbert space of a high-spin donor qudit with electric and magnetic fields},\ }\href {https://doi.org/10.1038/s41467-024-45368-y} {\bibfield  {journal} {\bibinfo  {journal} {Nature Communications}\ }\textbf {\bibinfo {volume} {15}},\ \bibinfo {pages} {1380} (\bibinfo {year} {2024})}\BibitemShut {NoStop}%
\bibitem [{\citenamefont {Appel}\ \emph {et~al.}(2024)\citenamefont {Appel}, \citenamefont {Ghorbal}, \citenamefont {Shofer}, \citenamefont {Zaporski}, \citenamefont {Manna}, \citenamefont {da~Silva}, \citenamefont {Haeusler}, \citenamefont {Gall}, \citenamefont {Rastelli}, \citenamefont {Gangloff},\ and\ \citenamefont {Atatüre}}]{appel2024manybodyquantumregisterspin}%
  \BibitemOpen
  \bibfield  {author} {\bibinfo {author} {\bibfnamefont {M.~H.}\ \bibnamefont {Appel}}, \bibinfo {author} {\bibfnamefont {A.}~\bibnamefont {Ghorbal}}, \bibinfo {author} {\bibfnamefont {N.}~\bibnamefont {Shofer}}, \bibinfo {author} {\bibfnamefont {L.}~\bibnamefont {Zaporski}}, \bibinfo {author} {\bibfnamefont {S.}~\bibnamefont {Manna}}, \bibinfo {author} {\bibfnamefont {S.~F.~C.}\ \bibnamefont {da~Silva}}, \bibinfo {author} {\bibfnamefont {U.}~\bibnamefont {Haeusler}}, \bibinfo {author} {\bibfnamefont {C.~L.}\ \bibnamefont {Gall}}, \bibinfo {author} {\bibfnamefont {A.}~\bibnamefont {Rastelli}}, \bibinfo {author} {\bibfnamefont {D.~A.}\ \bibnamefont {Gangloff}},\ and\ \bibinfo {author} {\bibfnamefont {M.}~\bibnamefont {Atatüre}},\ }\href {https://arxiv.org/abs/2404.19680} {\bibinfo {title} {Many-body quantum register for a spin qubit}} (\bibinfo {year} {2024}),\ \Eprint {https://arxiv.org/abs/2404.19680} {arXiv:2404.19680 [quant-ph]} \BibitemShut {NoStop}%
\end{thebibliography}%


\begin{thebibliography}{0}%
\makeatletter
\providecommand \@ifxundefined [1]{%
 \@ifx{#1\undefined}
}%
\providecommand \@ifnum [1]{%
 \ifnum #1\expandafter \@firstoftwo
 \else \expandafter \@secondoftwo
 \fi
}%
\providecommand \@ifx [1]{%
 \ifx #1\expandafter \@firstoftwo
 \else \expandafter \@secondoftwo
 \fi
}%
\providecommand \natexlab [1]{#1}%
\providecommand \enquote  [1]{``#1''}%
\providecommand \bibnamefont  [1]{#1}%
\providecommand \bibfnamefont [1]{#1}%
\providecommand \citenamefont [1]{#1}%
\providecommand \href@noop [0]{\@secondoftwo}%
\providecommand \href [0]{\begingroup \@sanitize@url \@href}%
\providecommand \@href[1]{\@@startlink{#1}\@@href}%
\providecommand \@@href[1]{\endgroup#1\@@endlink}%
\providecommand \@sanitize@url [0]{\catcode `\\12\catcode `\$12\catcode `\&12\catcode `\#12\catcode `\^12\catcode `\_12\catcode `\%12\relax}%
\providecommand \@@startlink[1]{}%
\providecommand \@@endlink[0]{}%
\providecommand \url  [0]{\begingroup\@sanitize@url \@url }%
\providecommand \@url [1]{\endgroup\@href {#1}{\urlprefix }}%
\providecommand \urlprefix  [0]{URL }%
\providecommand \Eprint [0]{\href }%
\providecommand \doibase [0]{https://doi.org/}%
\providecommand \selectlanguage [0]{\@gobble}%
\providecommand \bibinfo  [0]{\@secondoftwo}%
\providecommand \bibfield  [0]{\@secondoftwo}%
\providecommand \translation [1]{[#1]}%
\providecommand \BibitemOpen [0]{}%
\providecommand \bibitemStop [0]{}%
\providecommand \bibitemNoStop [0]{.\EOS\space}%
\providecommand \EOS [0]{\spacefactor3000\relax}%
\providecommand \BibitemShut  [1]{\csname bibitem#1\endcsname}%
\let\auto@bib@innerbib\@empty
\end{thebibliography}%

\end{document}



\title{Supplemental material for ``Qudit-based quantum error-correcting\\codes from irreducible representations of $\mathrm{SU}(d)$''}

\author{Robert Frederik Uy}
\affiliation{%
 Peterhouse, University of Cambridge, Cambridge CB2 1RD, United Kingdom
}%
\affiliation{
 Cavendish Laboratory, University of Cambridge, Cambridge CB3 0HE, United Kingdom
}%

\author{Dorian Gangloff}
\affiliation{
 Cavendish Laboratory, University of Cambridge, Cambridge CB3 0HE, United Kingdom
}%

\date{\today}

\maketitle


\section{Tensor product representations}
Many of the subsequent proofs assume knowledge about tensor product representations, so we have included some useful background information.

\vskip 12pt

Suppose $\boldsymbol{\Phi}_1 \colon G \longrightarrow \mathrm{GL}(V_1)$ and $\boldsymbol{\Phi}_2 \colon G \longrightarrow \mathrm{GL}(V_2)$ are representations of a Lie group $G$. Then the tensor product representation $\boldsymbol{\Phi}_1 \otimes \boldsymbol{\Phi}_2$ is given by
\begin{equation}
    \label{lie group tensor product rep}
    (\boldsymbol{\Phi}_1 \otimes \boldsymbol{\Phi}_2)(g) = \boldsymbol{\Phi}_1(g) \otimes \boldsymbol{\Phi}_2(g)
\end{equation}
for all $g \in G$ [33].

\vskip 12pt

Now, let $\boldsymbol{\hat{\Phi}}_1 \colon \mathfrak{g} \longrightarrow \mathrm{End}(V_1)$ and $\boldsymbol{\hat{\Phi}}_2 \colon \mathfrak{g} \longrightarrow \mathrm{End}(V_2)$ be the corresponding representations of the associated Lie algebra $\mathfrak{g}$. The tensor product representation $\boldsymbol{\hat{\Phi}}_1 \otimes \boldsymbol{\hat{\Phi}}_2$ is given by
\begin{equation}
    \label{lie algebra tensor product rep}
    (\boldsymbol{\hat{\Phi}}_1 \otimes \boldsymbol{\hat{\Phi}}_2)(\mathsf{G}) = \boldsymbol{\hat{\Phi}}_1(\mathsf{G}) \otimes \mathsf{I} + \mathsf{I} \otimes \boldsymbol{\hat{\Phi}}_2(\mathsf{G})
\end{equation}
for all $\mathsf{G} \in \mathfrak{g}$ [33].

\section{Proof of Lemma 4}
\renewcommand{\thetheorem}{4}
\begin{lemma}
    \label{Heisenberg-Weyl symmetry}
    Observe that
    \begin{align}
	\mathsf{\overline{X}}_d^\dagger \tilde{\mathsf{S}}^{(j,k)}_d \mathsf{\overline{X}}_d &= \tilde{\mathsf{S}}^{(j\ominus1,k\ominus1)}_d,\\
	\mathsf{\overline{X}}_d^\dagger \tilde{\mathsf{A}}^{(j,k)}_d \mathsf{\overline{X}}_d &= \tilde{\mathsf{A}}^{(j\ominus1,k\ominus1)}_d,\\
	\mathsf{\overline{X}}_d^\dagger \tilde{\mathsf{D}}_d^{(\ell)} \mathsf{\overline{X}}_d &= \tilde{\mathsf{D}}_d^{(\ell\ominus1)},\\
	\mathsf{\overline{Z}}_d^\dagger \tilde{\mathsf{S}}^{(j,k)}_d \mathsf{\overline{Z}}_d &= \Re\left(\zeta_d^{k-j}\right)\tilde{\mathsf{S}}^{(j,k)}_d - \Im\left(\zeta_d^{k-j}\right)\tilde{\mathsf{A}}^{(j,k)}_d,
    \end{align}
    for all $0 \leq j < k \leq d-1, 0 \leq \ell \leq d-2$. Above, $\ominus$ denotes subtraction modulo $d$.
\end{lemma}

\begin{proof}
    We first show that these equations are true for the fundamental irrep $\boldsymbol{\pi}_{d;1}^{}$. Observe that
    \begin{align}
        \mathsf{X}_d^\dagger \mathsf{S}^{(j,k)}_d \mathsf{X}_d &= \left(\sum_{p \in \mathbf{Z}_d} \ket{p}\bra{p \oplus 1}\right) \big(\ket{j}\bra{k} + \ket{k}\bra{j}\big) \left(\sum_{p \in \mathbf{Z}_d} \ket{p \oplus 1}\bra{p}\right)\\
        &= \ket{j \ominus 1}\braket{j|j}\braket{k|k}\bra{k \ominus 1} + \ket{k \ominus 1}\braket{k|k}\braket{j|j}\bra{j \ominus 1}\\
        &= \ket{j \ominus 1}\bra{k \ominus 1} + \ket{k \ominus 1}\bra{j \ominus 1}\\
        &= \mathsf{S}^{(j \ominus 1,k \ominus 1)}_d,\\
        \nonumber\\
        \mathsf{X}_d^\dagger \mathsf{A}^{(j,k)}_d \mathsf{X}_d &= \left(\sum_{p \in \mathbf{Z}_d} \ket{p}\bra{p \oplus 1}\right) \big(-i\ket{j}\bra{k} + i\ket{k}\bra{j}\big) \left(\sum_{p \in \mathbf{Z}_d} \ket{p \oplus 1}\bra{p}\right)\\
        &= -i \ket{j \ominus 1}\bra{k \ominus 1} + i\ket{k \ominus 1}\bra{j \ominus 1}\\
        &= \mathsf{A}^{(j \ominus 1,k \ominus 1)}_d,
    \end{align}
    \begin{align}
        \mathsf{X}_d^\dagger \mathsf{D}_d^{(\ell)} \mathsf{X}_d &= \left(\sum_{p \in \mathbf{Z}_d} \ket{p}\bra{p \oplus 1}\right) \big(\ket{\ell}\bra{\ell} - \ket{\ell \oplus 1}\bra{\ell \oplus 1}\big) \left(\sum_{p \in \mathbf{Z}_d} \ket{p \oplus 1}\bra{p}\right)\\
        &= \ket{\ell \ominus 1}\bra{\ell \ominus 1} - \ket{\ell}\bra{\ell}\\
        &= \mathsf{D}^{(\ell \ominus 1)}_d,\\
        \nonumber\\
        \mathsf{Z}_d^\dagger \mathsf{S}^{(j,k)}_d \mathsf{Z}_d &= \left(\sum_{p \in \mathbf{Z}_d} \zeta_d^{-p} \ket{p}\bra{p}\right) \big(\ket{j}\bra{k} + \ket{k}\bra{j}\big) \left(\sum_{p \in \mathbf{Z}_d} \zeta_d^p \ket{p}\bra{p}\right)\\
	&= \zeta_d^{k-j} \ket{j}\bra{k} + \zeta_d^{j-k}\ket{k}\bra{j}\\
	&= \Re(\zeta_d^{k-j}) \big(\ket{j}\bra{k} + \ket{k}\bra{j}\big) - \Im(\zeta_d^{k-j}) \big(-i\ket{j}\bra{k} + i\ket{k}\bra{j}\big)\\
	&= \Re\left(\zeta_d^{k-j}\right)\tilde{\mathsf{S}}^{(j,k)}_d - \Im\left(\zeta_d^{k-j}\right)\tilde{\mathsf{A}}^{(j,k)}_d,
    \end{align}
    as claimed.
    
    \vskip 12pt
    
    We then proceed by applying Eqs. (\ref{lie group tensor product rep}) and (\ref{lie algebra tensor product rep}). Notice that, since $\mathsf{X}_d$ and $\mathsf{Z}_d$ are unitary, we have
    \begin{align}
        \mathsf{\overline{X}}_d^\dagger \tilde{\mathsf{S}}^{(j,k)}_d \mathsf{\overline{X}}_d &= \left(\mathsf{X}_d^\dagger\right)^{\otimes N} \left( \sum_{i=0}^{N-1} \mathsf{I}^{\otimes i} \otimes \mathsf{S}^{(j,k)} \otimes \mathsf{I}^{\otimes (N-1-i)} \right) \mathsf{X}_d^{\otimes N}\\
        &= \sum_{i=0}^{N-1} \mathsf{I}^{\otimes i} \otimes \left(\mathsf{X}_d^\dagger\mathsf{S}^{(j,k)}\mathsf{X}_d\right) \otimes \mathsf{I}^{\otimes (N-1-i)}\\
        &= \sum_{i=0}^{N-1} \mathsf{I}^{\otimes i} \otimes \mathsf{S}^{(j \ominus 1,k \ominus 1)}_d \otimes \mathsf{I}^{\otimes (N-1-i)}\\
        &= \mathsf{\tilde{S}}^{(j \ominus 1,k \ominus 1)}_d,\\
        \nonumber\\
        \mathsf{\overline{X}}_d^\dagger \tilde{\mathsf{A}}^{(j,k)}_d \mathsf{\overline{X}}_d &= \left(\mathsf{X}_d^\dagger\right)^{\otimes N} \left( \sum_{i=0}^{N-1} \mathsf{I}^{\otimes i} \otimes \mathsf{A}^{(j,k)} \otimes \mathsf{I}^{\otimes (N-1-i)} \right) \mathsf{X}_d^{\otimes N}\\
        &= \sum_{i=0}^{N-1} \mathsf{I}^{\otimes i} \otimes \left(\mathsf{X}_d^\dagger\mathsf{A}^{(j,k)}\mathsf{X}_d\right) \otimes \mathsf{I}^{\otimes (N-1-i)}\\
        &= \sum_{i=0}^{N-1} \mathsf{I}^{\otimes i} \otimes \mathsf{A}^{(j \ominus 1,k \ominus 1)}_d \otimes \mathsf{I}^{\otimes (N-1-i)}\\
        &= \mathsf{\tilde{A}}^{(j \ominus 1,k \ominus 1)}_d,\\
        \nonumber\\
        \mathsf{\overline{X}}_d^\dagger \mathsf{\tilde{D}}_d^{(\ell)} \mathsf{\overline{X}}_d &= \left(\mathsf{X}_d^\dagger\right)^{\otimes N} \left( \sum_{i=0}^{N-1} \mathsf{I}^{\otimes i} \otimes \mathsf{D}_d^{(\ell)} \otimes \mathsf{I}^{\otimes (N-1-i)} \right) \mathsf{X}_d^{\otimes N}\\
        &= \sum_{i=0}^{N-1} \mathsf{I}^{\otimes i} \otimes \left(\mathsf{X}_d^\dagger\mathsf{D}_d^{(\ell)}\mathsf{X}_d\right) \otimes \mathsf{I}^{\otimes (N-1-i)}\\
        &= \sum_{i=0}^{N-1} \mathsf{I}^{\otimes i} \otimes \mathsf{D}_d^{(\ell \ominus 1)} \otimes \mathsf{I}^{\otimes (N-1-i)}\\
        &= \mathsf{\tilde{D}}_d^{(\ell \ominus 1)},
    \end{align}
    \begin{align}
        \mathsf{\overline{Z}}_d^\dagger \mathsf{\tilde{S}}_d^{(j,k)} \mathsf{\overline{Z}}_d &= \left(\mathsf{Z}_d^\dagger\right)^{\otimes N} \left( \sum_{i=0}^{N-1} \mathsf{I}^{\otimes i} \otimes \mathsf{S}^{(j,k)}_d \otimes \mathsf{I}^{\otimes (N-1-i)} \right) \mathsf{Z}_d^{\otimes N}\\
        &= \sum_{i=0}^{N-1} \mathsf{I}^{\otimes i} \otimes \left(\mathsf{Z}_d^\dagger\mathsf{S}^{(j,k)}_d\mathsf{Z}_d\right) \otimes \mathsf{I}^{\otimes (N-1-i)}\\
        &= \sum_{i=0}^{N-1} \mathsf{I}^{\otimes i} \otimes \left(\Re\left(\zeta_d^{k-j}\right)\mathsf{S}^{(j,k)}_d - \Im\left(\zeta_d^{k-j}\right)\mathsf{A}^{(j,k)}_d\right) \otimes \mathsf{I}^{\otimes (N-1-i)}\\
        &= \Re\left(\zeta_d^{k-j}\right)\mathsf{\tilde{S}}^{(j,k)}_d - \Im\left(\zeta_d^{k-j}\right)\mathsf{\tilde{A}}^{(j,k)}_d,
    \end{align}
    as claimed.
\end{proof}

\section{Proof of Theorem 5}
In addition to Lemma 4 in the main text, the proof of Theorem 5 will require extensive use of the following lemmas:
\renewcommand{\thetheorem}{S1}
\begin{lemma}
    \label{orthogonality of Zd eigenspaces}
    Suppose $\mathbf{u},\mathbf{v} \in \mathrm{Sym}^N(\mathbf{C}^d)$ satisfy $\mathsf{\overline{Z}}_d \mathbf{u} = \zeta_d^{w_\mathbf{u}} \mathbf{u}$ and $\mathsf{\overline{Z}}_d \mathbf{v} = \zeta_d^{w_\mathbf{v}} \mathbf{v}$, with $\zeta_d^{w_\mathbf{u}} \neq \zeta_d^{w_\mathbf{v}}$. Then
    \begin{equation}
        \mathbf{u}^\dagger \mathbf{v} = 0.
    \end{equation}
\end{lemma}

\begin{proof}
	Observe that
	\begin{align}
		\mathsf{\overline{Z}}_d \mathbf{u} = \zeta_d^{w_\mathbf{u}} \mathbf{u} \quad\Longrightarrow\quad \mathbf{u}^\dagger\mathsf{\overline{Z}}_d\mathbf{v} = \zeta_d^{w_\mathbf{u}} \mathbf{u}^\dagger \mathbf{v}\\
		\mathsf{\overline{Z}}_d \mathbf{v} = \zeta_d^{w_\mathbf{v}} \mathbf{v} \quad\Longrightarrow\quad \mathbf{u}^\dagger\mathsf{\overline{Z}}_d\mathbf{v} = \zeta_d^{w_\mathbf{v}} \mathbf{u}^\dagger \mathbf{v}
	\end{align}
	Thus, $(\zeta_d^{w_\mathbf{u}} - \zeta_d^{w_\mathbf{v}}) \mathbf{u}^\dagger \mathbf{v} = 0$, which implies that $\mathbf{u}^\dagger \mathbf{v} = 0$ since $\zeta_d^{w_\mathbf{u}} \neq \zeta_d^{w_\mathbf{v}}$.
\end{proof}

\vskip 12pt

\renewcommand{\thetheorem}{S2}
\begin{lemma}
    Let $\ket{S_\mathbf{u}} \in \mathrm{Sym}^N(\mathbf{C}^d)$. Observe that for all $0 \leq p < q \leq d-1$ and all $0 \leq \ell \leq d-2$,
    \begin{align}
        \mathsf{\tilde{S}}_d^{(p,q)} \ket{S_\mathbf{u}} &= (u_p+1)\ket{S_{\mathbf{v}}} + (u_q+1)\ket{S_{\mathbf{w}}},\\
        \mathsf{\tilde{D}}_d^{(\ell)} \ket{S_\mathbf{u}} &= (u_\ell - u_{\ell+1}) \ket{S_\mathbf{u}},
    \end{align}
    where $\mathbf{v} = (u_0,\cdots,u_p+1,\cdots,u_q-1,\cdots,u_{d-1})$ and $\mathbf{w} = (u_0,\cdots,u_p-1,\cdots,u_q+1,\cdots,u_{d-1})$.
\end{lemma}

[\textit{We adopt the convention that $\ket{S_\mathbf{t}} = \mathbf{0}$ if there exists $k$ such that $t_k < 0$ or $t_k > N$.}]

\begin{proof}
    Note that for each $i \in \mathbf{Z}_N$,
    \begin{align}
        \left(\mathsf{I}^{\otimes i} \otimes \mathsf{S}_d^{(p,q)} \otimes \mathsf{I}^{\otimes (N-1-i)}\right) \ket{S_{\mathbf{u}}} &= \sum_{\mathbf{u}' \in [\mathbf{u}]_\mathcal{R}^{}} \left(\mathsf{I}^{\otimes i} \otimes \mathsf{S}_d^{(p,q)} \otimes \mathsf{I}^{\otimes (N-1-i)}\right) \left(\ket{u'_0} \otimes \cdots \otimes \ket{u'_{d-1}}\right)\\
        &= \sum_{\mathbf{u}' \in [\mathbf{u}]_\mathcal{R}^{}, u'_i=p} \left(\mathsf{I}^{\otimes i} \otimes \mathsf{S}_d^{(p,q)} \otimes \mathsf{I}^{\otimes (N-1-i)}\right) \left(\ket{u'_0} \otimes \cdots \otimes \ket{p} \otimes \cdots \otimes \ket{u'_{d-1}}\right)\nonumber\\
        &\qquad + \sum_{\mathbf{u}' \in [\mathbf{u}]_\mathcal{R}^{}, u'_i=q} \left(\mathsf{I}^{\otimes i} \otimes \mathsf{S}_d^{(p,q)} \otimes \mathsf{I}^{\otimes (N-1-i)}\right) \left(\ket{u'_0} \otimes \cdots \otimes \ket{q} \otimes \cdots \otimes \ket{u'_{d-1}}\right)\\
        &= \sum_{\mathbf{u}' \in [\mathbf{u}]_\mathcal{R}^{}, u'_i=p} \ket{u'_0} \otimes \cdots \otimes \ket{q} \otimes \cdots \otimes \ket{u'_{d-1}}\nonumber\\
        &\qquad + \sum_{\mathbf{u}' \in [\mathbf{u}]_\mathcal{R}^{}, u'_i=q} \ket{u'_0} \otimes \cdots \otimes \ket{p} \otimes \cdots \otimes \ket{u'_{d-1}}\\
        &= \sum_{\mathbf{w}' \in [\mathbf{w}]_\mathcal{R}^{}, w'_i=q} \ket{w'_0} \otimes \cdots \otimes \ket{q} \otimes \cdots \otimes \ket{w'_{d-1}}\nonumber\\
        &\qquad + \sum_{\mathbf{v}' \in [\mathbf{v}]_\mathcal{R}^{}, v'_i=p} \ket{v'_0} \otimes \cdots \otimes \ket{p} \otimes \cdots \otimes \ket{v'_{d-1}},
    \end{align}
    Thus, we have
    \begin{align}
        \mathsf{\tilde{S}}_d^{(p,q)}\ket{S_\mathbf{u}} &= \sum_{i\in\mathbf{Z}_N} \left(\sum_{\mathbf{w}' \in [\mathbf{w}]_\mathcal{R}^{}, w'_i=q} \ket{w'_0} \otimes \cdots \otimes \ket{q} \otimes \cdots \otimes \ket{w'_{d-1}} + \sum_{\mathbf{v}' \in [\mathbf{v}]_\mathcal{R}^{}, v'_i=p} \ket{v'_0} \otimes \cdots \otimes \ket{p} \otimes \cdots \otimes \ket{v'_{d-1}}\right)\\
        &= v_p \sum_{\mathbf{v}' \in [\mathbf{v}]_\mathcal{R}^{}} \ket{v'_0} \otimes \cdots \otimes \ket{v'_{d-1}} + w_q \sum_{\mathbf{w}' \in [\mathbf{w}]_\mathcal{R}^{}} \ket{w'_0} \otimes \cdots \otimes \ket{w'_{d-1}}\\
        &= (u_p+1)\ket{S_\mathbf{v}} + (u_q+1)\ket{S_\mathbf{w}},
    \end{align}
    where we have used the fact that, for each $\mathbf{v}' \in [\mathbf{v}]_\mathcal{R}^{}$, the tensor product $\ket{v'_0} \otimes \cdots \otimes \ket{v'_{d-1}}$ was counted $v_p$ times in the sum
    \begin{equation}
        \sum_{i\in\mathbf{Z}_N} \sum_{\mathbf{v}' \in [\mathbf{v}]_\mathcal{R}^{}, v'_i=p} \ket{v'_0} \otimes \cdots \otimes \ket{p} \otimes \cdots \otimes \ket{v'_{d-1}}.
    \end{equation}

    \vskip 12pt

   Similarly, we have
    \begin{align}
        \mathsf{\tilde{D}}_d^{(\ell)} \ket{S_\mathbf{u}} &= \left(\sum_{i \in \mathbf{Z}_N} \mathsf{I}^{\otimes i} \otimes \mathsf{D}_d^{(\ell)} \otimes \mathsf{I}^{\otimes (N-1-i)}\right) \left(\sum_{\mathbf{u}' \in [\mathbf{u}]_\mathcal{R}^{}} \ket{u'_0} \otimes \cdots \otimes \ket{u'_{d-1}}\right)\\
        &= \sum_{i \in \mathbf{Z}_N} \left( \sum_{\mathbf{u}' \in [\mathbf{u}]_\mathcal{R}^{},u'_i = \ell} \ket{u'_0} \otimes \cdots \otimes \ket{\ell} \otimes \cdots \otimes \ket{u'_{d-1}} - \sum_{\mathbf{u}' \in [\mathbf{u}]_\mathcal{R}^{},u'_i = \ell+1} \ket{u'_0} \otimes \cdots \otimes \ket{\ell+1} \otimes \cdots \otimes \ket{u'_{d-1}} \right)\\
        &= u_\ell \sum_{\mathbf{u}' \in [\mathbf{u}]_\mathcal{R}^{}} \ket{u'_0} \otimes \cdots \otimes \ket{u'_{d-1}} - u_{\ell+1} \sum_{\mathbf{u}' \in [\mathbf{u}]_\mathcal{R}^{}} \ket{u'_0} \otimes \cdots \otimes \ket{u'_{d-1}}\\
        &= (u_\ell - u_{\ell+1})\ket{S_\mathbf{u}},
    \end{align}
    as claimed.
\end{proof}

\vskip 12pt

It is also helpful to delineate the  Knill-Laflamme conditions in our particular case:
\begin{enumerate}[\rm (KL1),leftmargin=3em,itemsep=2pt]
    \item $\bra{\overline{i}} \mathsf{\tilde{S}}^{(p,q)}_d \ket{\overline{j}} = c_{(\mathsf{I},\mathsf{\tilde{S}}^{(p,q)}_d)}\delta_{ij}$ for all $i,j \in \mathbf{Z}_d$ and $0 \leq p < q \leq d-1$
    \item $\bra{\overline{i}} \mathsf{\tilde{A}}^{(p,q)}_d \ket{\overline{j}} = c_{(\mathsf{I},\mathsf{\tilde{A}}^{(p,q)}_d)}\delta_{ij}$ for all $i,j \in \mathbf{Z}_d$ and $0 \leq p < q \leq d-1$
    \item $\bra{\overline{i}} \mathsf{\tilde{D}}^{(\ell)}_d \ket{\overline{j}} = c_{(\mathsf{I},\mathsf{\tilde{D}}^{(\ell)}_d)}\delta_{ij}$ for all $i,j \in \mathbf{Z}_d$ and $0 \leq \ell \leq d-2$
    \item $\bra{\overline{i}} \tilde{\Gamma}_d^{(p,q)} \tilde{\Upsilon}_d^{(r,s)} \ket{\overline{j}} = c_{(\tilde{\Gamma}_d^{(p,q)},\tilde{\Upsilon}_d^{(r,s)})}\delta_{ij}$ for all $i,j \in \mathbf{Z}_d$, $\Gamma,\Upsilon \in \{\mathsf{S},\mathsf{A}\}$, $0 \leq p < q \leq d-1$, and $0 \leq r < s \leq d-1$
    \item $\bra{\overline{i}} \tilde{\Lambda} \mathsf{\tilde{D}}^{(\ell)}_d \ket{\overline{j}} = c_{(\tilde{\Lambda},\mathsf{\tilde{D}}^{(\ell)}_d)}\delta_{ij}$ for all $i,j \in \mathbf{Z}_d$, $\tilde{\Lambda} \in \left\{ \mathsf{\tilde{S}}^{(p,q)}_d, \mathsf{\tilde{A}}^{(p,q)}_d \mid 0 \leq p < q \leq d-1 \right\}$, and $\ell \in \{0, \cdots, d-2\}$
    \item $\bra{\overline{i}} \mathsf{\tilde{D}}^{(\ell)}_d \tilde{\Lambda} \ket{\overline{j}} = c_{(\mathsf{\tilde{D}}^{(\ell)}_d,\tilde{\Lambda})}\delta_{ij}$ for all $i,j \in \mathbf{Z}_d$, $\tilde{\Lambda} \in \left\{ \mathsf{\tilde{S}}^{(p,q)}_d, \mathsf{\tilde{A}}^{(p,q)}_d \mid 0 \leq p < q \leq d-1 \right\}$, and $\ell \in \{0, \cdots, d-2\}$
    \item $\bra{\overline{i}} \mathsf{\tilde{D}}^{(p)}_d \mathsf{\tilde{D}}^{(q)}_d \ket{\overline{j}} = c_{(\mathsf{\tilde{D}}^{(p)}_d, \mathsf{\tilde{D}}^{(q)}_d)}\delta_{ij}$ for all $i,j \in \mathbf{Z}_d$ and $p,q \in \{0,\cdots,d-2\}$
\end{enumerate}
Above, we have already used the fact that the elements of $\mathcal{E}$ are Hermitian. [\textit{Here, we adopt physicists' definition of a Lie algebra, which differs from that of mathematicians by a factor of $i$.}] There is, therefore, no need to separately consider redundant conditions like $\bra{\overline{i}} \mathsf{\tilde{S}}^{(p,q)\dagger}_d \ket{\overline{j}} = 0$.

\vskip 12pt

Now, let us proceed with stating and proving Theorem 5:

\renewcommand{\thetheorem}{5}
\begin{theorem}
    A code $\mathcal{C} \colon \mathbf{C}^d \longrightarrow \mathrm{Sym}^N(\mathbf{C}^d)$ is a QECC if it satisfies the following conditions:
    \begin{enumerate}[itemsep=2pt]
        \item[\rm (C1)] $\bra{\overline{i}} \mathsf{\tilde{S}}^{(0,n)}_d \ket{\overline{j}} = 0$ for all $i,j \in \mathbf{Z}_d$ with $i \neq j$ and $n \in \left\{1,\cdots,\frac{d-1}{2}\right\}$
        \item[\rm (C2)] $\bra{\overline{i}} \tilde{\Gamma}_d^{(p,q)} \tilde{\Upsilon}_d^{(r,s)} \ket{\overline{j}} = c_{(\tilde{\Gamma}_d^{(p,q)},\tilde{\Upsilon}_d^{(r,s)})}\delta_{ij}$ for all $i,j \in \mathbf{Z}_d$, $\Gamma,\Upsilon \in \{\mathsf{S},\mathsf{A}\}$, $0 \leq p < q \leq d-1$, and $0 \leq r < s \leq d-1$
        \item[\rm (C3)] $\bra{\overline{k}} \mathsf{\tilde{D}}_d^{(d-2)} \ket{\overline{k}} = c_{(\mathsf{I}, \mathsf{\tilde{D}}^{(d-2)}_d)}$ for all $k \in \mathbf{Z}_d$
        \item[\rm (C4)] $\bra{\overline{k}} \mathsf{\tilde{D}}^{(\ell)}_d \mathsf{\tilde{D}}^{(d-2)}_d \ket{\overline{k}} = c_{(\mathsf{\tilde{D}}^{(\ell)}_d, \mathsf{\tilde{D}}^{(d-2)}_d)}$ for all $k \in \mathbf{Z}_d$ and $\ell \in \{0, \cdots, d-2\}$
    \end{enumerate}
    for some set of constants $\left\{ c_{(\mathsf{E},\mathsf{F})} \mid \mathsf{E},\mathsf{F} \in \mathcal{E} \right\}$.
\end{theorem}

\begin{proof}
    We systematically prove each KL condition one at a time:

    \vskip 6pt
    
    \paragraph{Condition KL1} We first show that $\bra{\overline{k}} \mathsf{\tilde{S}}^{(0,n)}_d \ket{\overline{k}} = 0$ for all $k \in \mathbf{Z}_d$ and $n \in \left\{1,\cdots,\frac{d-1}{2}\right\}$. Fix $k' \in \mathbf{Z}_d$ and $n' \in \left\{1,\cdots,\frac{d-1}{2}\right\}$. Observe that for any $\mathbf{u} \in \mathcal{U}_{d,N,k'}$,
    \begin{align}
       \mathsf{\overline{Z}}_d\ket{S_{(u_0+1,\cdots,u_{n'}-1,\cdots,u_{d-1})}} &= \zeta_d^{k'-n'}\ket{S_{(u_0+1,\cdots,u_{n'}-1,\cdots,u_{d-1})}},\\
       \mathsf{\overline{Z}}_d\ket{S_{(u_0-1,\cdots,u_{n'}+1,\cdots,u_{d-1})}} &= \zeta_d^{k'+n'}\ket{S_{(u_0-1,\cdots,u_{n'}+1,\cdots,u_{d-1})}}.
    \end{align}
    Since $d \nmid n'$, it follows that $\zeta_d^{k'} \not\in \left\{\zeta_d^{k'-n'},\zeta_d^{k'+n'}\right\}$. By Lemma S1 and Lemma S2, we have $\bra{\overline{k'}} \mathsf{\tilde{S}}_d^{(0,n')} \ket{S_\mathbf{u}} = 0$ for all $\mathbf{u} \in \mathcal{U}_{d,N,k'}$. Since $\ket{\overline{k'}} \in \mathrm{span}_\mathbf{C} \left\{ \ket{S_\mathbf{u}} \mid \mathbf{u} \in \mathcal{U}_{d,N,k'} \right\}$, this implies that $\bra{\overline{k'}} \mathsf{\tilde{S}}_d^{(0,n')} \ket{\overline{k'}} = 0$, as claimed.

    \vskip 16pt

    Combining the above result with Hypothesis C1, we have
    \begin{equation}
        \label{C1 intermediate step 1}
        \bra{\overline{i}} \mathsf{\tilde{S}}^{(0,n)}_d \ket{\overline{j}} = 0
    \end{equation}
    for all $i,j \in \mathbf{Z}_d, n \in \left\{1,\cdots,\frac{d-1}{2}\right\}$.

    \vskip 16pt

    To get from this to Condition KL1, we use the Heisenberg-Weyl symmetry of $\mathfrak{su}(d)$. Fix $i',j' \in \mathbf{Z}_d$ and $0 \leq p' < q' \leq d-1$. By Lemma 4, we have
    \begin{equation}
        \mathsf{\tilde{S}}_d^{(p',q')} = \mathsf{\overline{X}}_d^{p'} \mathsf{\tilde{S}}_d^{(0,q'-p')} \left(\mathsf{\overline{X}}_d^\dagger\right)^{p'} = \mathsf{\overline{X}}_d^{q'} \mathsf{\tilde{S}}_d^{(d-q'+p',0)} \left(\mathsf{\overline{X}}_d^\dagger\right)^{q'} = \mathsf{\overline{X}}_d^{q'} \mathsf{\tilde{S}}_d^{(0,d-q'+p')} \left(\mathsf{\overline{X}}_d^\dagger\right)^{q'},
    \end{equation}
    where we have used the symmetry of $\mathsf{\tilde{S}}_d^{(d-q'+p',0)}$ in the last equality. Thus,
    \begin{equation}
        \bra{\overline{i'}} \mathsf{\tilde{S}}_d^{(p',q')} \ket{\overline{j'}} = \bra{\overline{i' \ominus p'}} \mathsf{\tilde{S}}_d^{(0,q'-p')} \ket{\overline{j' \ominus p'}} = \bra{\overline{i' \ominus q'}} \mathsf{\tilde{S}}_d^{(0,d-q'+p')} \ket{\overline{j' \ominus q'}}
    \end{equation}
    Note that either $q'-p' \in \left\{1,\cdots,\frac{d-1}{2}\right\}$ or $d - q' + p' \in \left\{1,\cdots,\frac{d-1}{2}\right\}$. By Eq. (\ref{C1 intermediate step 1}), we are thus guaranteed that $\bra{\overline{i'}} \mathsf{\tilde{S}}_d^{(p',q')} \ket{\overline{j'}} = 0$.

    \vskip 16pt
    
    \paragraph{Condition KL2}
    Fix $i',j' \in \mathbf{Z}_d$ and $0 \leq p' < q' \leq d-1$. By Condition KL1, we have $\bra{\overline{i'}} \mathsf{\tilde{S}}^{(p',q')}_d \ket{\overline{j'}} = 0$. Observe that by Lemma 4,
    \begin{align}
        \bra{\overline{i'}} \mathsf{\tilde{A}}^{(p',q')}_d \ket{\overline{j'}} &= \frac{1}{\Im\left(\zeta_d^{q'-p'}\right)} \left(\Re\left(\zeta_d^{q'-p'}\right) \bra{\overline{i'}} \mathsf{\tilde{S}}^{(p',q')}_d \ket{\overline{j'}} - \bra{\overline{i'}} \mathsf{\overline{Z}}_d^\dagger \mathsf{\tilde{S}}^{(p',q')}_d \mathsf{\overline{Z}}_d \ket{\overline{j'}} \right)\\
        &= \frac{1}{\Im\left(\zeta_d^{q'-p'}\right)} \left(\Re\left(\zeta_d^{q'-p'}\right) \bra{\overline{i'}} \mathsf{\tilde{S}}^{(p',q')}_d \ket{\overline{j'}} - \zeta_d^{j'-i'}\bra{\overline{i'}} \mathsf{\tilde{S}}^{(p',q')}_d \ket{\overline{j'}} \right)\\
        &= 0.
    \end{align}
    Note that $\Im\left(\zeta_d^{q'-p'}\right) \neq 0$ since $d$ is odd and $p' \neq q'$.

    \vskip 16pt

    \paragraph{Condition KL3}
    We first show that $\bra{\overline{i}} \mathsf{\tilde{D}}_d^{(\ell)} \ket{\overline{j}} = 0$ for all $i,j \in \mathbf{Z}_d, i \neq j$ and $\ell \in \left\{0,\cdots,d-2\right\}$. Fix $i',j' \in \mathbf{Z}_d, i' \neq j'$ and $\ell' \in \left\{0,\cdots,d-2\right\}$. Note that $\mathsf{\overline{Z}}_d \mathsf{\tilde{D}}_d^{(\ell')} \ket{\overline{j'}} = \zeta_d^{j'} \mathsf{\tilde{D}}_d^{(\ell')} \ket{\overline{j'}}$. Since $\zeta_d^{i'} \neq \zeta_d^{j'}$, we can use Lemma S1 to deduce that $\bra{\overline{i'}} \mathsf{\tilde{D}}_d^{(\ell')} \ket{\overline{j'}} = 0$.

    \vskip 16pt

    Now, it remains to prove that Hypothesis C3 implies that $\bra{\overline{k}} \mathsf{\tilde{D}}_d^{(\ell)} \ket{\overline{k}} = c_{(\mathsf{I},\mathsf{\tilde{D}}_d^{(\ell)})}$ for all $k \in \mathbf{Z}_d$ and $\ell \in \{0,\cdots,d-2\}$. Fix $k' \in \mathbf{Z}_d$ and $\ell' \in \{0,\cdots,d-2\}$. Using Lemma 4, we obtain
    \begin{equation}
        \bra{\overline{k'}} \mathsf{\tilde{D}}_d^{(\ell')} \ket{\overline{k'}} = \bra{\overline{k'}} \left(\mathsf{\overline{X}}_d^\dagger\right)^{d-2-\ell'} \mathsf{\tilde{D}}_d^{(d-2)} \mathsf{\overline{X}}_d^{d-2-\ell'}\ket{\overline{k'}} = \bra{\overline{(d-2-\ell') \oplus k'}} \mathsf{\tilde{D}}_d^{(d-2)} \ket{\overline{(d-2-\ell') \oplus k'}} = c_{(\mathsf{I}, \mathsf{\tilde{D}}^{(d-2)}_d)}.
    \end{equation}

    \vskip 16pt

    \paragraph{Condition KL4}
    This is identical to Hypothesis C2.

    \vskip 16pt

    \paragraph{Condition KL5}
    Fix $i',j' \in \mathbf{Z}_d, \tilde{\Lambda}' \in \left\{ \mathsf{\tilde{S}}^{(p,q)}_d, \mathsf{\tilde{A}}^{(p,q)}_d \mid 0 \leq p < q \leq d-1 \right\}, \ell' \in \{0,\cdots,d-2\}$. Observe that $\mathsf{\overline{Z}}_d \mathsf{\tilde{D}}_d^{\ell'} \ket{\overline{j'}} = \zeta_d^{j'} \mathsf{\tilde{D}}_d^{(\ell')} \ket{\overline{j'}}$, from which we can deduce by Lemma \ref{orthogonality of Zd eigenspaces} that
    \begin{equation}
        \bra{\overline{k}} \mathsf{\tilde{D}}_d^{(\ell')} \ket{\overline{j'}} = 0 \iff k = j'.
    \end{equation}
    We proceed by using the completeness relation:
    \begin{equation}
    	\label{proof of KL5}
        \bra{\overline{i'}} \tilde{\Lambda}' \mathsf{\tilde{D}}_d^{(\ell')} \ket{\overline{j'}} = \bra{\overline{i'}} \tilde{\Lambda}' \left(\sum_{k \in \mathbf{Z}_d} \ket{\overline{k}}\bra{\overline{k}}\right) \mathsf{\tilde{D}}_d^{(\ell')} \ket{\overline{j'}} = \bra{\overline{i'}} \tilde{\Lambda}' \ket{\overline{j'}} \bra{\overline{j'}} \mathsf{\tilde{D}}^{(\ell')}_d \ket{\overline{j'}} = 0,
    \end{equation}
    where the last line follows from the form of Conditions KL1 and KL2 proven above, i.e. $\bra{\overline{i}} \tilde{\Lambda} \ket{\overline{j}} = 0$ for all $i,j \in \mathbf{Z}_d$ and $\tilde{\Lambda} \in \left\{ \mathsf{\tilde{S}}^{(p,q)}_d, \mathsf{\tilde{A}}^{(p,q)}_d \mid 0 \leq p < q \leq d-1 \right\}$.

    \vskip 16pt

    \paragraph{Condition KL6}
    The proof for this condition is very similar to that for Condition KL5. Notice that, even after swapping the order of $\tilde{\Lambda}'$ and $\mathsf{\tilde{D}}_d^{(\ell')}$ in Eq. (\ref{proof of KL5}), we still get 0.

    \vskip 16pt
    
    \paragraph{Condition KL7}
    First, we show that $\bra{\overline{i}} \mathsf{\tilde{D}}_d^{(p)} \mathsf{\tilde{D}}_d^{(q)} \ket{\overline{j}} = 0$ for all $i,j\in\mathbf{Z}_d$ and $p,q \in \{0,\cdots,d-2\}$. Fix $i',j' \in \mathbf{Z}_d, i \neq j$ and $p',q' \in \{0,\cdots,d-2\}$. Noting that $\mathsf{\overline{Z}}_d \mathsf{\tilde{D}}_d^{(p')} \ket{\overline{i'}} = \zeta_d^{i'} \mathsf{\tilde{D}}_d^{(p')} \ket{\overline{i'}}$ and $\mathsf{\overline{Z}}_d \mathsf{\tilde{D}}_d^{(q')} \ket{\overline{j'}} = \zeta_d^{j'} \mathsf{\tilde{D}}_d^{(q')} \ket{\overline{j'}}$. Since $i' \neq j'$, it follows from Lemma \ref{orthogonality of Zd eigenspaces} that
    \begin{equation}
        \bra{\overline{i'}} \mathsf{\tilde{D}}_d^{(p')} \mathsf{\tilde{D}}_d^{(q')} \ket{\overline{j'}} = 0.
    \end{equation}

    \vskip 16pt
    
    Next, we prove that $\bra{\overline{k}} \mathsf{\tilde{D}}_d^{(p)} \mathsf{\tilde{D}}_d^{(q)} \ket{\overline{k}} = 0$ for all $k \in \mathbf{Z}_d, 0 \leq p \leq q \leq d-2$. Fix $k' \in \mathbf{Z}_d, 0 \leq p' \leq q' \leq d-2$. By Lemma 4, we have
    \begin{align}
        \mathsf{\tilde{D}}_d^{(p')} \mathsf{\tilde{D}}_d^{(q')} &= \left(\mathsf{\overline{X}}_d^\dagger\right)^{d-2-q'} \mathsf{\tilde{D}}_d^{(d-2-q'+p')} \mathsf{\overline{X}}_d^{d-2-q'} \left(\mathsf{\overline{X}}_d^\dagger\right)^{d-2-q'} \mathsf{\tilde{D}}_d^{(d-2)} \mathsf{\overline{X}}_d^{d-2-q'}\\
        &= \left(\mathsf{\overline{X}}_d^\dagger\right)^{d-2-q'}\mathsf{\tilde{D}}_d^{(d-2-q'+p')} \mathsf{\tilde{D}}_d^{(d-2)} \mathsf{\overline{X}}_d^{d-2-q'}.
    \end{align}
    We then calculate
    \begin{equation}
        \bra{\overline{k'}} \mathsf{\tilde{D}}_d^{(p')} \mathsf{\tilde{D}}_d^{(q')} \ket{\overline{k'}} = \bra{\overline{k' \oplus (d-2-q')}} \mathsf{\tilde{D}}_d^{(d-2-q'+p')} \mathsf{\tilde{D}}_d^{(d-2)} \ket{\overline{k' \oplus (d-2-q')}} = 0,
    \end{equation}
    where the last equality follows from Hypothesis C4. Note that $d-2 - q' + p' \geq 0$.

    \vskip 16pt
	
    Finally, observe that $\mathsf{\tilde{D}}_d^{(p)} \mathsf{\tilde{D}}_d^{(q)} = \mathsf{\tilde{D}}_d^{(q)} \mathsf{\tilde{D}}_d^{(p)}$ for all $0 \leq p \leq q \leq d-2$. This implies that if $\bra{\overline{k}} \mathsf{\tilde{D}}_d^{(p)} \mathsf{\tilde{D}}_d^{(q)} \ket{\overline{k}} = 0$ for all $k \in \mathbf{Z}_d, 0 \leq p \leq q \leq d-2$, then
    \begin{equation}
        \bra{\overline{k}} \mathsf{\tilde{D}}_d^{(p)} \mathsf{\tilde{D}}_d^{(q)} \ket{\overline{k}} = 0
    \end{equation}
    for all $k \in \mathbf{Z}_d, p,q \in \{0,\cdots,d-2\}$.
\end{proof}

\section{Proof of Theorem 6}
We first prove the following lemma:
\renewcommand{\thetheorem}{S3}
\begin{lemma}
    Let $0 \leq p < q \leq d-1$ and $0 \leq r < s \leq d-1$. Let $\Gamma, \Upsilon \in \{\mathsf{S},\mathsf{A}\}$. Fix $\nu \in \mathbf{Z}_d$ and let $\mathbf{u} \in \mathcal{U}_{d,N,\nu}$.
    \begin{enumerate}[\rm (i),itemsep=2pt]
        \item If $p \neq q$ and $r \neq s$, then
        \begin{equation}
            \tilde{\Gamma}_d^{(p,q)}\tilde{\Upsilon}_d^{(r,s)} \ket{S_\mathbf{u}} \in \mathrm{span}_{\mathbf{C}\backslash\{0\}} \left\{ \ket{S_{\mathbf{a}^{}}}, \ket{S_{\mathbf{b}^{}}}, \ket{S_{\mathbf{c}^{}}}, \ket{S_{\mathbf{d}^{}}} \right\}
        \end{equation}
        with $\mathbf{a}, \mathbf{b}, \mathbf{c}, \mathbf{d}$ defined as follows:
        \begin{itemize}[itemsep=2pt]
            \item $a_p = u_p + 1$, $a_q = u_q - 1$, $a_r = u_r + 1$, $a_s = u_s - 1$, and $a_i = u_i$ for all $i \in \mathbf{Z}_d \backslash \{p,q,r,s\}$
            \item $b_p = u_p - 1$, $b_q = u_q + 1$, $b_r = u_r + 1$, $b_s = u_s - 1$, and $b_i = u_i$ for all $i \in \mathbf{Z}_d \backslash \{p,q,r,s\}$
            \item $c_p = u_p + 1$, $c_q = u_q - 1$, $c_r = u_r - 1$, $c_s = u_s + 1$, and $c_i = u_i$ for all $i \in \mathbf{Z}_d \backslash \{p,q,r,s\}$
            \item $d_p = u_p - 1$, $d_q = u_q + 1$, $d_r = u_r - 1$, $d_s = u_s + 1$, and $d_i = u_i$ for all $i \in \mathbf{Z}_d \backslash \{p,q,r,s\}$
        \end{itemize}

        \item If $p \neq r$ and $q = s$, then
        \begin{equation}
            \tilde{\Gamma}_d^{(p,q)}\tilde{\Upsilon}_d^{(r,s)} \ket{S_\mathbf{u}} \in \mathrm{span}_{\mathbf{C}\backslash\{0\}} \left\{ \ket{S_{\mathbf{a}^{}}}, \ket{S_{\mathbf{b}^{}}}, \ket{S_{\mathbf{c}^{}}}, \ket{S_{\mathbf{d}^{}}} \right\}
        \end{equation}
        with $\mathbf{a}, \mathbf{b}, \mathbf{c}, \mathbf{d}$ defined as follows:
        \begin{itemize}[itemsep=2pt]
            \item $a_p = u_p + 1$, $a_q = u_q - 2$, $a_r = u_r + 1$, $a_i = u_i$ for all $i \in \mathbf{Z}_d \backslash \{p,q,r\}$
            \item $b_p = u_p + 1$, $b_q = u_q$, $b_r = u_r - 1$, $b_i = u_i$ for all $i \in \mathbf{Z}_d \backslash \{p,q,r\}$
            \item $c_p = u_p - 1$, $c_q = u_q$, $c_r = u_r + 1$, $c_i = u_i$ for all $i \in \mathbf{Z}_d \backslash \{p,q,r\}$
            \item $d_p = u_p - 1$, $d_q = u_q + 2$, $d_r = u_r - 1$, $d_i = u_i$ for all $i \in \mathbf{Z}_d \backslash \{p,q,r\}$
        \end{itemize}
        
        \item If $p = r$ and $q \neq s$, then
        \begin{equation}
            \tilde{\Gamma}_d^{(p,q)}\tilde{\Upsilon}_d^{(r,s)} \ket{S_\mathbf{u}} \in \mathrm{span}_{\mathbf{C}\backslash\{0\}} \left\{ \ket{S_{\mathbf{a}^{}}}, \ket{S_{\mathbf{b}^{}}}, \ket{S_{\mathbf{c}^{}}}, \ket{S_{\mathbf{d}^{}}} \right\}
        \end{equation}
        with $\mathbf{a}, \mathbf{b}, \mathbf{c}, \mathbf{d}$ defined as follows:
        \begin{itemize}[itemsep=2pt]
            \item $a_p = u_p + 1$, $a_q = u_q - 2$, $a_s = u_s + 1$, $a_i = u_i$ for all $i \in \mathbf{Z}_d \backslash \{p,q,s\}$
            \item $b_p = u_p + 1$, $b_q = u_q$, $b_s = u_s - 1$, $b_i = u_i$ for all $i \in \mathbf{Z}_d \backslash \{p,q,s\}$
            \item $c_p = u_p - 1$, $c_q = u_q$, $c_s = u_s + 1$, $c_i = u_i$ for all $i \in \mathbf{Z}_d \backslash \{p,q,s\}$
            \item $d_p = u_p - 1$, $d_q = u_q + 2$, $d_s = u_s - 1$, $d_i = u_i$ for all $i \in \mathbf{Z}_d \backslash \{p,q,s\}$
        \end{itemize}
    \end{enumerate}
\end{lemma}

\begin{proof}
    This follows from two applications of Lemma S2.
\end{proof}

\vskip 16pt

We then proceed with proving Theorem 6:
\renewcommand{\thetheorem}{6}
\begin{theorem}
    If a code $\mathcal{C} \colon \mathbf{C}^d \longrightarrow \mathrm{Sym}^N(\mathbf{C}^d)$ is sparse, then $\mathcal{C}$ satisfies Condition C1 and Condition C2 with $(p,q) \neq (r,s)$.
\end{theorem}

\begin{proof}
    We prove the contrapositive. There are two cases to consider: Condition C1 is not satisfied or Condition C2 with $(p,q) \neq (r,s)$ is not satisfied.

    \vskip 16pt
    
    \paragraph{Condition C1}
    Suppose Condition C1 is not satisfied, i.e. for some $i',j' \in \mathbf{Z}_d$ and some $n' \in \{1,\cdots,\frac{d-1}{2}\}$,
    \begin{equation}
        \bra{\overline{i'}} \mathsf{\tilde{S}}_d^{(0,n')} \ket{\overline{j'}} \neq 0.
    \end{equation}
    This implies that there exists $\mathbf{u}, \mathbf{u}' \in \mathcal{U}_{d,N,0}$ such that $\alpha_\mathbf{u}, \alpha_{\mathbf{u}'} \neq 0$ and
    \begin{equation}
        \bra{s_{u_{i'},\cdots,u_{d-1},u_0,\cdots,u_{i'-1}}} \mathsf{\tilde{S}}_d^{(0,n')} \ket{s_{u'_{j'},\cdots,u'_{d-1},u'_0,\cdots,u'_{j'-1}}} \neq 0,
    \end{equation}
    whence it follows by Lemma S2 that one of the following is true:
    \begin{itemize}[itemsep=2pt]
        \item $\braket{s_{u_{i'},\cdots,u_{d-1},u_0,\cdots,u_{i'-1}}|s_{u'_{j'}+1,\cdots,u'_{j'+n'}-1,\cdots,u'_{j'-1}}} \neq 0$
        \item $\braket{s_{u_{i'},\cdots,u_{d-1},u_0,\cdots,u_{i'-1}}|s_{u'_{j'}-1,\cdots,u'_{j'+n'}+1,\cdots,u'_{j'-1}}} \neq 0$
    \end{itemize}
    Either way, we can set $\delta' = j'-i'$ and note that  \begin{equation}
        \sum_{\xi \in \mathbf{Z}_d} |u_\xi - u'_{\xi+\delta'}| = 2 \leq 4.
    \end{equation}
    Thus, $\mathcal{C}$ is not sparse.

    \vskip 16pt
    
    \paragraph{Condition C2 with $(p,q) \neq (r,s)$}
    Suppose that, for some $i',j' \in \mathbf{Z}_d$ and some $0 \leq p' < q' \leq d-1, 0 \leq r' < s' \leq d-1$ with $(p',q') \neq (r',s')$,
    \begin{equation}
        \bra{\overline{i'}} \tilde{\Gamma}_d^{(p',q')} \tilde{\Upsilon}_d^{(r',s')} \ket{\overline{j'}} \neq 0.
    \end{equation}
   This implies that there exists $\mathbf{u}, \mathbf{u}' \in \mathcal{Z}_{d,N,0}$ such that $\alpha_\mathbf{u}, \alpha_{\mathbf{u}'} \neq 0$ and
    \begin{equation}
        \bra{s_{u_{i'},\cdots,u_{d-1},u_0,\cdots,u_{i'-1}}} \tilde{\Gamma}_d^{(p',q')} \tilde{\Upsilon}_d^{(r',s')} \ket{s_{u'_{j'},\cdots,u'_{d-1},u'_0,\cdots,u'_{j'-1}}} \neq 0,
    \end{equation}
    By Lemma S3, we see that for $\delta' = j' - i'$,
    \begin{equation}
        \sum_{\xi \in \mathbf{Z}_d} |u_\xi - u'_{\xi+\delta'}| \in \{2,4\},
    \end{equation}
    whence we conclude that $\mathcal{C}$ is not sparse.
\end{proof}

\section{Proof of Theorem 7}
\renewcommand{\thetheorem}{7}
\begin{theorem}
	\label{reduced KL conditions for DPI codes}
    Let $\mathcal{C} \colon \mathbf{C}^d \longrightarrow \mathrm{DSym}^N(\mathbf{C}^d)$ be a doubly permutation-invariant code. If $\mathcal{C}$ satisfies
    \begin{enumerate}[\rm (QF1),leftmargin=3em,itemsep=2pt]
        \item $\bra{\overline{d-1}} \mathsf{\tilde{D}}_d^{(d-2)} \ket{\overline{d-1}} = 0$
        \item $\bra{\overline{0}} \mathsf{\tilde{D}}_d^{(d-2)}\mathsf{\tilde{D}}_d^{(d-2)} \ket{\overline{0}} = \bra{\overline{d-1}} \mathsf{\tilde{D}}_d^{(d-2)}\mathsf{\tilde{D}}_d^{(d-2)} \ket{\overline{d-1}}$
    \end{enumerate}
    then it satisfies Condition C3 and Condition C4.
\end{theorem}

\begin{proof}
    We prove each condition systematically.

    \vskip 6pt
    
    \paragraph{Condition C3}
    We first show that $\bra{\overline{0}} \mathsf{\tilde{D}}_d^{(d-2)} \ket{\overline{0}} = 0$. Fix an injection $f \colon \mathcal{W}_{d,N}/\hat{\mathcal{R}} \longrightarrow \mathcal{W}_{d,N}$ that maps each equivalence class to a chosen representative. Then
    \begin{align}
        \bra{\overline{0}} \mathsf{\tilde{D}}_d^{(d-2)} \ket{\overline{0}} &= \sum_{\mathbf{w} \in f(\mathcal{W}_{d,N}/\hat{\mathcal{R}})} |\alpha_\mathbf{w}|^2 \sum_{\mathbf{w}' \in [\mathbf{w}]_{\hat{\mathcal{R}}}^{}} \bra{S_{\mathbf{w}'}} \mathsf{\tilde{D}}_d^{(d-2)} \ket{S_{\mathbf{w}'}}\\
        &= \sum_{\mathbf{w} \in f(\mathcal{W}_{d,N}/\hat{\mathcal{R}})} |\alpha_\mathbf{w}|^2 \sum_{\mathbf{w}' \in [\mathbf{w}]_{\hat{\mathcal{R}}}^{}} \left(w'_{d-2} - w'_{d-1}\right) \label{2nd last 1}\\
        &= 0 \label{last 1},
    \end{align}
    where going Eq. (\ref{2nd last 1}) to Eq. (\ref{last 1}) requires the observation that, for each $\mathbf{w}' \in [\mathbf{w}]_{\hat{\mathcal{R}}}^{}$, there exists a unique element $\mathbf{w}'' \in [\mathbf{w}]_{\hat{\mathcal{R}}}^{}$ such that $w'_{d-2} = w''_{d-1}$, $w'_{d-1} = w''_{d-2}$, and $w'_i = w''_i$ for all $i \in \{0,\cdots,d-3\}$. [\textit{It is in arguments like this that we see the second form of permutation invariance in action!}]

    \vskip 16pt

    Next, we combine the above argument with Heisenberg-Weyl symmetry to show that $\bra{\overline{k}} \mathsf{\tilde{D}}_d^{(d-2)} \ket{\overline{k}} = 0$ for all $k \in \{1,\cdots,d-3\}$. Fix $k' \in \{1,\cdots,d-3\}$ and an injection $f \colon \mathcal{W}_{d,N}/\hat{\mathcal{R}} \longrightarrow \mathcal{W}_{d,N}$ that maps each equivalence class to a chosen representative. Notice that
    \begin{align}
        \bra{\overline{k'}} \mathsf{\tilde{D}}_d^{(d-2)} \ket{\overline{k'}} &= \bra{\overline{0}} \left(\mathsf{\overline{X}}_d^\dagger\right)^{k'} \mathsf{\tilde{D}}_d^{(d-2)} \mathsf{\overline{X}}_d^{k'} \ket{\overline{0}}\\
        &= \bra{\overline{0}} \mathsf{\tilde{D}}_d^{(d-2-k')} \ket{\overline{0}}\\
        &= \sum_{\mathbf{w} \in f(\mathcal{W}_{d,N}/\hat{\mathcal{R}})} |\alpha_\mathbf{w}|^2 \sum_{\mathbf{w}' \in [\mathbf{w}]_{\hat{\mathcal{R}}}^{}} \bra{S_{\mathbf{w}'}} \mathsf{\tilde{D}}_d^{(d-2-k')} \ket{S_{\mathbf{w}'}}\\
        &= \sum_{\mathbf{w} \in f(\mathcal{W}_{d,N}/\hat{\mathcal{R}})} |\alpha_\mathbf{w}|^2 \sum_{\mathbf{w}' \in [\mathbf{w}]_{\hat{\mathcal{R}}}^{}} \left(w'_{d-2-k'} - w'_{d-1-k'}\right)\\
        &= 0.
    \end{align}

    \vskip 8pt

    Finally, it remains to show that $\bra{\overline{d-2}} \mathsf{\tilde{D}}_d^{(d-2)} \ket{\overline{d-2}} = -\bra{\overline{d-1}} \mathsf{\tilde{D}}_d^{(d-2)} \ket{\overline{d-1}}$. Fix an injection $f \colon \mathcal{W}_{d,N}/\hat{\mathcal{R}} \longrightarrow \mathcal{W}_{d,N}$ that maps each equivalence class to a chosen representative. Observe that
    \begin{align}
        \bra{\overline{d-2}} \mathsf{\tilde{D}}_d^{(d-2)} \ket{\overline{d-2}} &= \sum_{\mathbf{w} \in f(\mathcal{W}_{d,N}/\hat{\mathcal{R}})} |\alpha_\mathbf{w}|^2 \sum_{\mathbf{w}' \in [\mathbf{w}]_{\hat{\mathcal{R}}}^{}} \left(w'_{0} - w'_{1}\right)\\
        &= -\sum_{\mathbf{w} \in f(\mathcal{W}_{d,N}/\hat{\mathcal{R}})} |\alpha_\mathbf{w}|^2 \sum_{\mathbf{w}'' \in [\mathbf{w}]_{\hat{\mathcal{R}}}^{}} \left(w''_{d-1} - w''_{0}\right)\\
        &= -\bra{\overline{d-1}} \mathsf{\tilde{D}}_d^{(d-2)} \ket{\overline{d-1}},
    \end{align}
    where we have used the fact that, for each $\mathbf{w}' \in [\mathbf{w}]_{\hat{\mathcal{R}}}^{}$, there exists a unique $\mathbf{w}'' \in [\mathbf{w}]_{\hat{\mathcal{R}}}^{}$ such that $w'_1 = w''_{d-1}$, $w'_{d-1} = w''_1$, and $w'_i = w''_i$ for all $i \in \{0,2,3,\cdots,d-2\}$.

    \vskip 16pt

    \paragraph{Condition C4}
    First, fix $k' \in \{0,\cdots,d-3\}$, and observe that
    \begin{align}
        \bra{\overline{k'}} \mathsf{\tilde{D}}_d^{(d-2)} \mathsf{\tilde{D}}_d^{(d-2)} \ket{\overline{k'}} &= \bra{\overline{0}} \mathsf{\tilde{D}}_d^{(d-2-k')} \mathsf{\tilde{D}}_d^{(d-2-k')} \ket{\overline{0}}\\
        &= \sum_{\mathbf{w} \in f(\mathcal{W}_{d,N}/\hat{\mathcal{R}})} |\alpha_\mathbf{w}|^2 \sum_{\mathbf{w}' \in [\mathbf{w}]_{\hat{\mathcal{R}}}^{}} \bra{S_{\mathbf{w}'}} \mathsf{\tilde{D}}_d^{(d-2-k')} \mathsf{\tilde{D}}_d^{(d-2-k')} \ket{S_{\mathbf{w}'}}\\
        &= \sum_{\mathbf{w} \in f(\mathcal{W}_{d,N}/\hat{\mathcal{R}})} |\alpha_\mathbf{w}|^2 \sum_{\mathbf{w}' \in [\mathbf{w}]_{\hat{\mathcal{R}}}^{}} (w'_{d-2-k'} - w'_{d-1-k'})^2\\
        &= \sum_{\mathbf{w} \in f(\mathcal{W}_{d,N}/\hat{\mathcal{R}})} |\alpha_\mathbf{w}|^2 \sum_{\mathbf{w}'' \in [\mathbf{w}]_{\hat{\mathcal{R}}}^{}} (w''_{d-2} - w''_{d-1})^2\\
        &= \bra{\overline{0}} \mathsf{\tilde{D}}_d^{(d-2)} \mathsf{\tilde{D}}_d^{(d-2)} \ket{\overline{0}}\\
        \nonumber\\
        \bra{\overline{d-2}} \mathsf{\tilde{D}}_d^{(d-2)} \mathsf{\tilde{D}}_d^{(d-2)} \ket{\overline{d-2}} &= \bra{\overline{0}} \mathsf{\tilde{D}}_d^{(0)} \mathsf{\tilde{D}}_d^{(0)} \ket{\overline{0}}\\
        &= \sum_{\mathbf{w} \in f(\mathcal{W}_{d,N}/\hat{\mathcal{R}})} |\alpha_\mathbf{w}|^2 \sum_{\mathbf{w}' \in [\mathbf{w}]_{\hat{\mathcal{R}}}^{}} (w'_0-w'_1)^2\\
        &= \sum_{\mathbf{w} \in f(\mathcal{W}_{d,N}/\hat{\mathcal{R}})} |\alpha_\mathbf{w}|^2 \sum_{\mathbf{w}'' \in [\mathbf{w}]_{\hat{\mathcal{R}}}^{}} (w''_{d-1}-w''_0)^2\\
        &= \bra{\overline{d-1}} \mathsf{\tilde{D}}_d^{(d-2)} \mathsf{\tilde{D}}_d^{(d-2)} \ket{\overline{d-1}}.
    \end{align}
    Together with Hypothesis QF2, these imply that, for all $k \in \mathbf{Z}_d$, $\bra{\overline{k}} \mathsf{\tilde{D}}_d^{(d-2)} \mathsf{\tilde{D}}_d^{(d-2)} \ket{\overline{k}} = c_{(\mathsf{\tilde{D}}_d^{(d-2)}, \mathsf{\tilde{D}}_d^{(d-2)})}$ for some $c_{(\mathsf{\tilde{D}}_d^{(d-2)}, \mathsf{\tilde{D}}_d^{(d-2)})}$.

    \vskip 16pt

    Now, fix $\ell' \in \{0,\cdots,d-4\}$ and $m' \in \mathbf{Z}_d$. Observe that
    \begin{align}
        \bra{\overline{m'}} \mathsf{\tilde{D}}_d^{(\ell')} \mathsf{\tilde{D}}_d^{(d-2)} \ket{\overline{m'}} &= \bra{\overline{0}} \mathsf{\tilde{D}}_d^{(\ell' \ominus m')} \mathsf{\tilde{D}}_d^{((d-2) \ominus m')} \ket{\overline{0}}\\
        &= \sum_{\mathbf{w} \in f(\mathcal{W}_{d,N}/\hat{\mathcal{R}})} |\alpha_\mathbf{w}|^2 \sum_{\mathbf{w}' \in [\mathbf{w}]_{\hat{\mathcal{R}}}^{}} (w'_{\ell' \ominus m'} - w'_{(\ell'+1) \ominus m'})(w'_{(d-2) \ominus m'} - w'_{(d-1) \ominus m'})\\
        &= 0
    \end{align}
    since for each $\mathbf{w}' \in [\mathbf{w}]_{\hat{\mathcal{R}}}^{}$, there is either of the following:
    \begin{itemize}[itemsep=2pt]
        \item a unique $\mathbf{w}'' \in [\mathbf{w}]_{\hat{\mathcal{R}}}^{}$ such that $w'_{\ell' \ominus m'} = w''_{(\ell'+1) \ominus m'}$, $w'_{(\ell'+1) \ominus m'} = w''_{\ell' \ominus m'}$, and $w'_i = w''_i$ for all $i \in \mathbf{Z}_d \backslash \{\ell' \ominus m', (\ell'+1) \ominus m'\}$
        \item a unique $\mathbf{w}''' \in [\mathbf{w}]_{\hat{\mathcal{R}}}^{}$ such that $w'_{(d-2) \ominus m'} = w'''_{(d-1) \ominus m'}$, $w'_{(d-1) \ominus m'} = w'''_{(d-2) \ominus m'}$, and $w'_i = w'''_i$ for all $i \in \mathbf{Z}_d \backslash \{(d-2) \ominus m', (d-1) \ominus m'\}$
    \end{itemize}

    \vskip 16pt

    Finally, fix $n' \in \mathbf{Z}_d$ and observe that
    \begin{align}
        \bra{\overline{n'}} \mathsf{\tilde{D}}_d^{(d-3)} \mathsf{\tilde{D}}_d^{(d-2)} \ket{\overline{n'}} &= \bra{\overline{0}} \mathsf{\tilde{D}}_d^{((d-3) \ominus n')} \mathsf{\tilde{D}}_d^{((d-2) \ominus n')} \ket{\overline{0}}\\
        &= \sum_{\mathbf{w} \in f(\mathcal{W}_{d,N}/\hat{\mathcal{R}})} |\alpha_\mathbf{w}|^2 \sum_{\mathbf{w}' \in [\mathbf{w}]_{\hat{\mathcal{R}}}^{}} (w'_{(d-3) \ominus n'} - w'_{(d-2) \ominus n'})(w'_{(d-2) \ominus n'} - w'_{(d-1) \ominus n'})\\
        &= 0
    \end{align}
    since for each $\mathbf{w}' \in [\mathbf{w}]_{\hat{\mathcal{R}}}^{}$, there is either of the following:
    \begin{itemize}[itemsep=2pt]
        \item a unique $\mathbf{w}'' \in [\mathbf{w}]_{\hat{\mathcal{R}}}^{}$ such that $w'_{(d-3) \ominus n'} = w''_{(d-1) \ominus n'}$, $w'_{(d-1) \ominus n'} = w''_{(d-3) \ominus n'}$, and $w'_i = w''_i$ for all $i \in \mathbf{Z}_d \backslash \{(d-3) \ominus n', (d-1) \ominus n'\}$
        \item a unique $\mathbf{w}''' \in [\mathbf{w}]_{\hat{\mathcal{R}}}^{}$ such that $w'_{(d-3) \ominus n'} = w'''_{(d-2) \ominus n'}$, $w'_{(d-2) \ominus n'} = w'''_{(d-3) \ominus n'}$, and $w'_i = w'''_i$ for all $i \in \mathbf{Z}_d \backslash \{(d-3) \ominus n', (d-2) \ominus n'\}$
        \item a unique $\mathbf{w}'''' \in [\mathbf{w}]_{\hat{\mathcal{R}}}^{}$ such that $w'_{(d-2) \ominus n'} = w''''_{(d-1) \ominus n'}$, $w'_{(d-1) \ominus n'} = w''''_{(d-2) \ominus n'}$, and $w'_i = w''''_i$ for all $i \in \mathbf{Z}_d \backslash \{(d-2) \ominus n', (d-1) \ominus n'\}$
    \end{itemize}
\end{proof}

\section{Proof of Theorem 8}
\renewcommand{\thetheorem}{8}
\begin{theorem}
	Let $\mathcal{C} \colon \mathbf{C}^d \longrightarrow \mathrm{DSym}^N(\mathbf{C}^d)$ be an SDPI code. If $\mathcal{C}$ satisfies
    \begin{enumerate}[\rm (QF3),leftmargin=3em,itemsep=2pt]
        \item $\bra{\overline{0}} \mathsf{\tilde{S}}_d^{(0,1)}\mathsf{\tilde{S}}_d^{(0,1)} \ket{\overline{0}} = \bra{\overline{d-1}} \mathsf{\tilde{S}}_d^{(0,1)}\mathsf{\tilde{S}}_d^{(0,1)} \ket{\overline{d-1}}$
    \end{enumerate}
    then it satisfies Condition C2 with $(p,q) = (r,s)$.
\end{theorem}

\begin{proof}
    We first show that $\bra{\overline{0}} \mathsf{\tilde{S}}_d^{(0,1)}\mathsf{\tilde{S}}_d^{(0,1)} \ket{\overline{0}} = \bra{\overline{1}} \mathsf{\tilde{S}}_d^{(0,1)}\mathsf{\tilde{S}}_d^{(0,1)} \ket{\overline{1}}$. Using Heisenberg-Weyl symmetry, we have
    \begin{equation}
        \bra{\overline{1}} \mathsf{\tilde{S}}_d^{(0,1)}\mathsf{\tilde{S}}_d^{(0,1)} \ket{\overline{1}} = \bra{\overline{0}} \mathsf{\overline{X}}_d^\dagger \mathsf{\tilde{S}}_d^{(0,1)} \mathsf{\overline{X}}_d \mathsf{\overline{X}}_d^\dagger \mathsf{\tilde{S}}_d^{(0,1)} \mathsf{\overline{X}}_d \ket{\overline{0}} = \bra{\overline{0}} \mathsf{\tilde{S}}_d^{(0,d-1)}\mathsf{\tilde{S}}_d^{(0,d-1)} \ket{\overline{0}}.
    \end{equation}
    Moreover, observe that, for any given injection $f \colon \mathcal{W}_{d,N}/\hat{\mathcal{R}} \longrightarrow \mathcal{W}_{d,N}$ that maps each equivalence class to a chosen representative, we have
    \begin{align}
        \mathsf{\tilde{S}}_d^{(0,d-1)} \ket{\overline{0}} &= \sum_{\mathbf{w} \in f(\mathcal{W}_{d,N}/\hat{\mathcal{R}})} |\alpha_\mathbf{w}|^2 \sum_{\mathbf{w}' \in [\mathbf{w}]_{\hat{\mathcal{R}}}^{}} [(w'_0+1)w'_{d-1} + w'_0(w'_{d-1}+1)] \ket{S_{\mathbf{w}'}}\\
        &= \sum_{\mathbf{w} \in f(\mathcal{W}_{d,N}/\hat{\mathcal{R}})} |\alpha_\mathbf{w}|^2 \sum_{\mathbf{w}'' \in [\mathbf{w}]_{\hat{\mathcal{R}}}^{}} [(w''_0+1)w''_{1} + w''_0(w''_{1}+1)] \ket{S_{\mathbf{w}''}}\\
        &= \mathsf{\tilde{S}}_d^{(0,1)} \ket{\overline{0}},
    \end{align}
    from which the conclusion follows directly.
    
    \vskip 16pt
    
    Next, we show that $\bra{\overline{k}} \mathsf{\tilde{S}}_d^{(0,1)}\mathsf{\tilde{S}}_d^{(0,1)} \ket{\overline{k}} = \bra{\overline{d-1}} \mathsf{\tilde{S}}_d^{(0,1)}\mathsf{\tilde{S}}_d^{(0,1)} \ket{\overline{d-1}}$ for all $k \in \{2,\cdots,d-1\}$. Fix $k' \in \{2,\cdots,d-2\}$. Note that
    \begin{align}
        \bra{\overline{k'}} \mathsf{\tilde{S}}_d^{(0,1)}\mathsf{\tilde{S}}_d^{(0,1)} \ket{\overline{k'}} &= \bra{\overline{0}} \mathsf{\tilde{S}}_d^{(d-k',d-k'+1)}\mathsf{\tilde{S}}_d^{(d-k',d-k'+1)} \ket{\overline{0}},\\
        \bra{\overline{d-1}} \mathsf{\tilde{S}}_d^{(0,1)}\mathsf{\tilde{S}}_d^{(0,1)} \ket{\overline{d-1}} &= \bra{\overline{0}} \mathsf{\tilde{S}}_d^{(1,2)}\mathsf{\tilde{S}}_d^{(1,2)} \ket{\overline{0}}.
    \end{align}
	Again, for any injection $f \colon \mathcal{W}_{d,N}/\hat{\mathcal{R}} \longrightarrow \mathcal{W}_{d,N}$ that maps each equivalence class to a chosen representative, we have
    \begin{align}
        \mathsf{\tilde{S}}_d^{(1,2)} \ket{\overline{0}} &= \sum_{\mathbf{w} \in f(\mathcal{W}_{d,N}/\hat{\mathcal{R}})} |\alpha_\mathbf{w}|^2 \sum_{\mathbf{w}' \in [\mathbf{w}]_{\hat{\mathcal{R}}}^{}} [(w'_1+1)w'_2 + w'_1(w'_2+1)] \ket{S_{\mathbf{w}'}}\\
        &= \sum_{\mathbf{w} \in f(\mathcal{W}_{d,N}/\hat{\mathcal{R}})} |\alpha_\mathbf{w}|^2 \sum_{\mathbf{w}'' \in [\mathbf{w}]_{\hat{\mathcal{R}}}^{}} [(w''_{d-k'}+1)w''_{d-k'+1} + w''_{d-k'}(w''_{d-k'+1}+1)] \ket{S_{\mathbf{w}''}}\\
        &= \mathsf{\tilde{S}}_d^{(d-k',d-k'+1)} \ket{\overline{0}},
    \end{align}
    as required.

    \vskip 16pt

    Combining the two results above with Hypothesis QF3, we have thus far proven that, for all $k \in \mathbf{Z}_d$, we have $\bra{\overline{k}} \mathsf{\tilde{S}}_d^{(0,1)} \mathsf{\tilde{S}}_d^{(0,1)} \ket{\overline{k}} = c_{(\mathsf{\tilde{S}}_d^{(0,1)},\mathsf{\tilde{S}}_d^{(0,1)})}$ for some $c_{(\mathsf{\tilde{S}}_d^{(0,1)},\mathsf{\tilde{S}}_d^{(0,1)})}$. As we show below, this implies that, for all $k \in \mathbf{Z}_d$ and $0 \leq p < q \leq d-1$, we have $\bra{\overline{k}} \tilde{\Gamma}_d^{(p,q)} \tilde{\Upsilon}_d^{(p,q)} \ket{\overline{k}} = c_{(\tilde{\Gamma}_d^{(p,q)}, \tilde{\Upsilon}_d^{(p,q)})}$ for some $c_{(\tilde{\Gamma}_d^{(p,q)}, \tilde{\Upsilon}_d^{(p,q)})}$. Fix $0 \leq p' < q' \leq d-1$. Observe that $\bra{\overline{p'}} \mathsf{\tilde{S}}_d^{(p',q')} \mathsf{\tilde{S}}_d^{(p',q')} \ket{\overline{p'}} = \bra{\overline{q'}} \mathsf{\tilde{S}}_d^{(p',q')} \mathsf{\tilde{S}}_d^{(p',q')} \ket{\overline{q'}} = \bra{\overline{0}} \mathsf{\tilde{S}}_d^{(0,1)}\mathsf{\tilde{S}}_d^{(0,1)} \ket{\overline{0}}$:
    \begin{align}
        \bra{\overline{q'}} \mathsf{\tilde{S}}_d^{(p',q')} \mathsf{\tilde{S}}_d^{(p',q')} \ket{\overline{q'}} &= \bra{\overline{q'}} \mathsf{\overline{X}}_d^{q'} \mathsf{\tilde{S}}_d^{(0,d+p'-q')} \left(\mathsf{\overline{X}}_d^\dagger\right)^{q'} \mathsf{\overline{X}}_d^{q'} \mathsf{\tilde{S}}_d^{(0,d+p'-q')} \left(\mathsf{\overline{X}}_d^\dagger\right)^{q'} \ket{\overline{q'}}\\
        &= \bra{\overline{0}} \mathsf{\tilde{S}}_d^{(0,d+p'-q')} \mathsf{\tilde{S}}_d^{(0,d+p'-q')} \ket{\overline{0}}\\
        &= \sum_{\mathbf{w} \in f(\mathcal{W}_{d,N}/\hat{\mathcal{R}})} |\alpha_\mathbf{w}|^2 \sum_{\mathbf{w}' \in [\mathbf{w}]_{\hat{\mathcal{R}}}^{}} [(w'_0+1)w'_{d+p'-q'} + w'_0(w'_{d+p'-q'}+1)]^2 \\
        &= \sum_{\mathbf{w} \in f(\mathcal{W}_{d,N}/\hat{\mathcal{R}})} |\alpha_\mathbf{w}|^2 \sum_{\mathbf{w}'' \in [\mathbf{w}]_{\hat{\mathcal{R}}}^{}} [(w''_0+1)w''_{1} + w''_0(w''_{1}+1)]^2\\
        &= \bra{\overline{0}} \mathsf{\tilde{S}}_d^{(0,1)} \mathsf{\tilde{S}}_d^{(0,1)} \ket{\overline{0}},\\
        \nonumber\\
        \bra{\overline{p'}} \mathsf{\tilde{S}}_d^{(p',q')} \mathsf{\tilde{S}}_d^{(p',q')} \ket{\overline{p'}} &= \bra{\overline{p'}} \mathsf{\overline{X}}_d^{p'} \mathsf{\tilde{S}}_d^{(0,q'-p')} \left(\mathsf{\overline{X}}_d^\dagger\right)^{p'} \mathsf{\overline{X}}_d^{p'} \mathsf{\tilde{S}}_d^{(0,q'-p')} \left(\mathsf{\overline{X}}_d^\dagger\right)^{p'} \ket{\overline{p'}}\\
        &= \bra{\overline{0}} \mathsf{\tilde{S}}_d^{(0,q'-p')} \mathsf{\tilde{S}}_d^{(0,q'-p')} \ket{\overline{0}}\\
        &= \sum_{\mathbf{w} \in f(\mathcal{W}_{d,N}/\hat{\mathcal{R}})} |\alpha_\mathbf{w}|^2 \sum_{\mathbf{w}' \in [\mathbf{w}]_{\hat{\mathcal{R}}}^{}} [(w'_0+1)w'_{q'-p'} + w'_0(w'_{q'-p'}+1)]^2\\
        &= \sum_{\mathbf{w} \in f(\mathcal{W}_{d,N}/\hat{\mathcal{R}})} |\alpha_\mathbf{w}|^2 \sum_{\mathbf{w}'' \in [\mathbf{w}]_{\hat{\mathcal{R}}}^{}} [(w''_0+1)w''_{1} + w''_0(w''_{1}+1)]^2\\
        &= \bra{\overline{0}} \mathsf{\tilde{S}}_d^{(0,1)} \mathsf{\tilde{S}}_d^{(0,1)} \ket{\overline{0}}.
    \end{align}
    Moreover, notice that for any $k \in \mathbf{Z}_d \backslash \{p',q'\}$,
    \begin{align}
        \bra{\overline{k}} \mathsf{\tilde{S}}_d^{(p',q')} \mathsf{\tilde{S}}_d^{(p',q')} \ket{\overline{k}} &= \bra{\overline{0}} \mathsf{\tilde{S}}_d^{(p' \ominus k, q' \ominus k)} \mathsf{\tilde{S}}_d^{(p' \ominus k, q' \ominus k)} \ket{\overline{0}}\\
        &= \sum_{\mathbf{w} \in f(\mathcal{W}_{d,N}/\hat{\mathcal{R}})} |\alpha_\mathbf{w}|^2 \sum_{\mathbf{w}' \in [\mathbf{w}]_{\hat{\mathcal{R}}}^{}} [(w'_{p \ominus k}+1)w'_{q' \ominus k} + w'_{p \ominus k}(w'_{q' \ominus k}+1)]^2\\
        &= \sum_{\mathbf{w} \in f(\mathcal{W}_{d,N}/\hat{\mathcal{R}})} |\alpha_\mathbf{w}|^2 \sum_{\mathbf{w}'' \in [\mathbf{w}]_{\hat{\mathcal{R}}}^{}} [(w''_1+1)w''_2 + w''_1(w''_2+1)]^2\\
        &= \bra{\overline{0}} \mathsf{\tilde{S}}_d^{(1,2)} \mathsf{\tilde{S}}_d^{(1,2)} \ket{\overline{0}}\\
        &= \bra{\overline{d-1}} \mathsf{\tilde{S}}_d^{(0,1)} \mathsf{\tilde{S}}_d^{(0,1)} \ket{\overline{d-1}}.
    \end{align}
    This shows that $\bra{\overline{k}} \mathsf{\tilde{S}}_d^{(p',q')} \mathsf{\tilde{S}}_d^{(p',q')} \ket{\overline{k}} = c_{(\mathsf{\tilde{S}}_d^{(p',q')}, \mathsf{\tilde{S}}_d^{(p',q')})}$ for all $k \in \mathbf{Z}_d$. By a similar argument, we can show that $\bra{\overline{k}} \mathsf{\tilde{A}}_d^{(p',q')} \mathsf{\tilde{A}}_d^{(p',q')} \ket{\overline{k}} = c_{(\mathsf{\tilde{A}}_d^{(p',q')}, \mathsf{\tilde{A}}_d^{(p',q')})}$, $\bra{\overline{k}} \mathsf{\tilde{S}}_d^{(p',q')} \mathsf{\tilde{A}}_d^{(p',q')} \ket{\overline{k}} = c_{(\mathsf{\tilde{S}}_d^{(p',q')}, \mathsf{\tilde{A}}_d^{(p',q')})}$, and $\bra{\overline{k}} \mathsf{\tilde{A}}_d^{(p',q')} \mathsf{\tilde{S}}_d^{(p',q')} \ket{\overline{k}} = c_{(\mathsf{\tilde{A}}_d^{(p',q')}, \mathsf{\tilde{S}}_d^{(p',q')})}$ for all $k \in \mathbf{Z}_d$. [\textit{The specific action of $\mathsf{\tilde{A}}_d^{(p',q')} \mathsf{\tilde{A}}_d^{(p',q')}$, $\mathsf{\tilde{S}}_d^{(p',q')} \mathsf{\tilde{A}}_d^{(p',q')}$, and $\mathsf{\tilde{A}}_d^{(p',q')} \mathsf{\tilde{S}}_d^{(p',q')}$ on any given basis vector differs from that of $\mathsf{\tilde{S}}_d^{(p',q')} \mathsf{\tilde{S}}_d^{(p',q')}$, but the same reasoning (using double permutation invariance) still works.}]

    \vskip 16pt

    It remains to show that $\bra{\overline{i}} \tilde{\Gamma}_d^{(p,q)} \tilde{\Upsilon}_d^{(p,q)} \ket{\overline{j}} = 0$ for all $i,j \in \mathbf{Z}_d$ with $i \neq j$ and $0 \leq p < q \leq d-1$. Note that it suffices to do this for $\Gamma = \Upsilon = \mathsf{S}$; a similar argument works for combinations with $\mathsf{A}$. Fix $i',j' \in \mathbf{Z}_d$ with $i \neq j$ and $0 \leq p' < q' \leq d-1$. Observe that
    \begin{align}
        \label{action of double S}
        \mathsf{\tilde{S}}_d^{(p' \ominus j', q' \ominus j')} \mathsf{\tilde{S}}_d^{(p' \ominus j', q' \ominus j')} \ket{\overline{0}} = \sum_{\mathbf{w} \in f(\mathcal{W}_{d,N}/\hat{\mathcal{R}})} |\alpha_\mathbf{w}|^2 \sum_{\mathbf{w}' \in [\mathbf{w}]_{\hat{\mathcal{R}}}^{}} \Big\{
        & (w_{p' \ominus j'}+2)(w_{p' \ominus j'}+1) \ket{S_{\mathbf{u}'}}\nonumber\\
        & + [w_{p' \ominus j'}(w_{q' \ominus j'}+1) + (w_{p' \ominus j'}+1)w_{q' \ominus j'}] \ket{S_{\mathbf{w}'}}\nonumber\\
        & + (w_{q' \ominus j'}+2)(w_{q' \ominus j'}+1) \ket{S_{\mathbf{v}'}}\Big\},
    \end{align}
    where $\mathbf{u}'$ and $\mathbf{v}'$ are defined by
    \begin{itemize}[itemsep=2pt]
        \item $u'_{p' \ominus j'} = w'_{p' \ominus j'}+2$, $u'_{q' \ominus j'} = w'_{q' \ominus j'}-2$, and $u'_i = w'_i$ for all $i \in \mathbf{Z}_d \backslash \{p' \ominus j',q' \ominus j'\}$
        \item $v'_{p' \ominus j'} = w'_{p' \ominus j'}-2$, $v'_{q' \ominus j'} = w'_{q' \ominus j'}+2$, and $v'_i = w'_i$ for all $i \in \mathbf{Z}_d \backslash \{p' \ominus j',q' \ominus j'\}$
    \end{itemize}    
    We then have
    \begin{equation}
        \bra{\overline{i'}} \mathsf{\tilde{S}}_d^{(p',q')} \mathsf{\tilde{S}}_d^{(p',q')} \ket{\overline{j'}} = \bra{\overline{i' \ominus j'}} \mathsf{\tilde{S}}_d^{(p' \ominus j', q' \ominus j')} \mathsf{\tilde{S}}_d^{(p' \ominus j', q' \ominus j')} \ket{\overline{0}} = 0,
    \end{equation}
    where the last equality follows from Eq. (\ref{action of double S}) and the hypothesis that $\mathcal{C}$ is sparse.
\end{proof}

\section{Proof of Theorem 10}
\renewcommand{\thetheorem}{10}
\begin{theorem}
    Let $\mathbf{a} = ((d-1)^2,\underbrace{0,\cdots,0}_{d-1})$, $\mathbf{b} = (d+1,d(d-3),\underbrace{0,\cdots,0}_{d-2})$, $\mathbf{c} = (0,\underbrace{d-1,\cdots,d-1}_{d-1})$. For each odd integer $d \geq 5$, there exists an error-correcting SDPI code with the logical code word
    \begin{equation}
        \label{general code word}
        \ket{\overline{0}} = \alpha_\mathbf{a}\ket{\ket{S_\mathbf{a}}} + \alpha_\mathbf{b}\ket{\ket{S_\mathbf{b}}} + \alpha_\mathbf{c}\ket{\ket{S_\mathbf{c}}}
    \end{equation}
    for some coefficients $\alpha_\mathbf{a}$, $\alpha_\mathbf{b}$, $\alpha_\mathbf{c}$. In particular, these coefficients are given by
    \begin{align}
        \alpha_\mathbf{a} &= \sqrt{\frac{d^3 - 5d^2 + d - 1}{2d^4 - 6d^3}},\\
        \alpha_\mathbf{b} &= \sqrt{\binom{(d-1)^2}{\mathbf{b}}^{-1} \frac{(d-1)(1-d\alpha_\mathbf{a}^2)}{d^2+d}},\\
        \alpha_\mathbf{c} &= \sqrt{\binom{(d-1)^2}{\mathbf{c}}^{-1} \left(1 - \alpha_\mathbf{a}^2 + \frac{1-d}{d^2+d}\alpha_\mathbf{b}^2\right)}.
    \end{align}
\end{theorem}

\begin{proof}
    With the code word in Eq. (\ref{general code word}), we get the following system of quadratic forms:
    \begin{align}
        (d-1)^2 ||\alpha_\mathbf{a}||^2 + (3d-1) ||\alpha_\mathbf{b}||^2 - (d-1) ||\alpha_\mathbf{c}||^2 &= 0 \label{coeff 1}\\
        (d-1)^4 ||\alpha_\mathbf{a}||^2 - (d^4 - 5d^3 + 4d^2 - 5d + 1) ||\alpha_\mathbf{b}||^2 + (d-1) ||\alpha_\mathbf{c}||^2 &= 0 \label{coeff 2}\\
        (d-1)^2 ||\alpha_\mathbf{a}||^2 + (2d^3 - 4d^2 - 3d - 1) ||\alpha_\mathbf{b}||^2 - (2d^2 - 3d + 1) ||\alpha_\mathbf{c}||^2 &= 0 \label{coeff 3}
    \end{align}
    Using \textsc{Mathematica}, we see that a non-trivial solution exists. The coefficients given in Eqs. (\ref{coeff 1}) to (\ref{coeff 3}) form a (normalized) solution to the system of quadratic forms.
\end{proof}

\begin{figure}[b]
	\centering
	\includegraphics[width=\columnwidth]{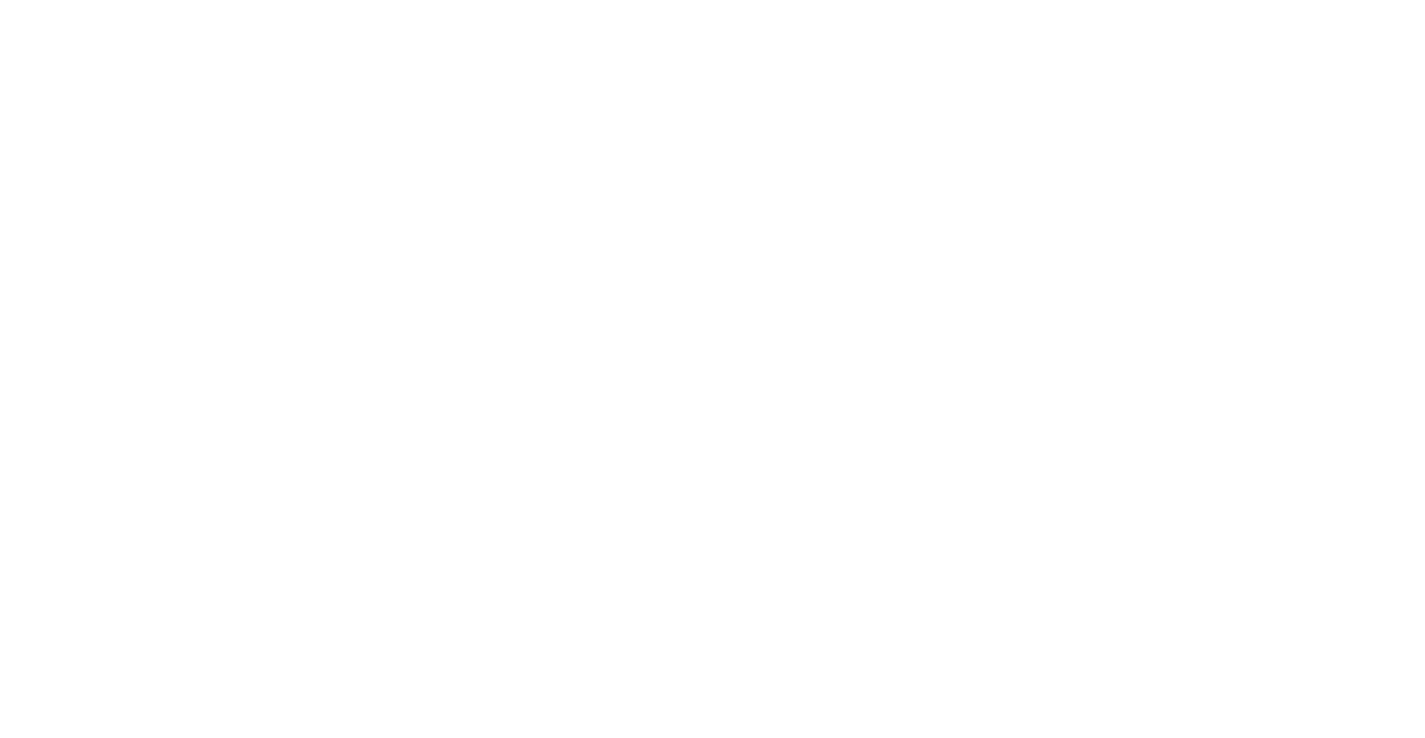}
\end{figure}

\bibliography{supplement}